\documentclass[aps, twocolumn, showpacs, superscriptaddress, prx, nofootinbib
]{revtex4-2}
\usepackage{enumitem}

\usepackage[utf8]{inputenc}
\usepackage[english]{babel}
\usepackage{diagbox}
\usepackage{makecell}
\usepackage{textcomp}
\usepackage{amsmath,etoolbox}
\usepackage{amsthm}
\usepackage{amssymb}
\usepackage[left=2.5cm,right=2.5cm,top=2.5cm,bottom=2.5cm,bindingoffset=0cm]{geometry}
\usepackage{dsfont}
\usepackage{bm}
\usepackage{leftidx}
\usepackage{float}
\makeatletter
\let\newfloat\newfloat@ltx
\makeatother
\usepackage[makeroom]{cancel}
\usepackage{indentfirst}
\usepackage{underoverlap}
\usepackage{hyperref}
\usepackage[caption=false]{subfig}
\usepackage[capitalise]{cleveref}
\crefname{section}{Sec.}{Secs.}
\crefname{figure}{Fig.}{Figs.}
\crefname{appendix}{App.}{Apps.}
\crefname{equation}{Eq.}{Eqs.}
\crefname{definition}{Definition}{Defs.}
\crefname{theorem}{Theorem}{Thms.}
\crefname{corollary}{Corollary}{Cors.}
\crefname{remark}{Remark}{Remarks}
\crefname{proposition}{Prop.}{Props.}
\usepackage{tikz}
\usetikzlibrary{
quantikz2,
decorations.pathmorphing, patterns,
cd}
\usepackage{ulem}
\usepackage{array}
\newcolumntype{L}[1]{>{\raggedright\let\newline\\\arraybackslash\hspace{0pt}}m{#1}}
\newcolumntype{C}[1]{>{\centering\let\newline\\\arraybackslash\hspace{0pt}}m{#1}}
\newcolumntype{R}[1]{>{\raggedleft\let\newline\\\arraybackslash\hspace{0pt}}m{#1}}
\usepackage{graphicx}
\usepackage{tikz-cd}
\usepackage[titletoc,title]{appendix}

\usepackage{adjustbox}
\usepackage{mathtools}
\usepackage{multirow}
\usepackage{placeins}
\usepackage{xargs}
\usepackage{xstring}
\usepackage{rotating}
\usepackage{calc}
\usepackage{upgreek}
\usepackage{algorithm}
\usepackage{algpseudocode}
\usepackage{tabularray}
\usepackage{pdftexcmds}
\makeatletter
\let\scshape\relax 
\DeclareRobustCommand\scshape{%
  \not@math@alphabet\scshape\relax
  \ifnum\pdf@strcmp{\f@family}{\familydefault}=\z@
    \fontfamily{qbk}%
  \fi
  \fontshape\scdefault\selectfont}
\makeatother
\usepackage{afterpage}

\DeclareMathAlphabet{\dutchcal}{U}{dutchcal}{m}{n}
\SetMathAlphabet{\dutchcal}{bold}{U}{dutchcal}{b}{n}
\DeclareMathAlphabet{\dutchbcal} {U}{dutchcal}{b}{n}
\usetikzlibrary{matrix,arrows,decorations.pathmorphing}
\newenvironment{sloppypar*}{\sloppy\ignorespaces}{\par}
\allowdisplaybreaks 

\newcommand {\sbra} [1] {\langle #1 |}
\newcommand {\sket} [1] {| #1 \rangle}

\newcommand{\me}{\mathrm{e}}
\newcommand {\unit} {\mathds{1}}

\DeclareMathOperator*{\argmax}{argmax}

\newcommand{\n}{\mathsf{n}}
\newcommand{\nq}{\mathsf{n}_{\mathsf{q}}}
\newcommand{\loc}{\mathsf{k}}
\newcommand{\scale}{\mathsf{m}}
\newcommand{\gloc}{\mathsf{D}}
\newcommand{\NLat}{{\mathsf{N_\Lambda}}}

\newcommand{\NH}{\mathsf{N_H}}
\newcommand{\Nell}{\mathsf{N_\ell}}

\newcommand{\NP}{\mathsf{N_P}}
\newcommand{\NPJ}{\NP^{(J)}}

\newcommand{\Set}{S}
\newcommand{\ind}{J}
\newcommand{\Lat}{\Lambda}
\newcommand{\cP}{\mathsf{c_P}}
\newcommand{\OO}{\mathcal{O}}
\newcommand{\Nind}{\mathsf{N}_{\mathsf{ind}}}

\newcommand{\dims}{\mathsf{d}}

\newcommand{\eps}{\varepsilon}
\newcommand{\nterms}{\Gamma}
\newcommand{\stages}{\Upsilon}

\DeclareMathOperator{\polylog}{polylog}

\newcommand{\Hgate}{\mathrm{H}}
\newcommand{\Xgate}{\mathrm{X}}
\newcommand{\Zgate}{\mathrm{Z}}
\newcommand{\Nphi}{N_{\phi}}

\usepackage{amsthm}

\newtheorem{theorem}{Theorem}
\newtheorem{lemma}[theorem]{Lemma}

\theoremstyle{definition}
\newtheorem{definition}[theorem]{Definition}
\newtheorem{remark}[theorem]{Remark}

\newcommand{\PREPARE}{{\textsc{prep} }}

\newcommand{\SELECT}{{\textsc{sel} }}

\makeatletter
\def\hlinewd#1{%
\noalign{\ifnum0=`}\fi\hrule \@height #1 \futurelet
\reserved@a\@xhline}
\makeatother

\newcommand{\hilb}{\dutchcal{h}}

\begin{document}
    
\title{Strategies for simulating time evolution of Hamiltonian lattice field theories}

\author{Siddharth Hariprakash$^*$}
\affiliation{Center for Theoretical Physics and Department of Physics, University of California, Berkeley, California 94720, USA}
\affiliation{Physics Division, Lawrence Berkeley National Laboratory, Berkeley, California 94720, USA}
\affiliation{National Energy Research Scientific Computing Center (NERSC), Lawrence Berkeley National Laboratory,
Berkeley, CA 94720, USA,
}

\author{Neel S. Modi$^*$}
\affiliation{Center for Theoretical Physics and Department of Physics, University of California, Berkeley, California 94720, USA}
\affiliation{Physics Division, Lawrence Berkeley National Laboratory, Berkeley, California 94720, USA}

\author{Michael Kreshchuk}
\affiliation{Physics Division, Lawrence Berkeley National Laboratory, Berkeley, California 94720, USA}

\author{Christopher F. Kane}
\affiliation{Department of Physics, University of Arizona, Tucson, Arizona 85719, USA}

\author{Christian W Bauer}
\affiliation{Physics Division, Lawrence Berkeley National Laboratory, Berkeley, California 94720, USA}

\begin{abstract}
Simulating the time evolution of quantum field theories given some Hamiltonian $H$ requires developing algorithms for implementing the  unitary operator $\me^{-iHt}$.
A variety of techniques exist that accomplish this task, with the most common technique used so far being Trotterization, which is a special case of the application of a product formula. 
However, other techniques exist that promise better asymptotic scaling in certain parameters of the theory being simulated, the most efficient of which are based on the concept of block encoding.

In this work we study the performance of such algorithms in simulating lattice field theories.
We derive and compare the asymptotic gate complexities of several commonly used simulation techniques in application to Hamiltonian Lattice Field Theories.
Using the scalar $\hat{\varphi}^4$ theory as a test, we also perform numerical studies and compare the gate costs required by Product Formulas and Signal Processing based techniques to simulate time evolution.
For the latter, we use the
the Linear Combination of Unitaries construction augmented with the Quantum Fourier Transform circuit to switch between the field and momentum eigenbases, which leads to immediate order-of-magnitude improvement in the cost of preparing the block encoding.

The paper also includes a pedagogical review of utilized techniques, in particular Product Formulas, LCU, Qubitization, QSP, as well as a technique we call HHKL based on its inventors' names.

\end{abstract}

\maketitle
\def\thefootnote{*}\footnotetext{These authors contributed equally to this work.}\def\thefootnote{\arabic{footnote}}

\tableofcontents


\section{Introduction\label{sec:introduction}}

Performing precise numerical simulations is an essential ingredient in the studies of Quantum Field Theories (QFTs), which describe the most fundamental interactions within the Standard Model of particle physics. 
Some of the difficulties arising in such studies include: a large, non-conserving number of degrees of freedom; different types of particles and their interactions; and non-perturbativity  of the strong interaction at large distances.
A common approach to numerical simulation of field theories is provided by Euclidean Lattice Field Theory, where space-time is discretized and observables are calculated using the imaginary-time path integral formalism. 
For a recent overview of the field, see~\cite{Davoudi:2022bnl}.
Such a conventional approach to lattice field theory relies on Monte-Carlo integration to perform the integral over the vast Hilbert space~\cite{Morningstar:2007zm}, which requires the weight of each field configuration to be real and positive. 
While this has been very successful, there are many observables, including those requiring dynamical information, for which Monte-Carlo integration is not possible due to the so-called sign problem~\cite{Loh:1990zz,Pan:2022fgf}.
The issue of the sign problem, however, can be avoided entirely by computing such observables within the framework of Hamiltonian Lattice Field Theories (HLFTs), in which space is discretized, while time is kept continuous~\cite{Kogut:1974ag,drell1979quantum,bardeen1976local,Bardeen:1979xx}.\footnote{\label{footnote:HLFT}Note that while HLFTs sometimes refer solely to quantum field theories in the equal-time formulation, discretized on a spatial lattice~\cite{Kogut:1974ag,drell1979quantum}, in this work we assume (unless explicitly stated otherwise) a more general scenario which applies to any type of spatial discretization (e.g., spatial/momentum lattice~\cite{Liu:2021otn} or any other single-particle basis set~\cite{Vary:2009gt}), any type of field discretization (first or second quantization~\cite{Klco:2018zqz,Liu:2021otn}), as well as to both equal-time and light-front quantization~\cite{bardeen1976local,Bardeen:1979xx,Pauli:1985pv,pauli1985discretized,Brodsky:1997de}.
}
While resources required for the simulation of HLFTs on classical computers scale exponentially with the number of lattice sites~\cite{Jordan:2017lea}, Jordan, Lee, and Preskill showed that quantum computers provide an exponential advantage over classical computers for this problem~\cite{Jordan:2012xnu}.
In fact, simulating QFT with the aid of quantum computers served as one of the motivations for the development of quantum computation~\cite{Feynman:1981tf}, and is still considered as one of its most promising applications. For a recent review of the applications of quantum computing to high-energy physics, see~\cite{Bauer:2022hpo}.

A major task in such a scenario is the ability to efficiently simulate on a quantum computer the time evolution of a given Hamiltonian.
In Ref.~\cite{Jordan:2011ci} Jordan, Lee, and Preskill  made a detailed proposal to simulate a scalar field theory on quantum computers.
Similarly to the earlier research~\cite{Lloyd:1996aai,zalka1998simulating,Wiesner:1996xg}, their work relies on the Suzuki-Trotter approximation~\cite{cohen1982eigenvalue,suzuki1976generalized,Suzuki:1991jtk,childs2019nearly,Childs:2019hts,Hatano:2005gh} to compute the exponential of the Hamiltonian, and hence time evolve the scalar field theory. 

This paper compares the asymptotic resource requirements of various techniques to implement the time evolution operator of quantum field theories, including Suzuki-Trotter expansions and nearly-optimal simulation strategies, such as Quantum Signal Processing (QSP)~\cite{Low:2016sck,Low:2016znh,Zeytinoglu:2022zda,Motlagh:2023oqc,Kikuchi:2023qbb} and the HHKL method named after the authors of Ref.~\cite{Haah:2018ekc}. For each method, we compute the asymptotic dependence with respect to various parameters characterizing the problem at hand, such as the error allowed in the approximation of the evolution operator, measures of Hamiltonian size / complexity, as well as the time for which the system is evolved. This comparison is carried out for two very general classes of lattice Hamiltonians, but we also provide more detailed results for a scalar quantum field theory as an explicit example.

Using nearly-optimal methods requires querying matrix elements of the Hamiltonian, typically in the form of the \emph{Block Encoding} (BE) subroutine.
In this work, we consider the BE for a bosonic degree of freedom, based on the Linear Combination of Unitaries (LCU) algorithm~\cite{Childs:2012gwh}, and achieve a significant reduction in the number of gates and ancillary qubits by taking advantage of the quantum Fourier transform for switching between the eigenbases of the field and conjugate momentum operators.
We do not consider BEs based on the sparse oracle access input model, for two reasons.
First, their construction is highly problem-specific~\cite{Toloui:2013xab,Babbush:2018ywg,Kreshchuk:2020dla,Kirby:2021ajp,rhodes2024exponential}. 
Second, for local lattice field theories, which we consider as one of the major applications of the present work, our studies~\cite{rhodes2024exponential} indicate that the LCU-based approach is more efficient than that based on the sparse oracle access.

\subsection{Paper Structure}

We start with setting up our notation in~\cref{sec:definitions} before summarizing the main results of this paper in~\cref{ssec:main_results}, with most of the information contained in~\cref{tab:methods,tab:site-local,tab:geometically-local}.
In~\cref{sec:review_alg} we present a review of the various simulation algorithms.
No new information is presented in this section, and the presentation is generic, not yet adapted for lattice field theory simulations.
In~\cref{sec:lattice_algorithms} we compute the asymptotic gate complexity of using these techniques to simulate the time evolution of 2 generic classes of lattice field theories. 
This section will provide derivations for the results presented in~\cref{ssec:main_results}. 
In~\cref{sec:quartic} we present concrete asymptotic gate complexities for a quartic scalar field theory.
We provide a discussion of our results in~\cref{sec:discussion} and conclude in~\cref{sec:conclusion}.

\subsection{Definitions and Notations used in this work}
\label{sec:definitions}
We begin by presenting the standard mathematical notation used to describe the asymptotic behavior of functions

\begin{definition}[Asymptotic Notations]
    \label{def:asymp_notation}
    {\it Let $f,g: \mathbb{R} \rightarrow \mathbb{R}$ be functions of real variables. Then we say that $f = \OO(g)$ if there exists a positive constant c such that $|f(x)| \leq c|g(x)|$ for all sufficiently large choices of $x$, i.e. $\forall x \geq x_0 \in \mathbb{R}$. We say that $f = \Omega(g)$ if $g = \OO(f)$ and that $f = \Theta(g)$ if $f = \OO(g)$ and $g = \OO(f)$.}      
\end{definition}

We see that $\OO(\cdot)$ and $\Omega(\cdot)$ denote asymptotic upper and lower bounds respectively, while $\Theta(\cdot)$ denotes an asymptotically \emph{tight} bound. We will make use of these notations throughout this work to denote the asymptotic scaling / behavior of the variables we present in this section. 

In HLFTs the degrees of freedom are typically the discretized quantum fields (in first quantization) or the eigenstates of the number operator (in second quantization) which live on the sites, links, or plaquettes of a lattice $\Lambda$.\footnote{In some cases, the lattice is naturally embedded into some metric space, e.g., when a lattice in real or momentum space is used.
More generally, it can be considered as the space which the index enumerating the degrees of freedom belongs to.
Furthermore, we use the term ``lattice site'' for any local ingredient of the lattice that can support a Hilbert space. While for fermions and scalar fields these are the actual lattice sites, for gauge fields this terminology also covers the links and plaquettes of the lattice.
}
For many Hamiltonians of interest, each term in the Hamiltonian couples fields belonging to a small subset of sites~\cite{Kogut:1974ag}.
Depending on which fields get coupled together, one finds different types of Hamiltonians. 

We first consider \emph{$\loc$-site local} Hamiltonians of the form
\begin{equation}
    \label{eq:local}
    \begin{gathered}
    H = \sum_{J\in \Set} H_{J}\,,
    \\
    J = \{j_1,\dots,j_\loc\}\in\Set
    \,,
    \quad j_i\in\Lat\,,
    \end{gathered}
\end{equation} 
where each $H_{J}$ acts nontrivially only on at most $\loc$ sites. The multi-index variable $J$, used to label each summand appearing in \cref{eq:local}, consists of a subset (containing $\loc$ elements) of the $\NLat$ total lattice sites.
Each lattice site is comprised of $\nq$ qubits, each term $H_J$ acts on $\loc \nq$ qubits, and the total system contains $\n = \NLat \nq$ qubits.\footnote{We note that our definition of $\loc$-site-local Hamiltonian reduces to the usual definition of a $\loc$-local Hamiltonian on $\n$-qubits~\cite{Childs:2019hts} if we set $\nq = 1$ and $\Lambda=\{1,2,\ldots,\n\}$, in which case $\NH \sim \n^\loc$ and $\NP \sim 4^\loc$.}
The cardinality of the set $\Set$, which contains sets of indices used to label individual terms $H_J$, is also the total number of individual terms (denoted by $\NH$) that form the full Hamiltonian $H$.
We assume that in general, for a $\loc$-site-local Hamiltonian, the following scaling relation holds:
\begin{equation}
\label{eq:setsize}
    \NH
    \equiv |\Set|
    = \mathcal{O}(\NLat^\scale)
    \,,
\end{equation}
where $\scale \left(\in \mathbb{R}_{\geq 0}\right) \leq \loc$ is a parameter fixed by the choice of a specific theory. 

We note that each summand $H_J$ can be expressed in the Pauli basis as
\begin{align}\label{eq:HJ}
    H_J = \sum_{i \in [\NP^{(J)}]}c_{i}^{(J)}P_{i}^{(J)}
    \,,
\end{align}
where each $P_{i}^{(J)} \in \mathbb{C}^{2^\n \times 2^\n}$ represents an $n$-qubit Pauli operator, each $c_{i}^{(J)} \in \mathbb{C}\setminus \{0\}$ is a complex coefficient, and $\NP^{(J)}$ enumerates the number of terms with non-zero coefficient in the Pauli decomposition of $H_J$. 
We introduce the quantity $\NP$ as the maximum number of Pauli strings appearing in the decomposition of any single term $H_J$. We define it as follows:
\begin{align}
    \label{def:NG}
    \NP \equiv \max_{J}\left( \NP^{(J)} \right)
    .
\end{align}
In this way, the total number of Pauli strings appearing in a Hamiltonian with $\NH$ summands can be upper bounded by $\NH \NP$. Note that for a $\loc$-site-local Hamiltonian $\NP = \OO(4^{\loc n_q})$ always holds as an upper bound.

An important quantity in asymptotic gate complexity estimates is the $1$-norm of a Hamiltonian $H$, denoted by $\|H\|_1$, which we define as
\begin{align}
    \label{def:1_norm}
    \|H\|_1 \equiv \sum_{J\in \Set} \|H_{J}\|
    \,,
\end{align}
where $\| \cdot \|$ represents the spectral / operator norm. 
The \emph{induced} $1$-norm of the Hamiltonian, denoted by $|||H|||_1$, can then be defined by maximizing the following sum over the qubit-index (i.e. the location within the multi-index variable $J \in \Set$) $\ell \in \{1,\dots,k\}$ and the site-index $j_{\ell} \in \Lambda$:
\begin{align}
    \label{def:induced_1_norm}
    \begin{gathered}
    |||H|||_1 \equiv \max_{\ell} \max_{j_\ell} \sum_{J_{\ell}^{(j_\ell)} \in S_{\ell}^{(j_\ell)}} \|H_J\|
    \,,
    \\
    J_{\ell}^{(j_\ell)} = \{j_1,\dots,j_{\ell-1},j_{\ell+1},\dots,j_\loc\} \in S_{\ell}^{(j_\ell)}
    \,,
    \\
    J = \{j_1, \dots, j_{\ell-1}, j_\ell, j_{\ell+1},\dots,j_\loc\}\in\Set
    \,.
    \quad j_i\in\Lat
    \,.
    \end{gathered}
    \raisetag{3\baselineskip}
\end{align}
In other words, the induced 1-norm performs the sum over all norms $\|H_J\|$ which have the lattice site $j_\ell$ at the $\ell$'th position, both of which are chosen to maximize the sum appearing in \cref{def:induced_1_norm}. 
This induced norm provides an upper bound on the \emph{strength} of the Hamiltonian at any single site due to interactions with the remaining (allowed by locality restrictions and the specific theory under consideration) sites in the lattice, such that the qubit-index of that site within the multi-index variable $J \in S$ is fixed. 
We denote by $\Nind$ the number of terms in the sum appearing in~\cref{def:induced_1_norm}:
 \begin{equation}
    \label{def:N_ind}
     \begin{gathered}
        \Nind \equiv | S_{\ell^{*}}^{(j_{\ell^{*}})} |
         \,,
         \\
         \quad \ell^{*},j_{\ell^{*}} = \argmax_{\ell}\argmax_{j_{\ell}} \sum_{J_{\ell}^{(j_\ell)} \in S_{\ell}^{(j_\ell)}} \|H_J\|\,.
     \end{gathered}
 \end{equation}

For a general $\loc$-site-local Hamiltonians,
we write
\begin{equation}
    \label{eq:Nindloc}
  \Nind = \OO\! \left( \NLat^{\scale_{\mathsf{ind}}} \right),
\end{equation}
where $\scale_{\mathsf{ind}}$ ($\in \mathbb{R}_{\geq 0}$) $\leq \loc - 1$ is again a parameter fixed by the choice of a specific theory.
We also define a quantity $\cP \in \mathbb{R}$ to be the maximum absolute value of the individual coefficients in the expansion of each $H_\ind$ in terms of Pauli operators. Writing each summand $H_J$ as in~\cref{eq:HJ}, one finds
\begin{align}\label{eq:cP}
    \cP \equiv \max_{J,i}\{|c_{i}^{(J)}|\}
    \,,
\end{align}
so that ${\|H\|_1\leq\cP\NP\NH}$ and ${|||H|||_1 \leq \cP \NP\Nind}$.

We define the \textit{exponentiation error} $\varepsilon$ as the difference between the exact time evolution operator $\me^{-iHt}$ and an approximate implementation $U(t)$
\begin{equation} 
\label{eq:eps}
\varepsilon \equiv \| U(t) - \me^{-iHt} \|
\,. 
\end{equation}

We summarize these results in a few definitions, that will be used throughout this paper
\begin{definition}[Site-local lattice Hamiltonian]
\label{def:loclattice}
{\it A $\loc$-site-local Hamiltonian on a lattice $\Lat$ of size $\NLat=|\Lambda|$ is one of the form:
\begin{equation}
    \begin{gathered}
    H = \sum_{J\in \Set} H_{J}\,,
    \\
    J = \{j_1,\dots,j_\loc\}\in\Set
    \,,
    \quad j_i\in\Lat\,,
    \nonumber
    \end{gathered}
\end{equation} 
where the number of terms in the set $S$ are given by $\NH$, each each $H_{J}$ acts nontrivially only on at most $\loc$ sites $j_i$, labeled by $J = \{j_1,\dots,j_\loc\}$, and each lattice site is comprised of $\nq$ qubits.
Let each term $H_J$ take the form~\eqref{eq:HJ} and, as in~\eqref{def:NG}, let $\NP$ be an upper bound 
    on the number of distinct Pauli strings appearing in the decomposition of any single term $H_J$. 
    Let $\cP$ be the largest magnitude of any Pauli coefficient over all terms as in~\eqref{eq:cP}, and $\Nind$ be the number of terms in the induced 1-norm of the Hamiltonian as in~\cref{def:N_ind}. }
\end{definition}

Many HLFT Hamiltonians have the property that they only couple fields at neighboring lattice locations together, giving rise to a \textit{geometric} locality. 
Defining a distance metric to the different lattice sites, allows to define a $\loc$-geometrically-site-local Hamiltonian, which is a sum of operators, each involving only sites belonging to a connected set.

\begin{definition}[Geometrically-site-local lattice Hamiltonian]
\label{def:geomloclattice}
{\it A $\loc$-geometrically-site-local lattice Hamiltonian is a $\loc$-site-local lattice of Hamiltonian in which each individual term $H_J$ acts on a connected set of sites $J$.\footnote{
One typically assumes that the lattice is embedded into some metric space. In our work, we assume a hypercubic lattice with spacing $1$.
In principle, one could define locality without introducing a metric structure, e.g., by using a more general topology on $\Set$.}
} 
\end{definition}

The restriction of geometric locality implies that each term $H_J$ acts upon belong to a ball of diameter $\loc$, which imposes a tighter bound on $\Nind$:\footnote{The precise upper bound in \cref{eq:Nindgeomloc} is given for a hypercubic lattice, but statement on independence of $\Nind$ on $\NLat$ is general.}
\begin{equation}
    \label{eq:Nindgeomloc}
  \Nind \leq \binom{\loc^\dims}{\loc} = \OO(1)\,,
\end{equation}
i.e., $\Nind$ no longer scales with $\NLat$ (note that here we have assumed a constant value for the spatial dimension $\dims$).
Furthermore, since each set $J$ can be labeled by the center of the corresponding ball, the number of sets and hence terms in the Hamiltonian is given by 
\begin{align}
\label{eq:NHgeomloc}
    \NH = \OO(\NLat)
    \,.
\end{align}

Hamiltonians of relativistic quantum field theories typically depend on the values of fields and their derivatives at a given point and thus one can expect that lattice formulations of these theories exist such that they are geometrically local. 
This is the case for the scalar field theory considered later in this paper for example.
For gauge theories, gauge invariance needs to be included into the considerations, and the locality of the Hamiltonian now depends on what basis is chosen for the fields. 
In so-called electric bases, Gauss's law constraints are still local, such that the final Hamiltonian is geometrically local in this basis. 
On the contrary, magnetic bases, behaving better at small bare coupling (which is required in the continuum limit), require non-geometrically local interactions~\cite{Haase:2020kaj,Ji:2020kjk,Bauer:2021gek,Bauer:2023jvw,Grabowska:2022uos,Kane:2022ejm}.

As we will also consider Hamiltonians that act only on $\OO(1)$ number of lattice sites, we will also provide the definition for them:
\begin{definition}
\label{def:O1loclattice}[$\OO(1)$ lattice Hamiltonian]
    {\it Let $H$ be a $(\loc=\NLat=\OO(1))$-site-local Hamiltonian, with $\NP$ and $\cP$ defined as in~\cref{def:NG,eq:cP}. Note that $\NLat=\OO(1)$ also implies $\NH = \OO(1)$, and that there is no difference between a site-local and geometrically site-local Hamiltonian in this case}.
\end{definition}

\begin{table}[t!]
    \begin{ruledtabular}
    \begin{tabular}{lr}
        $\Lat$ & Lattice
        \\ \hline
        $\NLat=|\Lambda|$ & Lattice size
        \\ \hline
        $\ind$ & \makecell[r]{Multi-index enumerating\\Hamiltonian terms} 
        \\ \hline
        $H=\sum_{\ind \in \Set} H_{\ind}$ & Hamiltonian 
        \\ \hline
        $\loc$ & Hamiltonian site-locality
        \\ \hline
        $H_\ind$ & Term acting on $\loc$ sites
        \\ \hline
        $\Set$ & Set containing sets of indices $\ind$
        \\ \hline
        $\NH= |\Set|$ & Number of summands forming H
        \\ \hline        
        $\nq$ & Number of qubits per lattice site
        \\ \hline        
        $\n=\NLat\nq$ & Total number of qubits
        \\ \hline
        $\NP^{(\ind)}$ & \makecell[r]{Number of Pauli terms representing \\ a single term $H_\ind$}
        \\ \hline        
        $\NP$ & \makecell[r]{$\max_{J}( \NP^{(J)} )$} 
        \\ \hline
        $\cP$ & Maximum Pauli coefficient in $H$
        \\ \hline
        $\gloc$ & Hamiltonian geometric locality
        \\ \hline
        $\dims$ & Spatial dimension
        \\ \hline
        $d$ & Degree of Chebyshev polynomial
        \\ \hline        
        $\eps$ & Exponentiation error
        \\ \hline
        $\scale$ & \makecell[r]{Parameter describing the scaling of $\NH$ \\ with respect to $\NLat$, see~\cref{eq:setsize}
        }
        \\ \hline
        $\scale_{\mathsf{ind}}$ & \makecell[r]{Parameter describing the scaling of $\Nind$ \\ with respect to $\NLat$, see~\cref{eq:Nindloc}
        }
        \\ \hline
        $\chi$ & Total number of elementary gates
    \end{tabular}
    \end{ruledtabular}
    \label{tab:notations}
    \caption{Notations used throughout the paper.}
\end{table}

\subsection{Main results}
\label{ssec:main_results}
The main results of this paper are summarized in \cref{tab:methods,tab:site-local,tab:geometically-local}. 
In this section we summarize only the results, while the detailed derivations of these expressions will be given in~\cref{sec:lattice_algorithms}.

We first consider the results for a lattice containing ${\cal O}(1)$ sites, and compare the implementation of the time evolution operator after a \emph{Product Formula} (PF) based splitting to the implementation using Quantum Signal processing (QSP).
Both of those techniques will be introduced and defined in~\cref{sec:review_alg}. 
The important parameters that the time evolution operator depends on are the time $t$ the system is evolved for, the exponentiation error $\varepsilon$, as well as the maximum number of Pauli terms $\NP$ and the maximum coefficient $\cP$ appearing in the Pauli decomposition of any single term in the Hamiltonian.
The scaling with respect to these parameters is shown in~\cref{tab:methods}, where $p$ denotes the order of the PF used. 
While the scaling of the PF implementation on all parameters improves as $p$ is increased, the number of gates needed is typically exponential in $p$, so the order of the PF is typically chosen as $p = {\cal O}(1)$. 
One can therefore see that QSP has a better scaling than PFs in all parameters, with the most dramatic improvement being in the scaling in $\varepsilon$, which is exponentially better. 

Next, we include the asymptotic scaling with the size of the lattice $\NLat$ for a $\loc$-site local Hamiltonian.
The Hamiltonian is a sum of terms $H_J$, each of which acts on $\loc = {\cal O}(1)$ lattice sites and in general do not commute with one another. 
One approach is to assemble the terms $H_J$ using a PF, with the implementation of each $H_J$ performed either using a PF or QSP, as just discussed. 
These two choices are compared in the first column of~\cref{tab:site-local}. 
One can see that the scaling is more favorable in all parameters when the individual terms $H_J$ are themselves implemented using a PF.

A second approach is to use QSP for the entire Hamiltonian, with the corresponding scaling shown in the second column of~\cref{tab:site-local}.
One can see that this leads to worse scaling (in general) with the number of lattice sites than the PF, while giving a much better scaling with $\varepsilon$. 
Therefore, the usage of PF may be preferential in situations when $\NH \gg 1 / \varepsilon$, which is often the case in field theory.
\renewcommand{\arraystretch}{2.3}
\begin{table}[h!]
    \centering
    \begin{tabular}{c|c}
        \makecell{Local\\implementation} & \makecell{Complexity\\for $\mathcal{O}(1)$ sites}
        \\
        \hline
        Product formula & $\OO\!\left[  \left( \NP \right)^{2 + \frac{1}{p}} \left( \cP t \right)^{1 + \frac{1}{p}}\varepsilon^{-\frac{1}{p}} \right]$
        \\
        \hline
        QSP& $\OO\!\left[\NP \log \NP\left(\NP \cP t + \log(1/\eps)\right)\right]$
    \end{tabular}
    \caption{
    Scaling of a Hamiltonian defined on ${\cal O}(1)$ lattice sites. The scaling is given in terms of the number of Pauli strings $\NP$ of a single term in the Hamiltonian~\cref{def:NG}, the largest coefficient of each of these Pauli strings $\cP$~\cref{eq:cP}, and the exponentiation error $\varepsilon$, defined in~\cref{eq:eps}.
    }
    \label{tab:methods}
\end{table}
\renewcommand{\arraystretch}{1}

Finally, we consider the case of geometrically-local Hamiltonians. 
In this case one still has the same two possibilities for the implementation, shown in the first two columns of~\cref{tab:geometically-local}.
Both have an improved scaling in the number of lattice sites $\NLat$ due to the geometric locality, and as before the QSP approach wins in terms of dependence on $\eps$ while the PF-based wins in terms of scaling with $\NLat$.
The geometric locality of the Hamiltonian, however, allows for the use of the HHKL algorithm, introduced in~\cref{sec:hhkl}, to assemble the $H_J$ terms, with each local term again implemented via either PF or QSP. 
The scaling of this approach is shown in the third column of~\cref{tab:geometically-local}. 
We can see that this leads to a better scaling in $\NLat$, with a combination of HHKL with QSP leading to the best overall scaling.
The scaling of course does not give any information on the prefactor, that could be considerably larger for HHKL than for product based approaches~\cite{childs2019nearly}. Furthermore, the scaling given here is not tight, so it could be that for realistic scenarios product based formulas still outperform HHKL based methods.
Note that the scaling of PFs is identical to that of HHKL combined with PFs.

As a case study, we consider a scalar lattice field theory and express the obtained asymptotic gate complexities in terms of model parameters.
This requires specifying the BE used, and we consider two methods for constructing the BE subroutine $U_H$ in the scalar lattice field theory.
The first method uses the standard Linear Combination of Unitaries (LCU) algorithm~\cite{Childs:2012gwh} to prepare a BE of the Hamiltonian $H$ of the entire system; this approach is efficient for arbitrary local Hamiltonians.
The second method improves upon the first by using the Fourier Transform to switch between the eigenbases of position and momentum operators, which exponentially improves the asymptotic cost of constructing each $U_{H_J}$ in terms of the number of qubits $\nq$ per site. The BE of the full Hamiltonian is then constructed using LCU to combine each $U_{H_J}$.
The comparisons of costs for these approaches is shown in~\cref{fig:BEs}.
Finally, we use the best available construction of BE in order to compare the gate counts between the PF and QSP approaches in application to a single site of the $\varphi^4$ theory, see~\cref{fig:cnot_3d}.

\begin{table*}[t]
    \centering
    \begin{tabular}{c!{\vrule width 2\arrayrulewidth}cc|cc}
        {\makecell{Local implementation}} &
        &{\makecell{Assembly using PF}} &
        &{\makecell{QSP for entire system}}
        \\ \hline
        Product formula &a)&
        $\vphantom{\biggl(}\OO\!\left( \Nind \NP \left( \NH\NP \cP t \right)^{1 + \frac{1}{p}}\varepsilon^{-\frac{1}{p}} \right)$ &&
        \\ \hline
        QSP &b)&
        $\vphantom{\biggl(}
        \begin{aligned}
    \vphantom{\biggl(}\OO&\Bigl[\Nind \left(\NH \NP \cP t \right)^{1 + \frac{1}{p}}\varepsilon^{-\frac{1}{p}}\NP 
        \\&\times
        \log \NP\log(\Nind \NH\NP \cP t/\eps)\Bigr]    
        \end{aligned}
        $
        &c)&
        $\OO\!\left[\NH \NP \log (\NH \NP)\left(\NH \NP \cP t + \log(1/\eps)\right)\right]$
    \end{tabular}
    \caption{Resource scaling for site-local Hamiltonians as in~\cref{def:loclattice}.
    The two rows label the local implementation for each term $\me^{-i H_J t}$ either  using a PF (top row) or using QSP (bottom row). The two columns show how the various terms $H_J$ are combined together, either using a PF (left column) or doing the whole evolution using QSP (bottom right).
    Cases a), b), and c) are studied in \cref{sssec:pfsl}, \cref{ssec:pfqsp}, and \cref{sssec:qspsl}, correspondingly.
    }
    \label{tab:site-local}
\end{table*}
\begin{table*}[t]
    \centering
    \begin{tabular}{c!{\vrule width 2\arrayrulewidth}cc|cc|cc}
        {\makecell{Local implementation}} &&
        {\makecell{Assembly using PF}} &&
        {\makecell{QSP for entire system}} &&
        {\makecell{Assembly using HHKL}}
        \\ \hline
        Product formula &d)&
        $\mathcal{O}\!\left( \NP \left( \NLat \NP \cP t \right)^{1 + \frac{1}{p}}\varepsilon^{-\frac{1}{p}} \right)$ &&

        &e)&
        $\begin{alignedat}{9}
        \vphantom{\biggl(}
        \mathcal O&\Bigl(\NP \left(\NLat \NP \cP t \right)^{1 + \frac{1}{p}}\varepsilon^{-\frac{1}{p}} \Bigr)
        \end{alignedat}$
        \\ \hline
        QSP &f)&
        $\begin{alignedat}{9}
        \mathcal{O}&\Bigl[ \left(\NLat \NP \cP t \right)^{1 + \frac{1}{p}}\varepsilon^{-\frac{1}{p}}\NP 
        \\&\times
        \log \NP\log( \NLat\NP \cP t/\eps)\Bigr]    
        \end{alignedat}$
        &g)&
        $\begin{alignedat}{9}
        \mathcal O&\Bigl[\NLat \NP \log (\NLat \NP)
        \\&\times
        \left(\NLat \NP \cP t + \log(1/\eps)\right)\Bigr]    
        \end{alignedat}$
        &h)&
        $\begin{alignedat}{9}
        \vphantom{\biggl(}
        \mathcal{O}\Bigl[\NLat& \NP^2 \cP t \left[\log(\NLat t/\eps)\right]^{\dims}
        \\&\times
         \log (\NP\log(\NLat t/\eps))\Bigr]
        \end{alignedat}$
    \end{tabular}
    \caption{Resource scaling for geometrically site-local Hamiltonians as in~\cref{def:loclattice}.
    The two rows label the local implementation for each term $\me^{-i H_J t}$ either  using a PF (top row) or using QSP (bottom row). The three columns show how the various terms $H_J$ are combined together, either using a PF (left column), using the HHKL algorithm (right column), or doing the whole evolution using QSP (bottom middle).
    Cases d), f), and g) are obtained from cases a), b), and c), respectively, in \cref{tab:site-local} upon setting $\NH=\OO(\NLat)$ and $\Nind=\OO(1)$.
    Cases e) and h) are studied in~\cref{ssec:hhkl}. 
    }
    \label{tab:geometically-local}
\end{table*}

\section{Review of Simulation Algorithms}
\label{sec:review_alg}

In this section, we review some algorithms commonly used in the literature for simulating the time evolution of generic quantum systems.
We will focus on algorithms based on Product Formulas (PFs)~\cite{childs2019nearly,Childs:2019hts} and Quantum Signal Processing (QSP)~\cite{Low:2016sck,Low:2016znh}, both of which can be applied to small and spatially extended systems.
We also review the HHKL algorithm~\cite{Haah:2018ekc}, which may serve as a ``nearly-optimal wrapper'' of local evolution operators in geometrically-local system. These algorithms will serve as the building blocks from which we express algorithms for Hamiltonian lattice field theories in later chapters.

The discussion in this section aims to be a standalone discussion that gives the reader a basic understanding of the different tools used to implement a time evolution operator.
While we aim to make this section accessible to readers without much background in the topic, we will still provide a precise set of statements and theorems that will be used later.
Note that this section is not detailed enough to re-derive all of these results, and for this we refer the reader to the papers cited in this section.

Note also that this discussion is kept very general, without restricting ourselves to the Hamiltonians considered in this work (except when stated explicitly otherwise).

\subsection{Time evolution by Product Formulas}
\label{PF}

We begin by considering a time-independent Hamiltonian 
\begin{align}
    \label{eq:time_ind_H}
    H = \sum_{\gamma=1}^{\Gamma}H_\gamma
\end{align}
comprised of $\Gamma$ summands, such that one can efficiently construct quantum circuits to implement the exponential $\me^{-iH_\gamma t }$ of each summand appearing on the RHS of \cref{eq:time_ind_H}. 
The unitary operator (called the time evolution operator for $H$) given by the exponential $\me^{-iHt}$ describes evolution under the full Hamiltonian $H$ for some arbitrary time $t$. 
PFs allow us to construct an approximation to the full time evolution operator by decomposing it into a product of exponentials, each involving a single summand. 
The general form of a PF (see Ref.~\cite{Childs:2019hts}) is given by 
\begin{equation}
    \label{def:PF}
    \me^{-iHt} \approx
    S(t) = \prod_{\upsilon=1}^{\Upsilon}\mathcal{P}_\upsilon\!\left(\prod_{\gamma = 1}^{\Gamma}\me^{-ita_\upsilon^{(\gamma)}H_{\gamma}}\right)
    ,
\end{equation}
such that each $a_\upsilon^{(\gamma)} \in \mathbb{R}$ and $\sum_{\upsilon=1}^{\Upsilon}a_\upsilon^{(\gamma)} = 1$ $\forall \gamma \in [\Gamma] \equiv \{1,2,\dots,\Gamma\}$. 
The latter condition ensures that each individual term $H_\gamma$ is evolved for the total time $t$. 
The parameter $\Upsilon$ determines the total number of stages in the PF, and for each stage $\upsilon \in [\Upsilon]$ there exists a corresponding permutation operator $\mathcal{P}_\upsilon$ that determines the ordering of the individual exponentials within that stage (the specific actions of the operators $\mathcal{P}_\upsilon$ depend on the chosen PF construction). 
PFs of the form \cref{def:PF} provide good approximations to the full time evolution operator when the evolution time t is small, and in particular one defines a $p^{\text{th}}$ order PF $S_p(t)$ as one that is correct up to $p^{\text{th}}$ order in $t$
\begin{align}
\label{eq:pth_order_PF_intro}
    S_p(t) = \me^{-iHt} + \mathcal{O}(t^{p+1})\,.
\end{align}
The simplest example of a PF arises from the use of the Baker-Campbell-Hausdorff formula 
\begin{align}
    \me^{A+B} = \me^A \me^B + \OO([A,B])
    \,,
\end{align}
and thus we can write
\begin{align}
    \me^{-iHt} \approx S_1(t) =  \prod_{\gamma=1}^{\Gamma} \me^{-iH_\gamma t }
    \,.
\end{align}
This approximation is known as the first order Lie-Trotter formula. The higher (even) order Suzuki-Trotter formulas are defined recursively as
\begin{align}
S_2(t) &\coloneqq \prod_{i=1}^{\nterms}\me^{\frac{-iH_it}{2}}\prod_{i=\nterms}^{1}\me^{\frac{-iH_it}{2}}\,,
    \nonumber\\
    S_{2k}(t) &\coloneqq S_{2k-2}(u_kt)^2S_{2k-2}((1-4u_k)t)\nonumber\\
    & \qquad \qquad \times S_{2k-2}(u_kt)^2\,,
\label{eq:pth_order}
\end{align}
where $u_k \equiv 1/\left(4 - 4^{1/(2k-1)}\right)$,  and the product given by $\prod_{i=\nterms}^{1}\left( \cdot \right)$ corresponds to reversing the order in which we choose the variable $i$ with respect to the first product given by $\prod_{i=1}^{\nterms}\left( \cdot \right)$. 
In theory, one could use such a construction to build arbitrarily large order PFs. 
However the number of stages $\Upsilon$, and hence the number of elementary gate operations required to construct a circuit implementing such a PF, grows exponentially with the order $p$ (see ~\cite{Childs:2019hts}). 
For this reason, we use relatively low order formulas for practical applications and our results assume that $p,\Upsilon(p) = \OO(1)$.
Note that~\cref{eq:pth_order} is only one possible form of~\cref{def:PF}, and investigating alternatives to it is an active area of research~\cite{Hastings:2014wyq,Poulin:2014yec,tranter2018comparison,Hu:2022nyd,tranter2019ordering,Schmitz:2021ruv,Ostmeyer:2022lxs}.

Since PFs have corrections that depend polynomially on the evolution time $t$, they give good approximations for small $t$. Given a PF valid for small $t$, however, we can obtain one for longer times.
The general procedure is to divide the full evolution time $t$ into $r$ (known as the \emph{Trotter number}) smaller time steps, and apply the chosen $p^{\text{th}}$ order PF to each step, describing the evolution for a smaller time $t/r$, before multiplying together the results from the individual steps
\begin{align}
    S_p(t) = (S_p(t/r))^r,
\end{align}
and the associated exponentation error $\varepsilon$ is given by
\begin{equation}
\label{eq:exp_error}
    \varepsilon = \|(S_p(t/r))^r - \me^{-iHt }\|
    \,.
\end{equation}
To work out the resource requirements to satisfy such an error bound, we ask for the required Trotter number $r$ for a $p^{\text{th}}$ order formula to satisfy \cref{eq:exp_error} given a fixed value of $\varepsilon$. 
In general the answer depends on the asymptotic scaling on the evolution time $t$, as well as the norms of the individual summands $H_\gamma$ and the commutators amongst them. 

While \cref{eq:pth_order_PF_intro} provides the correct asymptotic dependence of the error on $t$,~\cite{Childs:2019hts} and~\cite{Berry:2005yrf} present bounds that show explicit dependence on the 1-norm of the full Hamiltonian (given by $\sum_{\gamma=1}^{\nterms}\|H_\gamma\|$). 
Note however, that these bounds are not tight. 
For example, when all the summands $H_\gamma$ commute, the exponentiation error $\varepsilon$ becomes zero while these error bounds can become arbitrarily large. 
This was further addressed in \cite{Childs:2019hts}, where the authors present another (tighter) bound that results from exploiting the commutativity properties of the summands $H_\gamma$. 
This led to the following asymptotic scaling of the Trotter number
\begin{equation}\label{eq:trotter_number_1}
            r = \OO\!\left( \frac{\Tilde{\alpha}^{\frac{1}{p}}t^{1 + \frac{1}{p}}}{\varepsilon^{\frac{1}{p}}} \right)
            \,,
        \end{equation}
where $\Tilde{\alpha} \equiv \sum_{\gamma_1,\gamma_2,...,\gamma_{p+1} = 1}^{\Gamma} \lvert\lvert [ H_{\gamma_{p+1}},... [H_{\gamma_2},H_{\gamma_1}] ] \rvert\rvert$.

Note that not every Trotter number $r$ satisfying the asymptotic upper bound \eqref{eq:trotter_number_1} will yield the desired error bound. Rather, the above result merely provides an asymptotic upper bound on the \textit{minimum} Trotter number $r$ needed.

\subsection{Time evolution by Qubitization}
\label{QSP}

While most applications to date simulating relativistic quantum field theories use PFs as discussed in~\cref{PF}, there is an alternate approach for simulating the time evolution of quantum systems that is based on a very general technique called \emph{Quantum Signal Processing} (QSP)~\cite{Low:2016sck}.
As described in more detail in Sec.~\cref{sec:qsptimeevol}, QSP can be used to apply broad classes of matrix functions of the Hamiltonian to a quantum state.

We begin by clarifying the terminology as confusion often arises in the literature.
QSP is an algorithm which, given access to a circuit implementing a unitary operator $W$ with eigenvalues $\{\me^{i\theta_\lambda}\}$, constructs a circuit implementing a unitary $V$ with eigenvalues $\{\me^{ih(\theta_\lambda)}\}$, for a wide class of polynomials $h$.
In the context of quantum simulation the unitary $W$ is taken to be the so-called \emph{Szegedy quantum walk operator} $W_H$ associated with the Hamiltonian $H$.
The quantum walk operator $W_H$ is a particular type of Hamiltonian \emph{block encoding} operator, which can be constructed from more general block encodings $U_H$ by means of the \emph{Qubitization}~\cite{Low:2016znh} procedure.

Confusingly, the term \emph{Qubitization} is also often used for the entire process of simulating time evolution based on the usage of QSP and calls to $W_H$.
Moreover, sometimes \emph{Qubitization} serves as an umbrella term for all the simulation techniques which are not based on product formul\ae and/or rely on some form of block encoding. On top of that, \emph{Qubitization} has also been used in a completely unrelated context, for denoting the process of mapping physical degrees of freedom onto qubits~\cite{Huang:2021pwq,Alexandru:2022son}.
In this work, \emph{Qubitization} will be used in the first, restricted sense, referring to the process of constructing $W_H$.

There are two main motivations for studying QSP-based quantum simulation algorithms. First, they have a better asymptotic complexity than PFs~\cite{Low:2016sck,Low:2016znh}. Second, they are compatible with a wider class of Hamiltonians than those which can be efficiently mapped onto Pauli Hamiltonians, and for which the number of qubits is sublinear in the number of degrees of freedom in the system~\cite{Babbush:2019yrx,Kreshchuk:2020dla,Kirby:2021ajp}.

As with PFs, the aim of this section is a pedagogical overview of techniques based on QSP, however we do provide precise statements of the theorems needed in the rest of the paper. 
We will rely on results from the literature to provide proofs of these theorems. 

Simulating the time evolution operator of a given Hamiltonian using QSP relies on a a few key insights, which we will now review in turn.

\subsubsection{Idea of block encoding}
\label{sec:BEidea}

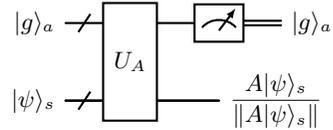
\begin{figure}[h]
\begin{quantikz}
\lstick{$\sket{g}_a$}& \gate[2]{U_A} \qwbundle{}& \meter{} & \setwiretype{c} \rstick{$\sket{g}_a$}\\
\lstick{$\sket{\psi}_s$} & \qwbundle{} \qw & \qw \rstick{$\dfrac{A\sket{\psi}_s}{\|A\sket{\psi}_s\|}$}
\end{quantikz}
\caption{Circuit implementing the block encoding $U_A$ of an operator $A$, as defined in~\cref{eq:BEdef}. If the ancillary register is measured in the state $\sket{g}_a$, the state of the system is ${A\sket{\psi}_s}/{\|A\sket{\psi}_s\|}$.}
\label{fig:BE}
\end{figure}

Block encoding (BE) is a technique that allows one to implement the action of a given operator $A$ on an arbitrary quantum state directly in a quantum circuit, despite this operator not being unitary.
This can be accomplished by rescaling $A$ and padding it with extra entries to fit it into a unitary matrix that can act simultaneously on the given quantum state together with an appropriate number of ancilla qubits.
Schematically, this means that the operator $A$, which acts on a system Hilbert space ${\cal H}_s$, is embedded into a unitary matrix acting on a larger Hilbert space ${\cal H}_a \otimes {\cal H}_s$, where ${\cal H}_a$ is the Hilbert space of the ancillary qubits. Schematically, $U_A$ is given by
\begin{align}
\label{eq:BEdef}
    U_A = 
    \begin{pmatrix}
        A / \alpha & * \\
        * & *
    \end{pmatrix}
    .
\end{align}
The rescaling by the parameter $\alpha>0$, known as \emph{scale factor}, is required, since the norm of any subblock of a unitary matrix has to be less than unity. 
Equation~\eqref{eq:BEdef} therefore implies that $U_A$ is, in fact, a BE for the entire equivalence class of matrices $A$ differing by a constant factor.

If we define the state $\sket{g}_a=G\sket{0}_a$ as the state that projects out the block of $U_A$ containing the operator $A/\alpha$, we can write 
\begin{align}
    \label{eq:beproj}
    A/\alpha &= (\sbra{g}_a\otimes \unit_s) U_A (\sket{g}_a\otimes \unit_s)\,, \\
    U_A \sket{g}_a\sket{\psi}_s &= 
        \sket{g}_a (A/\alpha \sket{\psi}_s) + \sket{\bot}_{as}
        \,,
\end{align}    
where $\sket{\bot}_{as}$ is a vector perpendicular to $\sket{g}_a$. 
This implies that $A$ is only applied to $\sket{\psi}_s$ iff the ancillary register is measured in $\sket{g}_a$. 
We shall term the probability of this to happen the ``success probability of the block encoding'', which is given by ${||A \ket{\psi}_s||/\alpha}$.
In the commonly considered case $\sket{g}_a=\sket{0}_a$ one can schematically write
\begin{align}
    \label{eq:beprojzero}
    \sket{0}_a\sket{\psi}_s &=
    \begin{pmatrix}
        \sket{\psi}_s \\ 0
    \end{pmatrix},\\
    U_A \sket{0}_a\sket{\psi}_s &= 
        \begin{pmatrix}
        (A/\alpha) \sket{\psi}_s \\ \ast 
    \end{pmatrix}.
\end{align}    

The success probability depends on how exactly $U_A$ is constructed.
In order for the algorithm to be considered efficient, the success probability has to scale inverse polynomially with problem parameters.
The scale factor plays its role in defining the complexity of the simulation algorithm, as will be discussed later.

In much of this work we will use this technique to block encode a Hamiltonian $H$, and we call the corresponding BE $U_H$.
The scale factor can then be absorbed by evolving a rescaled Hamiltonian
\begin{align}
    \tilde{H} = H / \alpha,
\end{align}
for the time
\begin{equation}
\tilde{t} = \alpha \, t\,,    
\end{equation}
so that
\begin{equation}
    Ht = \tilde{H} \tilde{t}\,.
\end{equation}

In the following discussion, we will choose a diagonal basis for the (rescaled) system Hamiltonian, such that basis vectors of the system Hilbert space are eigenstates of the Hamiltonian
\begin{align}
    \tilde{H} \sket{\lambda}_s = \lambda \sket{\lambda}_s 
    \,,
\end{align}
with $|\lambda| < 1$.
This basis is used only to simplify the discussion of the technique.
Being able to diagonalize the Hamiltonian is not required for its use.

As already discussed, the vector $\sket{g}_a$ of the ancillary Hilbert space ${\cal H}_a$ defines the block in which the Hamiltonian is encoded in the unitary matrix $U_H$. 
In other words, the BE acts on the tensor product of a state $\sket{\lambda}_s$ and $\sket{g}_a$ as
\begin{align}
\label{eq:UHaction}
    U_H \sket{g}_a\sket{\lambda}_s = \lambda\, \sket{g}_a\sket{\lambda}_s + \sqrt{1-\lambda^2} \sket{\bot^\lambda}_{as}
    \,,
\end{align}
where we have defined a normalized transverse state that is orthogonal to the vector $\sket{g}_a$, such that $(\sbra{g}_a\otimes \unit_s)\sket{\bot}_{as}=0$.
Thus, the action of $\tilde{H}$ can be extracted by acting with $U_H$ on an enlarged Hilbert space and post-selecting on the ancilla qubit being in the $\sket{g}_a$ state, \textit{i.e.}
\begin{align}
    \sbra{g}_a U_H \sket{g}_a\sket{\lambda}_s = \lambda\, \sket{\lambda}_s
    \,.
\end{align}
Note that if the ancillary Hilbert space ${\cal H}_a$ is 2-dimensional, one could choose $\sket{g}_a = \sket{0}_a$ and $\sket{\bot}_a = \sket{1}_a$. 

Many different techniques for BE have been discussed in the literature. 
Among many others, one can use general algorithms such as the FABLE algorithm~\cite{camps2022fable}, approaches that rely on a technique called linear combination of unitaries (LCU)~\cite{Childs:2012gwh,Lin:2022vrd}, a technique called QETU~\cite{Dong:2022mmq}, and techniques using access to the matrix elements of the sparse matrix and compact mappings $H$~\cite{Low:2016znh,Lin:2022vrd,Toloui:2013xab,Babbush:2017oum,Kirby:2021ajp,Kreshchuk:2020dla}.
The choice of BE is typically dictated by the number of ancillary qubits required, the value of the scale factor $\alpha$, the gate complexity required for its implementation, whether an approximate implementation of the Hamiltonian is sufficient, and whether it is applicable to general sparse Hamiltonians.

One issue with BEs discussed so far is that the action of the BE does not stay in the Hilbert space spanned by ${\rm span}\{\sket{g}_a \sket{\lambda}_s, \sket{\bot^\lambda}_{as}\}$. 
In other words, the action of $U_H$ on $\sket{\bot^\lambda}_{as}$ can produce a state with different value of $\lambda' \neq \lambda$, such that in general
\begin{align}
\begin{aligned}
    \sbra{g}_a\sbra{\lambda'}_s U_H \sket{g}_a\sket{\bot^\lambda}_{as} &\neq 0\\
    \sbra{\bot^{\lambda'}}_{as} U_H \sket{\bot^\lambda}_{as} &\neq 0
    \,.
\end{aligned}
\end{align}
Therefore, repeated actions of the unitary $U_H$ are not easily known.
This can be dealt with by constructing a more constrained BE $W_H$, for which one can ensure that multiple actions of this BE on a given eigenstate will never produce a state with a different eigenvalue. 
This constrained BE is called a \emph{qubitized block encoding} and is discussed in~\Cref{sssec:qubitization}.

\subsubsection{Block encoding by LCU}
\label{sec:BEtech}
\begin{figure*}[t]
\begin{quantikz}
\lstick[4]{$\sket{0^{\otimes k}}_a$}
& \gate[4]{\PREPARE} & \octrl{1} \gategroup[5,steps=4,style={dashed,rounded corners, inner
sep=6pt}]{$\SELECT$} & \ctrl{1} & \qw \ \cdots \ & \ctrl{1} & \gate[4]{\PREPARE^\dagger} & \meter{} & \setwiretype{c} \rstick{$0$} \\
& & \octrl{1} & \octrl{1} & \ \cdots \ & \ctrl{1} && \meter{} & \setwiretype{c} \rstick{$0$}\\
\vdots \setwiretype{n}& & \ \vdots \  & \vdots & \ \ddots \  \setwiretype{n} & \ \vdots \ && \ \vdots \ & \ \vdots \ \\
&  & \octrl{1} \wire[u][1]{q} & \octrl{1} \wire[u][1]{q} & \ \cdots \ & \ctrl{1} \wire[u][1]{q} && \meter{} & \setwiretype{c} \rstick{$0$}\\
\lstick{$\sket{\psi}_s$} & \qw \qwbundle{} & \gate{U_0} & \gate{U_1} & \ \cdots \ & 
\gate{U_{M-1}} &&& \rstick{$T\sket{\psi}_s/\|T\sket{\psi}_s\|$}
\end{quantikz}
\caption{Linear Combination of Unitaries (LCU) implements the BE $U_T$ of an operator $T$ expressed as a sum of elementary operators with known circuit implementations, reproduced from~\cite{Childs:2012gwh}.
Each operator $U_i$ is controlled on the binary string representing its index $i$.
The size of the ancillary register is $k=\lceil\log_2M\rceil$.
For illustration purposes, $\log_2 (M)$ is assumed to be an integer $k$ such that the binary string of $M-1$ contains all 1's.
}
\label{fig:LCU}
\end{figure*}
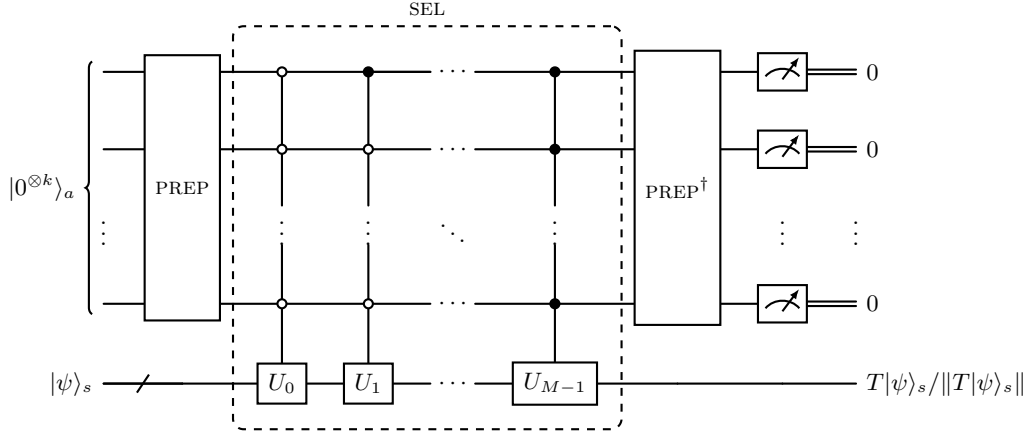
A generic subroutine, called Linear Combination of Unitaries (LCU), can be used as a general-purpose technique for constructing BEs of qubit Hamiltonians.
After introducing the general technique, we briefly discuss several modifications of LCU as well as alternative approaches to constructing BEs.

Consider the situation were one has a combination of $M$ unitary operators $U_0$, \ldots, $U_{M-1}$, each acting on a system Hilbert space ${\cal H}_s$ and the goal is to construct the operator
\begin{align}
    \label{eq:lcuT}
    T = \sum_{i=0}^{M-1} \beta_i \, U_i
    \,,
\end{align}
acting on the same system Hilbert space.\footnote{We assume $\beta_i>0$ since the complex phases of $\beta_i$ can be absorbed by $U_i$.} 
\footnote{This technique can also be used to add non-unitary operators by first block-encoding them and then adding the unitaries using LCU.}

This can be achieved in the following way. 
First, consider an ancilla register with Hilbert space ${\cal H}_a$ containing $k = \lceil \log M \rceil$ qubits.
Labeling the computational basis of these states through the integers corresponding to the binary string encoded in the qubits, one defines a so-called \emph{prepare oracle} ($\PREPARE$) to create a state $\sket{\beta}$:
\begin{align}
    \label{eq:prepstate}
   \sket{\beta}_a = \PREPARE_a \sket{0}_a = \frac{1}{\sqrt\beta} \sum_{i=0}^{M-1} \sqrt{\beta_i} \sket{i}_a
   \,,
\end{align}
where $\sket{0}_a$ corresponds to the state with all qubits in ${\cal H}_a$ in the $\sket{0}_a$ state and we have defined $\beta \equiv \|\vec{\beta}\|_1 = \sum_i |\beta_i|$.

Next, using the so-called \emph{select oracle} ($\SELECT$), one performs operations $U_i$ on states in ${\cal H}_s$, controlled on the values of qubits in ${\cal H}_a$
\begin{align}
    \label{eq:lcuselect}
    \SELECT = \sum_{i=0}^{M-1} \sket{i}_a \sbra{i}_a \otimes U_i 
    \,.
\end{align}
At the end one reverses the prepare oracle. 
The complete circuit on the Hilbert space ${\cal H}_{as} = {\cal H}_{a} \otimes {\cal H}_{s}$ can therefore be written as
\begin{align}
    \label{eq:lcube}
    U_T = (\PREPARE_a^\dagger \otimes \unit_s) \, \SELECT \, (\PREPARE_a \otimes \unit_s) \,.
\end{align}
Letting this operator act on a state $\sket{0}_a \sket{\psi}_s$ one finds
\begin{align}
    U_T\sket{0}_a \sket{\psi}_s = \frac{1}{\beta}\sket{0}_a T \sket{\psi}_s + \sket{\bot}_{as}
    \,,
\end{align}
where $\sket{\bot}$ is an unnormalized state transverse to $(\sket{0}_a\otimes \unit_s)$.
This implies that $U_T$ is a BE of $T$ with a scale factor $\beta$.

From the above discussion it is clear that the number of ancilla qubits to combine $M$ unitaries is given by ${\rm dim}({\cal H}_a) \sim \log M$, and the number of gates is given by $M$ gates that are controlled on $\log M$ qubits, which can be reduced to $\OO(M \log M)$ elementary gates.

While the presented approach to LCU relies on the binary encoding of terms on the RHS of~\cref{eq:lcuT} in the ancillary register, other encodings can be used as well.
In particular, a variation of the LCU algorithm exists, which takes advantage of the Hamiltonian locality~\cite{Kirby:2022ncy}.
For an $k$-local $\nq$-qubit Hamiltonian, it requires $2\nq$ ancillary qubits yet only uses $2k$-controlled gates.

The LCU algorithm can be applied to a Hamiltonian provided in a form of a ${2^{\nq}\times2^{\nq}}$ matrix, upon decomposing it into  Pauli strings as $H=\sum_{i=0}^{4^{\nq}-1} \beta_i P_i$, where $P_i$ are the Pauli strings and the coefficients $c_i$ are found as $c_i=\frac{1}{2^{\nq}}\operatorname{Tr}(H P_i)$.
An alternative approach to constructing BEs for Hamiltonians given in the matrix form is considered in~\cite{camps2022fable}.

The LCU approach is often combined with other methods.
In~\cref{sssec:lcuft} we shall consider its modification which takes into account the structure of particular lattice Hamiltonians.

Lastly, we briefly mention that simulation techniques based on BE, considered in the next section, are compatible with a wide class of \emph{sparse} Hamiltonians, including those which cannot be efficiently mapped onto Pauli strings.
A general approach to constructing BEs for such systems is based on the usage of \emph{sparse oracle} subroutines.
As was mentioned in~\cref{sec:introduction}, their construction is highly problem-specific~\cite{Toloui:2013xab,Babbush:2018ywg,Kreshchuk:2020dla,Kirby:2021ajp,rhodes2024exponential} and does not seem to provide advantage over LCU for local HLFTs ~\cite{rhodes2024exponential}, which are considered the main application of this work.


\subsubsection{Qubitization}
\label{sssec:qubitization}

In approaches to quantum simulation based on Quantum Signal Processing (QSP), one approximates the time evolution operator  $\me^{-i H t}$ by constructing polynomials of the Hamiltonian operator.
In what follows we shall base our algorithms upon the Qubitization technique~\cite{Low:2016znh}, which utilizes a special \emph{qubitized} form of block encoding $W_H$, also known as \emph{Szegedy quantum walk operator} or the \emph{iterate}).\footnote{\label{footnote:QSVT}A different QSP-based approach to approximating $\me^{-i H t}$ amounts to using Quantum Singular Value Transformation~\cite{gilyen2019quantum,Lin:2022vrd}. This method is advantageous in the sense that it avoids controlled calls to the block encoding, and uses directly $U_H$ instead of $W_H$. However, as such an approach does not allow one to simultaneously implement both even and odd parts of $\me^{-i H t}$, it requires adding them using LCU at the final stage. This, in turn, leads to success probabilities which (even upon potential amplitude amplification) cannot be made exponentially close to 1, as required by the HHKL algorithm.}

The qubitized block encoding $W_H$ is designed so that its successive applications to states of the form $\sket{g}_a\sket{\lambda}_s$ produce a state belonging to the subspace $\operatorname{span}\{\sket{g}_a \sket{\lambda}_s, \sket{\bot^\lambda}_{as}\}$, for all eigenstates $\sket{\lambda}$.
This fact implies that repeated applications of $W_H$ produce block encodings of powers of $H$, which would not be true for $U_H$. 
Such a scenario is realized iff each state $\sket{\bot^\lambda}_{as}$ factorizes as
$\sket{\bot^\lambda}_{as} \equiv \sket{\bot^\lambda}_a \sket{\lambda}_s$:
\begin{equation}
\begin{aligned}
\label{eq:Whdef}
W_H \sket{g}_a\sket{\lambda}_s &= \lambda\, \sket{g}_a\sket{\lambda}_s - \sqrt{1-\lambda^2} \sket{\bot^\lambda}_{as} \\
    &= \left( \lambda\, \sket{g}_a - \sqrt{1-\lambda^2} \sket{\bot^\lambda}_a
 \right)\sket{\lambda}_s
    \,.
\end{aligned}    
\end{equation}
Note that the minus sign between the two terms is conventional.

For any dimension of ${\cal H}_a$, the orthogonal state $\sket{\bot^\lambda}_a$ can now be defined as
\begin{align}
\sket{\bot^\lambda}_a\sket{\lambda}_s = \frac{\lambda \unit_{as} - W_H  }{\sqrt{1-\lambda^2}}\sket{g}_a\sket{\lambda}_s
    \,,
\end{align}
which
implies that for each eigenvalue $\lambda$ of the Hamiltonian there is one specific transverse vector $\sket{\bot^\lambda}_a$ in ${\cal H}_a$. 

From~\cref{eq:Whdef} it follows that the operator $W_H$ is block diagonal (in the sense that it doesn't mix the Hamiltonian eigenstates), and can be written as
\begin{align}
    \label{eq:wwtilde}
    \left(\unit_a \otimes \sbra{\lambda'}_s\right) W_H \left(\unit_a \otimes \sket{\lambda}_s\right) = \widetilde{W}_H^{(\lambda)} \delta_{\lambda \lambda'}
    \,,
\end{align}
where $\widetilde{W}_H^{(\lambda)}$ is some operator acting on the register $a$.
Information relevant to block encoding of $H$ is contained in two-dimensional blocks of $\widetilde{W}_H^{(\lambda)}$ corresponding to subspaces
\begin{equation}
\hilb_{a}^{(\lambda)}={\rm span}\{\sket{g}_a, \sket{\bot^\lambda}_{a}\} \in {\cal H}_a \,,
\end{equation}
which we denote by $W_H^{(\lambda)}$:
\begin{align}
\label{eq:wwtildematrix}
\begin{aligned}
W_H^{(\lambda)} &= \begin{pmatrix} \sbra{g}_a \widetilde{W}_H^{(\lambda)} \sket{g}_a & \sbra{g}_a \widetilde{W}_H^{(\lambda)} \sket{\bot^\lambda}_a \\  \sbra{\bot^\lambda}_a  \widetilde{W}_H^{(\lambda)} \sket{g}_a & \sbra{\bot^\lambda}_a  \widetilde{W}_H^{(\lambda)} \sket{\bot^\lambda}_a \end{pmatrix}_{\hilb_{a}^{(\lambda)}}
    \\
    &= \begin{pmatrix} \lambda & - \sqrt{1-\lambda^2} \\ \sqrt{1-\lambda^2} & \lambda \end{pmatrix}_{\hilb_{a}^{(\lambda)}}
    \\
    &= \exp ( -i Y \theta_\lambda)
    \,,    
\end{aligned}
\end{align}
where 
\begin{align}
    \label{eq:thetalambda}
    \theta_\lambda = \arccos(\lambda) \,.
\end{align}

While the qubitized BE $W_H$ gives the same result as the original BE $U_H$ when acted on $\sket{g}_{a} \sket{\lambda}_s$, it has a much simpler structure when applied multiple times:
\begin{align}
\label{eq:oneWone}
    \left(\unit_{a} \otimes \sbra{\lambda'}_s \right) \left(W_H\right)^k \left(\unit_{a} \otimes \sket{\lambda}_s \right) = \left(\widetilde{W}_H^{(\lambda)}\right)^k\delta_{\lambda' \lambda}
    \,.
\end{align}
Using Eq.~\eqref{eq:wwtildematrix}, one can show that
\begin{align}
\label{eq:wwtildematrixpowerk}
\begin{aligned}
&\left(W_H^{(\lambda)}\right)^k 
    \\
    &= \begin{pmatrix} \sbra{g}_a \left(\widetilde{W}_H^{(\lambda)}\right)^k \sket{g}_a & \sbra{g}_a \left(\widetilde{W}_H^{(\lambda)}\right)^k \sket{\bot^\lambda}_a \\  \sbra{\bot^\lambda}_a  \left(\widetilde{W}_H^{(\lambda)}\right)^k \sket{g}_a & \sbra{\bot^\lambda}_a  \left(\widetilde{W}_H^{(\lambda)}\right)^k \sket{\bot^\lambda}_a \end{pmatrix}_{\hilb_a^{(\lambda)}}
    \\
    &=\begin{pmatrix} T_k(\lambda) & - \sqrt{1-\lambda^2} \, U_{k-1}(\lambda)  \\ \sqrt{1-\lambda^2} \, U_{k-1}(\lambda) & T_k(\lambda) \end{pmatrix}_{\hilb_a^{(\lambda)}}
    \\
    &= \exp \left( -i Y k \theta_\lambda \right),
\end{aligned}
\end{align}
where $T_k(\lambda)$ and $U_k(\lambda)$ are Chebyshev polynomials of the first and second kind, respectively.
Given that Chebyshev polynomials are an excellent 
polynomial basis for approximating $L^\infty$ functions on a finite interval, one can hope to construct a general function of the eigenvalues (and therefore of the Hamiltonian matrix itself) by combining powers of the $W_H$ operator. 
Note that while the matrix of $\widetilde{W}_H^{(\lambda)}$ may contain blocks of size larger than $2\times 2$ when the ancillary register contains more than a single qubit, only its $2\times2$ blocks $W_H^{(\lambda)}$ are relevant for the block encodings of $H$ and its powers.
In the next section we discuss how this technique can be used for approximating the time evolution operator.

While so far we have only defined the unitary $W_H$ through its action on a given state, a very important result is that it can be obtained from a BE of the Hamiltonian $U_H$ using a process known as \emph{qubitization}~\cite{Low:2016znh}, which we now describe.
In many cases of interest, it is possible to construct the iterate by assuming a form $W_H = S' U_H$, and searching for some $S'$ that acts only on the ancillary qubits. 
The conditions $S'$ must satisfy can be stated more succinctly by first letting $S'$ be a product of a reflector $R_g$ about the state $\sket{g}_a$ and another operator $S$ that acts only on the ancillary qubits, i.e.
\begin{equation}
\label{eq:Whdefas}
    W_H = \left(R_g \otimes \unit_s \right) \left( S \otimes \unit_s \right) U_H\,,
\end{equation}
where 
\begin{equation}
\label{eq:refl_a}
    R_g = \left(2 \sket{g}_a \sbra{g}_a - \unit_a \right).
\end{equation}
The circuits for $W_H$ and $R_g$ are shown in~\Cref{fig:WH,fig:Rg}.

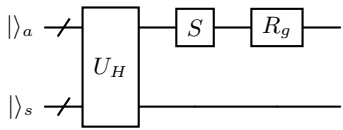
\begin{figure}[h]
\begin{quantikz}
    \lstick[1]{$\sket{}_a$} & \gate[2]{U_H} \qwbundle{} & \gate{S} & \gate[style={inner ysep=-1.1pt}]{R_g} & 
    \\
    \lstick[1]{$\sket{}_s$} & \qwbundle{}  & & & 
\end{quantikz}
\caption{
Circuit for the qubitized block encoding $W_H$ defined in~\cref{eq:Whdefas}.
}
\label{fig:WH}
\end{figure}

\begin{figure}[h]
\begin{quantikz}
\lstick[4]{$\sket{}_a$} & \gate[4]{G^\dagger} & & \octrl{1} & & \gate[4]{G} & \qw & \\
\vdots \ \setwiretype{n} &  & & \ \vdots \ & & & \vdots & \\
& & & \octrl{1} \wire[u][1]{q} & & & \qw & \\
& & \gate{\Xgate} & \gate{\Zgate} & \gate{\Xgate} & & \gate{-1} &
\end{quantikz}
\caption{
Circuit for the reflection operator $R_{a}$ defined in~\cref{eq:refl_a} and used in~\cref{eq:Whdefas}, see~\cite{Lin:2022vrd} for alternative implementations. $G$ is the oracle for preparing the state $\sket{g}_a$, i.e. ${G\sket{0}_a=\sket{g}_a}$.
The $\boxed{-1}$ gate adds an overall phase to the circuit (negative sign), and can be implemented, e.g., as $\Xgate\Zgate\Xgate\Zgate$.
}
\label{fig:Rg}
\end{figure}
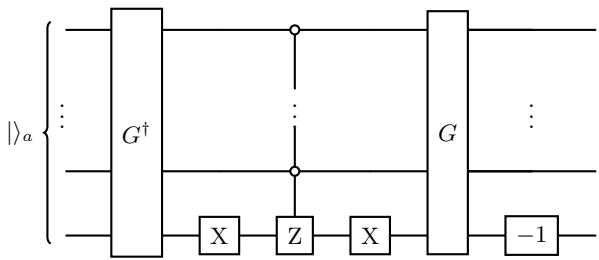

Written in this way, it was shown~\cite{Low:2016znh} that one can construct $W_H$ from \Cref{eq:Whdefas} if $S$ satisfies
\begin{subequations}
\label{eq:S_cond}
\begin{align}
    \label{eq:S_cond1}
    &\left(\sbra{g}_a \otimes \unit_s \right) \left(S \otimes \unit_s\right) U_H \left(\sket{g}_a \otimes \unit_s \right) = H/\alpha, 
    \\
    \label{eq:S_cond2}
    &\left(\sbra{g}_a \otimes \unit_s \right) \left[ \left(S \otimes \unit_s\right) U_H \right]^2 \left(\sket{g}_a \otimes \unit_s \right) = \unit_a\,.
\end{align}
\end{subequations}

One particularly simple case is when $U_H$ is its own inverse, \emph{i.e.} $U_H^2 = \unit_{as}$.
In this case, one can simply set $S = \unit_{a}$.
This construction can be applied to the case when the block encoding $U_H$ is constructed via the LCU procedure~\cite{Babbush:2018ywg}. 
To see this, we must modify our perspective slightly.
In the original construction, we had $U_H = (\PREPARE^\dagger) \SELECT (\PREPARE)$ and $\sket{g}_a = \ket{0}_a$.
Alternatively, we can regroup terms and instead take $U_H \to U_H^\text{LCU} \equiv \SELECT$ and $\sket{g}_a \to \sket{g^\text{LCU}}_a \equiv \PREPARE \sket{0}_a = \sket{\beta}_a$.
Because $(U_H^\text{LCU})^2 = \SELECT^2 = \unit_{as}$, we can simply set $S = \unit_{a}$ to construct $W_H$:
\begin{equation}
\label{eq:WH_for_lcu}
\begin{aligned}
    &\left(\sbra{0}_a \otimes \unit_s\right)U_H\left(\sket{0}_a \otimes \unit_s\right)
    \\
    &=\left(\sbra{0}_a \otimes \unit_s\right)(\PREPARE^\dagger) \SELECT (\PREPARE)\left(\sket{0}_a \otimes \unit_s\right)
    \\
    &=\left(\sbra{\beta}_a \otimes \unit_s\right) U_H^\text{LCU} \left(\sket{\beta}_a \otimes \unit_s\right).
\end{aligned}
\end{equation}
It is important to note that the BE $U_H^\text{LCU}$ is defined with $\ket{g}_a =\PREPARE \sket{0}_a$.

The construction described above also applies to cases where the block encoding is based on Hamiltonian sparsity, since one can use sparse oracles to construct block encodings that are Hermitian and thus self-inverse~\cite{Lin:2022vrd}. In more general scenarios where this is not necessarily the case (e.g., as in Ref.~\cite{Novikau:2021yra}), the operator $S'$ can always be constructed at the expense of introducing a single ancillary qubit and performing two controlled calls to $U_H$~\cite{Low:2016znh}.
Note that while this procedure does not increase the asymptotic cost of preparing $W_H$, it does result in a significant increase of the overall prefactor. 
Importantly, for all methods used in this work, the BE's are also Hermitian, and the general procedure requiring an extra qubit was avoided.
For completeness, the procedure for constructing $W_H$ from an arbitrary $U_H$ is presented in App.~\ref{app:general_W}.

\subsubsection{
QSP using Qubitization\label{sec:qsptimeevol}}

\begin{figure*}[t]
\begin{quantikz}
\lstick{$\sket{0}_b$} & \gate{\Hgate} & \gate[3]{V_H(\phi_{1})} & \gate[3]{V_H(\phi_{2})} & \qw \ \cdots \ & \gate[3]{V_H(\phi_{\Nphi})} & \gate{\Hgate} & \meter{} & \setwiretype{c} \rstick{0} \\
\lstick{$\sket{0}_a$} & \qw \qwbundle{} & & \qw \cdots & \qw \ \cdots \ & & \qw & \meter{} & \setwiretype{c} \rstick{0}  \\
\lstick{$\sket{\psi}_s$} & \qw \qwbundle{} & & & \qw \ \cdots \ & & \qw & \qw & \qw \rstick{$\approx \me^{-iHt}\sket{\psi}$}
\end{quantikz}
\caption{Quantum Signal Processing (QSP) circuit used for simulating time evolution, reproduced from~\cite{Novikau:2021yra}.
The circuit prepares a block encoding for the approximation of the time evolution operator $\me^{-iHt}\sket{\psi}$.
Circuit for $V_H(\phi)$ is shown in~\cref{fig:VH_bas}.
When all the ancillary qubits are measured in the zero states, the desired operator is applied to the system register $\sket{\psi}_s$.
Without loss of generality, $\sket{g}_a$ is set to $\sket{g}_a=\sket{0}_a$.
}
\label{fig:QSP}
\end{figure*}
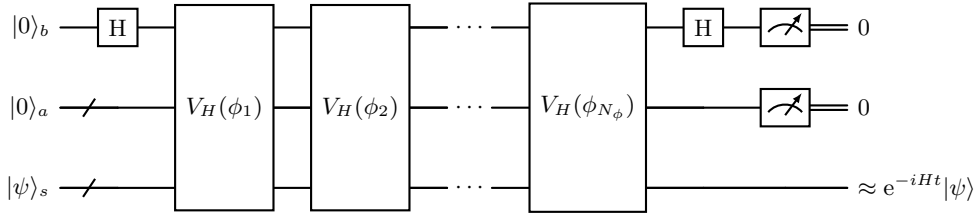

\begin{figure}[t]
\begin{quantikz}
\lstick{$\sket{}_b$} & \gate{R_z(-\frac{\phi}{2})} & \gate{\Hgate} & \ctrl{1} & \gate{\Hgate} & \gate{R_z(\frac{\phi}{2})} & \\
\lstick{$\sket{}_a$} & \phantomgate{R_z(-\frac{\phi}{2})}\qwbundle{} & & \gate[2]{W_H} & & & \\
\lstick{$\sket{}_s$} & \phantomgate{R_z(-\frac{\phi}{2})}\qwbundle{} & & & & &
\end{quantikz}
\caption{Circuit for $V_H(\phi)$, as defined is~\cref{eq:VHdef}.}
\label{fig:VH_bas}
\end{figure}
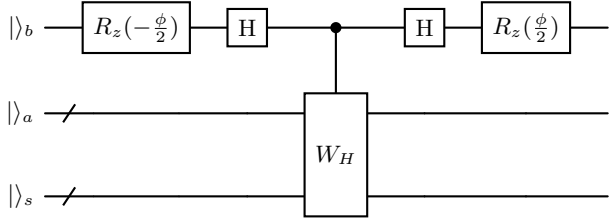
Having just demonstrated that repeated applications of $W_H$ construct block encodings of Chebyshev polynomials of the Hamiltonian, one natural strategy for constructing a polynomial approximation to $\me^{-i t H}$ is as follows. 
One would first construct circuits with different powers of $W_H$ and then add these circuits with the appropriate coefficients using an LCU procedure.  
However, this would result in a block encoding approximating $\me^{-i H t}$ with a scale factor that grows with the degree of the polynomial approximation and cannot be made arbitrarily close to one.
This is highly undesirable for the reasons previously discussed in~\Cref{footnote:QSVT}, and below we discuss how this limitation can be overcome with the aid of QSP.

In order to use the operator $W_H$ of the previous section to calculate the exponential of the Hamiltonian, we first note that within the subspaces $\operatorname{span}\{\sket{g}_a \sket{\lambda}_s, \sket{\bot^\lambda}_{as}\}$ the qubitized operator $W_H$ can be diagonalized by defining the states
\begin{align}
    \sket{\theta_\pm^\lambda}_{as} \equiv \frac{1}{\sqrt{2}} \left( \sket{g}_{a} \pm i \sket{\bot^\lambda}_{a} \right)\sket{\lambda}_s
    \,.
\end{align}
Using \cref{eq:Whdef} one can easily show that
\begin{align}
    W_H \sket{\theta_\pm^\lambda}_{as} = \me^{i \theta_\pm^\lambda}\sket{\theta_\pm^\lambda}_{as}
    \,,
\end{align}
with
\begin{align}
    \theta_\pm^\lambda
    = \mp \theta_\lambda
    \,,
\end{align}
where $\theta_\lambda$ is defined in~\cref{eq:thetalambda}.\footnote{The fact that $W_H$ encodes information about the spectrum of the system via $\operatorname{spec}W^{(\lambda)}_H=\{\me^{\pm i\arccos(\lambda)}\}$ can be used for implementing a highly efficient version of Quantum Phase Estimation algorithm based using controlled calls to $W_H$~\cite{Babbush:2018ywg}.}

Given this form, one can now define a new operator $V_H(\phi)$, parameterised by $\phi$, which can be constructed from the operator $W_H$. 
It requires one more ancillary qubit $\sket{b} \in {\cal H}_b$, such that the operator $V_H(\phi)$ acts on the Hilbert space ${\cal H}_{bas}  \equiv {\cal H}_b \otimes {\cal H}_a \otimes {\cal H}_s$. 
It is defined by (see~\cref{fig:VH_bas})
\begin{align}
\label{eq:VHdef}
    V_H(\phi) = (\me^{-i \phi Z_b/2}\otimes \unit_{as}) W'_H (\me^{i \phi Z_b/2}\otimes \unit_{as})
    \,,
\end{align}
where $Z_b$ is the usual Pauli $Z$ operator acting on the qubit $\sket{b}$, and $W'_H$ is the $W_H$ operator controlled on the qubit $\sket{b}$, which now acts on the Hilbert space ${\cal H}_{bas}$:
\begin{align}
    W_H' = \sket{+}_b \sbra{+}_b \otimes \unit_{as} + \sket{-}_b \sbra{-}_b \otimes W_H
    \,.
\end{align}

One can show through explicit computation that
\begin{align}
\begin{aligned}
\bigl(\unit_b \otimes \sbra{\theta^{\lambda'}_\pm}_{as}\bigr)    V_H(\phi)&\bigl(\unit_b \otimes \sket{\theta^{\lambda}_\pm}_{as}\bigr)\\
    & \quad = V_H^{(\theta^\lambda_\pm)} (\phi) \, \delta_{\lambda \lambda'}
    \,,    
\end{aligned}
\end{align}
with
\begin{align}
\label{eq:VHphi_bas}
\begin{aligned}
    V_H^{(\theta)}(\phi) &=     
    \me^{i \frac{\theta}{2}}\begin{pmatrix}
    \cos\frac{\theta}{2} & - i \sin\frac{\theta}{2}\me^{-i \phi} \\
    \vphantom{\Bigl(}
    - i \sin\frac{\theta}{2}\me^{i \phi} & \cos\frac{\theta}{2}
    \end{pmatrix}_{\!b} \\
    &= \me^{i \frac{\theta}{2}} \me^{-i \frac{\theta}{2}\left(X_b \cos\phi + Y_b \sin\phi \right)}
    \,,
\end{aligned}
\end{align}
where $X_b$ and $Y_b$ denote the Pauli-$X$ and $Y$ matrices in ${\cal H}_b$. 

What this construction has accomplished is to create a unitary operator $V_H(\phi)$ that is block diagonal in the original Hilbert space ${\cal H}_s$, with the entries of the $2 \times 2$ matrix in ${\cal H}_b$ given by functions of the eigenvalues of the Hamiltonian (through $\cos \theta_\pm^\lambda$), parameterized by some $\phi$. 
By multiplying several of these operators in sequence (each with a different value of~$\phi$), one obtains an operator
\begin{align}
\label{eq:vprod}
    V_H^{(\theta)}(\vec \phi) \equiv V_H^{(\theta)}(\phi_{\Nphi}) \ldots V_H^{(\theta)}(\phi_{1})\,,
\end{align}
that is still block diagonal in ${\cal H}_s$.
The resulting operator is a $2\times 2$ matrix in ${\cal H}_b$ which acquires the form of
\begin{align}
\label{eq:VHGeneralMatrix}
    V_H^{(\theta)}(\vec \phi) = \mathcal{A}(\theta) \unit_b + i \mathcal{B}(\theta) X_b + i \mathcal{C}(\theta)Y_b + i \mathcal{D}(\theta) Z_b
    \,,
\end{align}
where the functions $\mathcal{A}(\theta) \ldots \mathcal{D}(\theta)$ are all parameterized by the angles $\vec\phi = (\phi_1, \ldots, \phi_{\Nphi})$ and, similarly to \cref{eq:wwtildematrixpowerk}, can be expressed in terms of Chebyshev polynomials of $\lambda$ upon substituting $\theta = \arccos(\lambda)$.

The classes of functions $\mathcal{A}(\theta) \ldots \mathcal{D}(\theta)$ that can be realized by the above procedure are described by the theory of Quantum Signal Processing.
In particular, it was shown in~\cite{Low:2016jxt} that
one can construct the vector of angles $\vec\phi$ to obtain functions $\mathcal{A}(\theta)$ and $\mathcal{C(\theta)}$ of the form
\begin{align}
\begin{aligned}
    |\mathcal{A}^2(\theta)|&  + |\mathcal{C}^2(\theta)| \leq 1 \,, \qquad
    \mathcal{A}(0) = 1\,,\\
    \mathcal{A}(\theta) &= \sum_{k=0}^{{\Nphi}/2} a_k \cos (k\theta)\,,\\
    \mathcal{C}(\theta) &= \sum_{k=1}^{\Nphi/2} c_k \sin (k\theta)
\,.
\end{aligned}
\end{align}
The two other functions, $\mathcal{B}(\theta)$ and $\mathcal{D}(\theta)$ are not relevant to the calculation since projecting onto $\sket{+}_b$ leaves only $\mathcal{A}(\theta)$, $\mathcal{C}(\theta)$.
Reversing this logic, one can start from some desired functions $\mathcal{A}(\theta)$, $\mathcal{C}(\theta)$, and determine phases $\vec{\phi}$ such that $V_H^{(\theta)}(\vec \phi)$ reproduces these functions (approximately). 
This can be accomplished using an efficient classical procedure given in~\cite{Low:2016jxt}.\footnote{A recent version of QSP, \emph{generalized QSP}, significantly improves the cost of calculating phases~$\vec{\phi}$~\cite{Motlagh:2023oqc}.}

We omit now a few final details, which can  be found in~\cite{Low:2016znh}, and simply state the relevant results. 
As can be confirmed from~\cref{eq:VHGeneralMatrix}, letting the operator $V_H^{(\theta)}(\vec \phi)$ act on the state $\sket{+}_b$ and post-selecting on $\sbra{+}_b$ picks out the function 
\begin{equation}
\label{eq:Vfinal}
\begin{aligned}    
    \sbra{+}_b\sbra{g}_{a} \sbra{\lambda'}_s &V_H(\vec \phi)\sket{+}_b\sket{g}_{a} \sket{\lambda}_s
    \\
    & \qquad = \delta_{\lambda \lambda'} \left[ A(\lambda) + i C(\lambda) \right]
    .
\end{aligned}    
\end{equation}
The functions $A(\lambda)$ and $C(\lambda)$ are directly related to the functions $\mathcal{A}(\theta)$, $\mathcal{C}(\theta)$, with the detailed relation given in Ref.~\cite{Low:2016znh}.
The success probability for measuring the system in the state $\sket{+}_b\sket{g}_{a}$ for a given value $\lambda$ depends on the functions and is given by 
\begin{align}
\label{eq:QSP_p_def}
    p \geq | A^2(\lambda)| + | C^2(\lambda)|
    \,.
\end{align}
The final task is to determine which functions ${A}(\lambda)$ and ${C}(\lambda)$ (and therefore which phases $\vec{\phi}$) approximate the time evolution operator with uncertainty bounded by $\varepsilon$,
\begin{align}
\label{eq:qspapprox}
    \left| \me^{i \lambda \tilde{t}} - [ A(\lambda) + i C(\lambda)] \right| < \varepsilon
    \,.
\end{align}
Since ${A}(\lambda)$ and ${C}(\lambda)$ can be expressed in terms of Chebyshev polynomials, they can be used for approximating the real and imaginary part of $\me^{i \lambda \tilde{t}}$ through the rapidly converging Jacobi-Anger expansion~\cite{Low:2016sck}:
\begin{subequations}
\begin{align}
    & \begin{aligned}
    A(\lambda) = J_0(t) + 2\sum_{k\text{ even}>0}^{N_\phi/2} (-1)^{k/2} J_k(\tilde{t}) T_k(\lambda)\,, 
    \end{aligned}
    \\
    & \begin{aligned}
    C(\lambda) = 2 \sum_{k\text{ odd}>0}^{N_\phi/2} (-1)^{(k-1)/2} J_{k}(\tilde{t}) T_{k}(\lambda)\,,
    \end{aligned}
\end{align}
\end{subequations}
(note that $N_\phi$ is always chosen to be even).
In this case, the circuit in \cref{fig:QSP} would implement the approximate BE of the evolution operator $\me^{-iHt}$.
The number of phases needed for a given value of $\varepsilon$ is bounded by~\cite{Low:2016sck}
\begin{align}
    \label{eq:nphi}
    \Nphi = \OO\!\left(\tilde{t} + \log(1/\varepsilon)\right) = \OO\!\left(\alpha t + \log(1/\varepsilon)\right)
    .
\end{align}
Note also that~\cref{eq:QSP_p_def,eq:qspapprox} imply that 
\begin{align}
    p \geq 1 - \OO(\varepsilon)
    \,,
\end{align}
since the time evolution operator has unit magnitude.
Note that due to the $\log(1/\epsilon)$ scaling, one can require the success probability to be exponentially close to one for only a polynomial cost (see also~\Cref{footnote:QSVT}).

It is worth mentioning that the success rate of QSP is always $\OO(1)$, independent of the success rate $\OO(1/\alpha)$ of the BE $U_H$ used. 
This is due to the fact that we have used QSP to implement a unitary function ($\me^{-i H t}$) which has unit norm. 
This implies that $ A(\lambda)^2 + C(\lambda)^2 \approx 1$, which together with \cref{eq:QSP_p_def} gives a success probability that is close to one.
While this is a desirable feature, the cost of a poor success probability of the BE is now reflected in the number of phases $N_\phi$ required to approximate the time evolution operator to an error $\epsilon$, as seen by Eq.~\eqref{eq:nphi}.
For this reason, significant savings can be achieved by performing dedicated studies to minimize $\alpha$. 

A number of algorithms, based on the concept of QSP, can be used for quantum simulation.
Algorithms such as Quantum Singular Value Transformation (QSVT)~\cite{gilyen2019quantum,Lin:2022vrd} and (QETU)~\cite{Dong:2022mmq,Chan:2023lbm,Kane:2023jdo} can be utilized for near-optimal state preparation~\cite{Lin:2020zni,Dong:2022mmq,Kane:2023jdo} and construction of block encodings~\cite{bebyqsp}, while utilizing the recently developed Generalized QSP (GQSP) algorithm~\cite{Motlagh:2023oqc} can readily improve by a constant factor the gate counts obtained with the QSP procedure~\cite{bebyqsp}.

\subsection{HHKL}
\label{sec:hhkl}
The results presented so far (PFs and QSP) can be applied to general site-local Hamiltonians (see~\cref{def:loclattice}).
For geometrically-site-local Hamiltonians (see~\cref{def:geomloclattice}) one can use the fact that information from one area of the lattice can only spread to areas not directly connected to it through repeated actions of the time evolution operators, which implies that the speed with which this information spreads is bounded. 
This is formalized by the Lieb-Robinson bound~\cite{Lieb:1972wy}.
This idea was used in a paper by Haah, Hastings, Kothari and Low (HHKL)~\cite{Haah:2018ekc} to devise a scheme for simulating the time evolution of geometrically-site-local Hamiltonians. 

As before, this section is aimed at giving the reader a basic understanding of the technique, without going into too many of the details. 
In order to understand these details the reader is encouraged to read the excellent description of the algorithm in the original work, where all statements are also carefully proved. 

The main ingredient to the HHKL algorithm is that the time evolution of a large quantum system with geometrically local Hamiltonian can be obtained from the time evolution of two smaller systems that share some overlap, which has to be of at least the size given by the ball of size $\OO(1)$ used in~\cref{def:geomloclattice}. 
Thus, the system Hilbert space ${\cal H}_s$ is divided into three parts, ${\cal H}_{a,b,c}$ such that ${\cal H}_s = {\cal H}_a \otimes {\cal H}_b \otimes {\cal H}_c$.
In the discussion of~\cite{Haah:2018ekc} is assumed that the local Hamiltonian $H_J$ such that $||H_J|| \leq 1$. Thus, in the remainder we work with 
\begin{align}
    \tilde{H}_J = \frac{H_J}{||H_J||}\,, \qquad \tilde{t} = ||H_J|| t
    \,.
\end{align}
First, one uses that the evolution of a system for some long time $\tilde{t}$ can be split into the product of unitaries over shorter times $\Delta t=\OO(1)$:
\begin{align}
    U(\tilde{t}, 0) = U(\tilde{t}, \tilde{t}-\Delta t) \cdots U(\Delta t,0)
    \,.
\end{align}
The total number of such unitaries is $\OO(\tilde{t})$.

Lemma 6 of~\cite{Haah:2018ekc} states that the time evolution $U_s(\Delta t)$ of a Hamiltonian on the whole system over time $\Delta t$ can be obtained by multiplying together the forward time evolution in the system ${\cal H}_{a\cup  b} = {\cal H}_a \otimes {\cal H}_b$ and ${\cal H}_{b\cup c} = {\cal H}_b \otimes {\cal H}_c$ and the backwards evolution in the system ${\cal H}_b$ as
\begin{align}
    \label{eq:hhklmain}
    U_{a\cup b\cup c}(\Delta t) = U_{a\cup b}(\Delta t) U^\dagger_b(\Delta t) U_{b\cup c}(\Delta t)  + \OO(\delta)
    \,,
\end{align}
where the error due to this approximation is given by
\begin{align}
\label{eq:HHKL_delta_def}
    \delta \equiv \delta(a,b,c) &= \mathcal{O}\left(\me^{-\mu \, {\rm d}(a,c)} \sum_{X\in {\rm BD}(ab,c)}||\tilde{H}_X||\right)
    \,.
\end{align}
For now we assume that each of the $U(\Delta t)$ on the RHS of~\cref{eq:hhklmain} is implemented exactly. 
In~\cref{eq:HHKL_delta_def} $\mu > 0$ is a constant, ${\rm d}(a,c)$ is the smallest distance between any point in $a$ and any point in $c$, and the norm $||\tilde{H}_J||$ is summed over all terms that are on the boundary of $a\cup b$ and $c$, indicated by ${\rm BD}(ab,c)$. 
The uncertainty is therefore exponentially suppressed in the distance between the regions $a$ and $c$.
This implies that this approximation works well for large enough distance $\ell \equiv {\rm d}(a,c)$.

This result can be applied recursively, splitting the evolution into smaller and smaller pieces, each of which just have to be larger than the ball of size $\sim \OO(1)$.
If the size of the original system is given by $L^\dims$ (where $\dims$ represents the dimension of the Euclidean space we consider with size $L$ in each dimension) and the final pieces are each of size  $\ell^\dims$, one ends up with 
\begin{align}
    m = \OO\!\left(\frac{L^\dims}{\ell^\dims}\right)
\end{align}
individual pieces.
Taking into account the $\mathcal{O}(\tilde  t)$ evolution operators we started from, the final error $\mathcal{O}(\varepsilon)$ can be reached by requiring 
\begin{align}
\label{eq:delta_def}
    \delta = \OO\!\left( \frac{\varepsilon}{\tilde{t} \, m} \right)
    .
\end{align}
We see that the sum of the norms of the boundary terms appearing in \cref{eq:HHKL_delta_def} is $\OO(L^{\dims-1})$ and thus from \cref{eq:HHKL_delta_def,eq:delta_def} one finds it suffices to choose
\begin{align}
    \ell = \Theta\!\left( \log \left( \frac{L^\dims \tilde{t}}{\varepsilon} \right)\right)
    .
\end{align}

One can then estimate the required resources using the fact that one needs to compute 
\begin{align}
    \OO(\tilde{t} m) = \OO\!\left( L^\dims \tilde{t} \log^{-\dims}\left(\frac{L^\dims \tilde{t}}{\eps} \right)\right)
\end{align}
time evolution operators on a system of size $\OO(\ell^\dims)$.\footnote{Note that decreasing $\eps$ is logarithmically decreasing the number of evolution operators. This number clearly has to be larger than unity, meaning $\eps$ cannot be exponentially small with system size.}
The cost and the associated additional error for implementing each of these ``local'' pieces depends on the algorithm used for their implementation.
In particular, if an
algorithm with cost polynomial in $\ell$ and polylogarithmic in $1/\eps$ is used (such as Qubitization), then the cost of the HHKL algorithm remains $\OO(L^\dims \tilde{t} \polylog(L^\dims\tilde{t}/\eps))$.
Implementing the time evolution of local blocks using PF would lead to an algorithm whose asymptotic cost dependence on $\eps$ is given by that of the employed PF method, i.e., no longer polylogarithmic in $\eps$.
This can make sense if, for a given value of $\eps$, the PF-based method leads to lower gate counts as compared to a nearly-optimal implementation~-- for a system of size $\ell$ but not of size~$L$~\cite{childs2019nearly}.

\subsection{Summary and precise statement of techniques used}
\label{sec:methods_summary}

\subsubsection{General Product Formulas}
\begin{definition}
    A $p^{\text{th}}$ order PF $S_p(t)$ is a PF of the form~\eqref{def:PF}, such that it approximates the true exponential $\me^{-iHt}$ to $p^{\text{th}}$ order in $t$, i.e.,
    \begin{equation}
        \label{eq:pth_order_PF}
        \me^{-iHt} = S_p(t) + \mathcal{O}(t^{p+1})\,.
    \end{equation}
\end{definition}

\begin{theorem}
[Trotter number with commutator scaling]
\label{them:childs_pf}
        Let $H = \sum_{\gamma=1}^{\Gamma}H_\gamma$ be an operator comprised of $\Gamma$ Hermitian summands. Let $S_p(t)$ be a $p^{\text{th}}$ order PF as defined by \eqref{eq:pth_order_PF}, with $t \geq 0$. Then, if we require an exponentiation error $\mathcal{O}(\varepsilon)$, as defined by \eqref{eq:exp_error}, it suffices to choose
        \begin{equation}\label{eq:trotter_number}
            r = \OO\!\left( \frac{\Tilde{\alpha}^{\frac{1}{p}}t^{1 + \frac{1}{p}}}{\varepsilon^{\frac{1}{p}}} \right)
            ,
        \end{equation}
        where $\Tilde{\alpha} \equiv \sum_{\gamma_1,\gamma_2,...,\gamma_{p+1} = 1}^{\Gamma} \lvert\lvert [ H_{\gamma_{p+1}},... [H_{\gamma_2},H_{\gamma_1}] ] \rvert\rvert$
\end{theorem}
\begin{proof}
    This theorem follows directly from Theorem 11 and Corollary 12 of~\cite{Childs:2019hts}.
\end{proof}

\subsubsection{QSP}
\label{sec:QSP}
In this section we summarize the results of the section on QSP in a concise and precise form which will be referred to later in the text.

\begin{definition}
[Block encoding]
\label{def:be}
Let $H$ be a Hamiltonian acting on the $n$-qubit Hilbert space $\mathcal H_s\cong \mathbb C^{2^n}$.
Let $k \ge  0$ be an integer and let $U_H$ be a unitary operator on the joint Hilbert space $\mathcal H_a \otimes \mathcal H_s$, where $\mathcal H_a \cong \mathbb C^{2^k}$ denotes the Hilbert space of $k$ ancillas.
Then, for $\alpha\in \mathbb R^+$, the unitary operator $U_A$ is called an $(\alpha,k)$\textit{-block-encoding} of $A$ if for every state $\sket{\psi} \in \mathcal H_s$:
\begin{equation}
\label{eq:lcusimple}
\left(\unit_{s}\otimes \sbra{g}_g\right)U_A \left(\unit_s\otimes \sket{g}_a\right) = A/\alpha\,,
\end{equation}
where $\sket{g}_a$ is some state in $\mathcal H_a$, $\unit_{s}$ is the identity operator on $\mathcal H_s$.
The value $\alpha$ chosen such that $\|A\|/\alpha \leq 1$ is called the \textit{scale factor} of the BE.
\end{definition}

BEs can be constructed using LCU as follows
\begin{lemma}
[LCU block encoding]
\label{LCU}
    Let $U_0, \ldots, U_{M-1}$ be $M$ unitary operators acting on a system Hilbert space ${\cal H}_s$ and define the operator $T = \sum_{i=0}^{M-1} \beta_i U_i$. The LCU method provides a $(\beta, \log M)$ BE of the operator $T$, where $\beta \equiv \sum_{i=0}^{M-1} \| \beta_i \|$.
    The gate complexity is given by $\OO(M \log M)$ elementary gates.
\end{lemma}
\begin{proof}
    The proof of this Lemma can be found in many textbooks, see for example~\cite{Lin:2022vrd}.
    The asymptotic cost stems from the complexity of implementing the $\SELECT$ subroutine, since the cost of preparing an arbitrary state on $\log_2 M$ qubits in $\PREPARE$ is $\OO(M)$~\cite{Shende:2006onn,Plesch:2011vwn}, which is cheaper than the cost of $\SELECT$, as we explain next.
    The gate complexity of $\SELECT$ is $\OO(M \log M)$, which comes from sequentially implementing $M$ unitaries, each being controlled on $k=\log M$ ancillary qubits.
    Decomposing $k$-controlled gates into elementary gates can be achieved at a cost linear in $k$, with the constant prefactor being smallest if $k$ additional ancillary qubits are allowed to be used~\cite{Nielsen:2012yss,Babbush:2018ywg}.
    For applications considered in this work, this gives a minor overhead in the qubit cost since $k$ is logarithmic is the parameters of the system.
\end{proof}

Given a BE of the Hamiltonian, one can obtain an approximation to the time evolution operator using QSP.
\begin{theorem}
[QSP with Qubitization]
\label{thm:qspwq}
Let $U_H$ be a an $(\alpha,k)$-block-encoding of a Hamiltonian $H$ acting on a Hilbert space $\mathcal H_s$.
Furthermore, let $t$ and $\eps$ be positive real values. 
Then there is a quantum circuit that produces a $\left(1 \pm \mathcal O(\eps),k+2\right)$-block-encoding of the time evolution operator $\me^{-iHt}$ up to exponentiation error $\eps$, and possesses a complexity of \begin{equation} \OO\!\left(\alpha t + \log(1/\eps)\right), \end{equation} relative to the oracle $U_H$.
The probability for this BE to succeed is
\begin{align}
    p > 1 - \OO(\varepsilon)\,.
\end{align}
\end{theorem}
\begin{proof}
We will not prove this result rigorously, but point the reader to the pedagogical sections of this paper, with all statements made there having been rigorously proved in the literature~\cite{Low:2016znh}. 
The fact that the BE is $(\tilde \alpha,\tilde k) = (1 \pm \mathcal O(\eps),k+2)$ can be easily understood.
First, the BE of the time evolution operator requires two additional qubits, one ($\sket{q}$) for the qubitization and one ($\sket{b}$) to add the functional dependence on $\phi$, giving $\tilde k = k+2$.
The fact that $\tilde \alpha = 1 \pm \mathcal O(\varepsilon)$ follows from the fact that the operator we are block-encoding is a unitary operator, and our approximation is good up to error $\varepsilon$. Since the norm of a unitary operator is 1, this gives that both $\tilde \alpha$ and the success probability are unity (up to $\OO(\varepsilon)$). 
\end{proof}

\subsubsection{HHKL}

\begin{theorem}
[HHKL Algorithm]
\label{thm:hhkl_helper}
Let $\Lambda$ be a lattice of qubits embedded in $\dims$-dimensional Euclidean space of size $L$ in each dimension, and let $\hat H$ be a geometrically-local Hamiltonian, normalized such that each local term has a norm of order unity. 
Furthermore, let $\tilde{t}$ and $\eps$ be positive real values. 
Then there exists an algorithm that approximates $\me^{-i\tilde{H}\tilde{t}}$ to exponentiation error $\varepsilon$, which requires\footnote{We assume that the error $\eps$ large enough that the number of calls to the helper function is larger than one.}
\begin{align}
\label{eq:HHKL_calls}
    \OO\!\left( \frac{L^\dims \tilde{t}}{ \log^{\dims}\left(L^\dims \tilde{t}/\eps \right)}\right)
\end{align}
calls to a helper function that implements the time evolution for a system of size
\begin{align}
\label{eq:HHLK_ell_def}
    \ell = \Theta\!\left(\log\left(L^\dims \tilde{t} / \varepsilon\right)\right)
    ,
\end{align}
with error 
\begin{align}
\label{eq:HHLK_delta_def}
    \delta = \OO\!\left( \frac{\varepsilon \log^\dims(L^\dims \tilde{t} / \varepsilon)}{L^\dims \tilde{t}}\right)
\end{align}
\end{theorem}
\begin{proof}
    The results follow directly from our discussion in ~\cref{sec:hhkl}.
    For formal proofs of the results stated there we refer the reader to the original paper~\cite{Haah:2018ekc}.
\end{proof}
 \section{Simulating Time Evolution on the Lattice\label{sec:lattice_algorithms}}
In this section, we apply the algorithms and methods of Section \ref{sec:review_alg} to the case of Hamiltonians defined in~\cref{sec:definitions}. 
Similar to
\Cref{sec:review_alg}, this section is organized by the techniques used (PFs, QSP, and HHKL).
For each of these sections, we begin by considering the case of $\OO(1)$-site-local RHLFTs defined on $\OO(1)$-many sites, which behave very similarly to generic quantum systems, and will serve as a review of the relationship between our notation and other common notations.
We then proceed to the case of $\loc$-site-local and $\loc$-geometrically-site-local Hamiltonians, whose simulation can roughly be categorized as being composed of simulation methods for $\OO(1)$-site-local Hamiltonians that are ``glued'' via PFs, LCU, or HHKL, depending on the flavor of locality whose preservation is desired.

\subsection{Product Formulas}\label{PF_NH}
Before we investigate the scaling of PFs for different local Hamiltonians with different numbers of lattice sites, we give two results that apply to all Hamiltonians considered in this work. 

The general results on PFs depended on a parameter $\tilde \alpha$, which was related to the commutator of terms in the Hamiltonian. Our first Lemma relates this to the norm of the Hamiltonian
\begin{lemma}
    \label{lemma:comm}
    Let H be a $\loc$-site-local Hamiltonian, as defined in \cref{def:loclattice}. Then the nested commutator $\Tilde{\alpha} \equiv \sum_{J_1,J_2,\dots,J_{p+1} \in S} \lvert\lvert [ H_{J_{p+1}}, \dots [H_{J_2},J_{J_1}] ] \rvert\rvert$ can be asymptotically bounded as
    \begin{equation}
        \Tilde{\alpha} = \OO\!\left( |||H|||_{1}^{p} ||H||_1 \right) .
    \end{equation}
\end{lemma}
\begin{proof}
    This result has been proven in Appendix G of \cite{Childs:2019hts}.
\end{proof}

The next Lemma will give an expression that will allow to combine a PF with another approximation to the exponentials that make up the PF. 
The PF introduced in~\cref{def:PF} is used as an approximation to $\me^{-i H t}$, and another approximation $U_\upsilon^{(J)}(t/r)$ is used to approximate each term of the form $\me^{-ita_{\upsilon}^{(J)}H_{J}/r}$.
\begin{lemma}
    \label{lemma:error_delta}
    Let $H$ be a site-local Hamiltonian as defined by ~\ref{def:loclattice}. 
    Let $S_p(t)$ be a $p^{\text{th}}$ order PF for $\me^{-i H t}$ as given in~\cref{def:PF}, and let the unitary $U_\upsilon^{(J)}(t/r)$ be an approximation to $\me^{-ita_{\upsilon}^{(J)}H_{J}/r}$ for all choices of allowed pairs $(\upsilon,J)$, with $\upsilon \in \{1,\dots,\Upsilon (p)\}$ and $J \in S$. 
    Define the PF $\tilde S_p(t)$ as the PF for $\me^{-i H t}$ obtained by replacing $\me^{-ita_{\upsilon}^{(J)}H_{J}/r}$ in~\cref{def:PF} with $U_\upsilon^{(J)}(t/r)$
    \begin{align}
        \tilde S_p(t/r) \equiv \prod_{\upsilon=1}^{\Upsilon(p)}\mathcal{P}_\upsilon\left(\prod_{J \in S}U_\upsilon^{(J)}(t/r)\right)
        \,.
    \end{align}
    If 
    \begin{align}
        \left\|U_\upsilon^{(J)}(t/r) - \me^{-ita_{\upsilon}^{(J)}H_{J}/r}\right\| = \OO\!\left( \frac{\varepsilon}{\NH r \Upsilon(p)} \right)
        ,
    \end{align}
    it is true that
    \begin{equation}
        \left\|\tilde S_p(t) - \me^{-iHt}\right\| = \OO(\varepsilon)\,.
    \end{equation}
\end{lemma}
\begin{proof}

Demanding that the new PF still reproduces the time evolution operator with the same exponentiation error we can write 
\begin{align}
\begin{aligned}
    & \left\| \me^{-i H t} - \tilde S_p(t) \right\|\\
    &= \left\| \me^{-i H t} - S_p(t) + \left[S_p(t/r)\right]^r - \left[\tilde S_p(t/r)\right]^r \right\|\\
    &\leq \OO(\varepsilon) + \left\|\left[\tilde S_p(t/r)\right]^r - \left[S_p(t/r)\right]^r \right\|\\
    & = \OO(\varepsilon) + r \, \NH \Upsilon(p) \left\|U_\upsilon^{(J)}(t/r) - \me^{-it/ra_{\upsilon J}H_{J}} \right\|\\
    & = \OO(\varepsilon)
    \,.
\end{aligned}
\end{align}
\end{proof}

\subsubsection{General site-local Hamiltonian\label{sssec:pfsl}}
We can now give results that apply to both general site-local Hamiltonians (see~\cref{def:loclattice}).
\begin{theorem}[Product Formulas for site-local Hamiltonians]\label{thm:PF_NH}
    Let $H$ be a site-local Hamiltonian as defined in~\cref{def:loclattice} or~\cref{def:geomloclattice}. 
    To simulate the time evolution operator $\me^{-iHt}$ up to exponentiation error $\mathcal{O}(\varepsilon)$ it suffices to choose a Trotter number $r$ such that:
    \begin{equation}\label{eq:PF_NH_r}
        r = \OO\!\left( \Nind\NH^{1/p}\frac{\left( \NP \cP t \right)^{1 + \frac{1}{p}}}{\varepsilon^{\frac{1}{p}}} \right).
    \end{equation}  Furthermore, if the Trotter number $r$ is chosen with tight the asymptotic scaling
    \begin{equation}\label{eq:r_omega}
        r = \Theta\!\left( \Nind \NH^{1/p}\frac{\left( \NP \cP t \right)^{1 + \frac{1}{p}}}{\varepsilon^{\frac{1}{p}}} \right),
    \end{equation}
    then the associated elementary gate count $\chi$ can be asymptotically bounded as
    \begin{equation}\label{eq:PF_NH_gates}
        \chi = \OO\!\left( \Nind \NP  \frac{\left( \NH\NP \cP t \right)^{1 + \frac{1}{p}}}{\varepsilon^{\frac{1}{p}}} \right).
    \end{equation}
\end{theorem}
\begin{proof}
Using~\cref{them:childs_pf,lemma:comm}one finds that the Trotter number $r$ is given by
\begin{align}
    r = \OO\!\left( \frac{|||H|||_{1} ||H||_1^{\frac{1}{p}}t^{1 + \frac{1}{p}}}{\varepsilon^{\frac{1}{p}}} \right)
            \,.
\end{align}
Using previously established notation, we have
\begin{equation}
\begin{aligned}
    ||H||_1 &= \mathcal{O}\left( \NH \sum_{J}||H_J|| \right) = \mathcal{O}( \NH \NP \cP )\,,\\
    |\|H\||_1 &= \OO(\Nind \NP \cP)
    \,,
\end{aligned}
\end{equation}
which immediately gives
\begin{equation}
\label{eq:r_2}
    r = \OO\!\left( \Nind \NH^{1/p}\frac{\left( \NP \cP t \right)^{1 + \frac{1}{p}}}{\varepsilon^{\frac{1}{p}}} \right)
    .
\end{equation}

    Assume now that we wish to compute the total gate count $\chi$ for a quantum circuit implementing an approximation of $\me^{-iHt}$ to total exponentiation error $\OO(\eps)$ based on the splitting $\left(S_p(t/r)\right)^r$, where the Trotter number $r$ in fact satisfies~\eqref{eq:r_omega} (i.e., the asymptotic scaling is tight) so that the spectral norm $\|\left( S_p(t/r) \right)^r - \me^{-iHt}\| = \OO(\eps)$. This splitting takes the following form (see~\cref{def:PF})
    \begin{equation}
        \Bigl[S_{p}(t/r)\Bigr]^{r}
        \equiv \left(\prod_{\upsilon=1}^{\Upsilon(p)}\mathcal{P}_\upsilon\!\left(\prod_{J \in \Set} \exp\left(\frac{-it a_{\upsilon}^{(J)}H_J}{r}\right)\!\right)\!\right)^{r}.
    \end{equation}
    Each $S_p(t/r)$ thus consists of $\mathcal{O}\left(\NH \Upsilon(p)\right)$-many exponentials of the form $\exp\left(-it a_{\upsilon}^{(J)}H_J/r\right)$. 
    Each of these exponentials can be approximated using a different PF $\tilde S_{\tilde p}^{(\upsilon, J)}(t/r\tilde{r})$ based on the Pauli decomposition of each term $H_J$ (see \cref{eq:HJ}). This splitting takes the form
    \begin{equation}
    \begin{alignedat}{9}
        \Bigl[&\tilde S_{\tilde p}^{(\upsilon, J)}(t/(r\tilde r))\Bigr]^{\tilde r}
        \\&\equiv \left(\prod_{\upsilon=1}^{\Upsilon(\tilde p)}\tilde{\mathcal{P}}_\upsilon\!\left(\prod_{i=1}^{\NPJ} \exp\left(\frac{-it \tilde a_{\upsilon}^{(i)} c_i^{(J)}P_i^{(J)}}{r \tilde r}\right)\!\right)\!\right)^{\tilde r}.
    \end{alignedat}
    \end{equation}
    From~\cref{lemma:error_delta} we know that to ensure that the total error remains $\OO(\varepsilon)$, the accuracy of this PF has to be 
    \begin{align}
    \begin{alignedat}{9}
        & \left\|\left(\tilde S_{\tilde p}^{(\upsilon, J)}(t/r\tilde r)\right)^{\tilde r} - \exp\left(\frac{-it a_{\upsilon}^{(J)}H_J}{r}\right)\right\| \\
        &= \OO\!\left(\frac{\varepsilon}{\NH r \Upsilon(p)}\right).
    \end{alignedat}
    \end{align}
    Substituting this into the result of ~\cref{them:childs_pf} we therefore find 
\begin{equation}
\begin{aligned}
    \tilde r &= \Omega\!\left(\left(\frac{\NP\cP t}{r}\right)^{1 + \frac 1{\tilde p}}\left(\frac{\varepsilon}{\NH \Upsilon(p) r}\right)^{-\frac{1}{\tilde p}}\right)\\
    & = \Omega\left(\frac{\Upsilon(p)}{\Nind} \left(\frac{\NP\cP t}{\NH \eps}\right)^{\frac{p}{\tilde p}}\right)
    .
\end{aligned}
\end{equation}
We choose $\tilde p = p$ and use that $\Upsilon(p) = \OO(1)$. We also have that $\tilde r, \Nind \in \mathbb{N}$ have to both be at least 1, and hence one finds $\tilde r = \Omega(1)$.

Using that $\Upsilon(\tilde p) = \OO(1)$, we have that the resource scaling is thus given by
\begin{equation}
\begin{aligned}
    \chi &= \OO(\NH \NP r)\\
    &= \mathcal{O}\!\left( \Nind \NP \frac{\left( \NH \NP \cP t \right)^{1 + \frac{1}{p}}}{\varepsilon^{\frac{1}{p}}} \right).
\end{aligned}
\end{equation}
\end{proof}

\subsubsection{\texorpdfstring{$\OO(1)$}{O(1)} site-local Hamiltonian}

Having obtained the results for a general site-local Hamiltonian, one can obtain the results for an $\OO(1)$ site-local Hamiltonian
\begin{theorem}[Product Formulas for $\NLat = \mathcal{O}(1)$ sites]
    \label{thm:PF_O_1}
    Let $H$ be an $\OO(1)$ Hamiltonian as defined in~\cref{def:O1loclattice}
    Let $S_p(t)$ be a $p^{\text{th}}$ order PF as given by~\cref{def:PF} and~\cref{eq:pth_order_PF}. 
    Then to simulate the time evolution operator $\me^{-iHt}$ up to exponentiation error $\mathcal{O}(\varepsilon)$ it suffices to choose a Trotter number $r$ such that:
    \begin{equation}
        r = \OO\!\left( \frac{\left( \NP \cP t \right)^{1 + \frac{1}{p}}}{\varepsilon^{\frac{1}{p}}} \right)
        .
    \end{equation}
    Furthermore, if the Trotter number $r$ is chosen with the tight asymptotic scaling
    \begin{equation}\label{eq:r_omega_O_1}
        r = \Theta\!\left( \frac{\left( \NP \cP t \right)^{1 + \frac{1}{p}}}{\varepsilon^{\frac{1}{p}}} \right)
        ,
    \end{equation}
    then the associated elementary gate count $\chi$ can be asymptotically bounded as
    \begin{equation}\label{eq:PF_O(1)_gates}
        \chi = \OO\!\left( \frac{ \NP \left(\NP \cP t \right)^{1 + \frac{1}{p}}}{\varepsilon^{\frac{1}{p}}} \right)
        .
    \end{equation}
\end{theorem}
\begin{proof}
    This result is obtained immediately by setting $\NH \sim \Nind = \OO(1)$ in~\cref{thm:PF_NH}. 
\end{proof}

\subsubsection{Geometrically-site-local Hamiltonian}

\begin{theorem}[Product Formulas for geometrically-site-local Hamiltonians]\label{thm:PF_NH_GEOM}
    Let $H$ be a geometrically site-local Hamiltonian as defined in~\cref{def:geomloclattice}. 
    To simulate the time evolution operator $\me^{-iHt}$ up to exponentiation error $\mathcal{O}(\varepsilon)$ it suffices to choose a Trotter number $r$ such that:
    \begin{equation}\label{eq:PF_NH_r_2}
        r = \OO\!\left( \NLat^{1/p}\frac{\left( \NP \cP t \right)^{1 + \frac{1}{p}}}{\varepsilon^{\frac{1}{p}}} \right).
    \end{equation}  Furthermore, if the Trotter number $r$ is chosen with tight the asymptotic scaling
    \begin{equation}\label{eq:r_omega_2}
        r = \Theta\!\left( \NLat^{1/p}\frac{\left( \NP \cP t \right)^{1 + \frac{1}{p}}}{\varepsilon^{\frac{1}{p}}} \right),
    \end{equation}
    then the associated elementary gate count $\chi$ can be asymptotically bounded as
    \begin{equation}\label{eq:PF_NH_gates_2}
        \chi = \OO\!\left( \NP  \frac{\left( \NLat\NP \cP t \right)^{1 + \frac{1}{p}}}{\varepsilon^{\frac{1}{p}}} \right).
    \end{equation}
\end{theorem}
\begin{proof}
    This result follows directly from~\cref{thm:PF_NH} making the replacements $\Nind = \OO(1)$ and $\NH = \OO(\NLat)$ as given in~\cref{eq:Nindgeomloc,eq:NHgeomloc}.
\end{proof}
\subsection{Quantum Signal Processing\label{ssec:qsp}}

As discussed in detail in~\cref{sec:QSP}, QSP provides a technique to take the BE of a given Hamiltonian and obtain the approximation to the BE of its time evolution operator. 
We will begin by deriving the general result of using QSP to implement the time evolution operator of the full system Hamiltonian. 

We will also derive results that combine the techniques of PFs with those of QSP.
In this approach we will use a PF to split the full time evolution operator into products of exponentials of the form $\me^{-ita_{\upsilon}^{(J)}H_{J}/r} $, and then use QSP to to obtain expression for these exponentials requiring only $\OO(1)$ lattice sites. 

In both of these approaches we assume that the original BE is achieved using LCU, as explained in~\cref{sec:BEtech}, but we note that other techniques are possible as well.

\subsubsection{General site-local Hamiltonian\label{sssec:qspsl}}

We first establish the scaling of the LCU that will be used for the BE
\begin{lemma}[LCU for site-local Hamiltonians]\label{thm:lcu2}
Let $H$ be a site-local Hamiltonian as as in~\cref{def:loclattice}. 
There exists a quantum circuit acting on $\n \equiv \NLat \nq$-many qubits together with $k = \OO\left(\log (\NH \NP)\right)$-many ancillas that produces an $\left(\OO\left(\NH \NP \cP\right),k\right)$-block-encoding of $H$ on the $\n$ qubits, and requires a total of $\OO\left(\NH \NP \log (\NH \NP)\right)$-many elementary quantum gates.
\end{lemma}
\begin{proof}
    This follows directly from~\cref{LCU}, if we observe that the total number of Pauli strings is given by $\NH \NP$ and the sum of the coefficients of all Pauli strings is bounded by $\beta < \NH \NP$. 
\end{proof}

Given this result now allows us to state the central QSP result we use in this work:
\begin{theorem}[QSP for site-local Hamiltonians]\label{thm:qspwq2}
Let $H$ be a site-local Hamiltonian as defined by ~\ref{def:loclattice}, and let $t$ and $\eps$ be positive values. 
Then there exists a quantum circuit that serves as a $\left(1\pm \OO(\eps),\OO(\log (\NH \NP))\right)$-block-encoding for $\me^{-iHt}$ up to exponentiation error $\eps$, and requires an elementary gate count of 
\begin{equation}
\label{eq:qspwq2} 
\chi = \OO\!\left[\NH \NP \log (\NH \NP)\left(\NH \NP \cP t + \log(1/\eps)\right)\right]. \end{equation}
\end{theorem}
\begin{proof}
    This result follows from~\cref{thm:lcu2,thm:qspwq}. From~\cref{thm:lcu2} one obtains that LCU give a $\left(\OO\left(\NH \NP \cP\right),\OO\left(\log (\NH \NP)\right)\right)$ BE of the Hamiltonian. Using this result in~\cref{thm:qspwq} one finds that QSP gives a $\left(1\pm \OO(\eps),\OO(\log (\NH \NP))\right)$-block-encoding for $\me^{-iHt}$. 
    Using the resource scaling from~\cref{thm:qspwq}, with the scale factor and resource scaling of LCU from~\cref{thm:lcu2}, gives the resource requirements stated in the theorem.
\end{proof}

\begin{remark}
Note that we have assumed that the block-encoding is formed from a sum over all Pauli strings appearing in the Hamiltonian. 
When more structure is known about a given field theory, the factor in the complexity due to the block-encoding may be reduced.
\end{remark}

\subsubsection{\texorpdfstring{$\OO(1)$}{O(1)} site-local Hamiltonian}
As for the PF, we can obtain the QSP requirements of an $\OO(1)$ site-local Hamiltonian:
\begin{theorem}[QSP for $\OO(1)$ sites]\label{thm:qspwq1}
Let H be an $\OO(1)$ Hamiltonian as defined in~\cref{def:O1loclattice}, and let $t$ and $\eps$ be positive values. Then there exists a quantum circuit that serves as a $\left(1\pm \OO(\eps),\OO(\log (\NP))\right)$-block-encoding for $\me^{-iHt}$ up to exponentiation error $\eps$, and requires an elementary gate count of \begin{equation}
\label{eq:qspwq1} 
\chi = \OO\!\left[\NP \log \NP\left(\NP \cP t + \log(1/\eps)\right)\right]. \end{equation}
\end{theorem}
\begin{proof}
This follows directly from~\cref{thm:qspwq2} by taking the limit $\NH = \OO(1)$. 
\end{proof}

\subsubsection{Geometrically-site-local Hamiltonian}
\begin{theorem}[QSP for geometrically-site-local Hamiltonians]\label{thm:qspwq3}
Let $H$ be a geometrically site-local Hamiltonian as defined by ~\ref{def:loclattice}, and let $t$ and $\eps$ be positive values. 
Then there exists a quantum circuit that serves as a $\left(1\pm \OO(\eps),\OO(\log (\NH \NP))\right)$-block-encoding for $\me^{-iHt}$ up to exponentiation error $\eps$, and requires an elementary gate count of \begin{equation}\label{eq:qspwq2_geom} 
\chi = \OO\!\left[\NLat \NP \log (\NLat \NP)\left(\NLat \NP \cP t + \log(1/\eps)\right)\right]. 
\end{equation}
\end{theorem}
\begin{proof}
    This result follows directly from~\cref{thm:qspwq2} making the replacements $\NH = \OO(\NLat)$ as given in~\cref{eq:NHgeomloc}.
\end{proof}
\subsection{Combining Product Formulas with QSP\label{ssec:pfqsp}}
In this section we consider combining a PF with QSP. 
We will use~\cref{lemma:error_delta} to write a PF for $\me^{-i H t}$, but then use QSP for $U_\upsilon^{(J)}(t/r)$, which serves as an approximation to each term of the form $\exp\left(-ita_{\upsilon}^{(J)}H_J/r\right)$ appearing in the PF. 

\begin{theorem}[QSP applied to Product Formulas]\label{thm:qspwqpf}
Let $H$ be a site-local Hamiltonian as defined by ~\ref{def:loclattice}, and let $t$ and $\eps$ be positive values. 
By using QSP (see Theorem \ref{thm:qspwq1}) to simulate each exponential term $\me^{-ia_v^{(J)}H_J t/r}$ appearing in the PF $S_p(t/r)$ to error $\eps/(r \NH)$, and taking
\begin{equation}\label{eq:r_ref}
r = \Theta\!\left(\Nind \NH^{1/p}\frac{\left( \NP \cP t \right)^{1 + \frac{1}{p}}}{\varepsilon^{\frac{1}{p}}} \right),
\end{equation}
we obtain a quantum circuit that serves as a $(1\pm\OO(\eps),\log(\NH\NP))$-block-encoding for the full time evolution operator $\me^{-iHt}$ to exponentiation error $\OO(\eps)$.
The gate count for this circuit is given by
\begin{equation}\label{eq:qspwqpf}
\begin{alignedat}{9}
\chi = \OO\!\Bigl[\Nind &\left(\NH \NP \cP t \right)^{1 + \frac{1}{p}}\varepsilon^{-\frac{1}{p}}\NP\log \NP
\\&\times
\log(\Nind \NH\NP \cP t/\eps)\Bigr].
\end{alignedat}
\end{equation}
\end{theorem}

\begin{proof}
From~\cref{lemma:error_delta} we know that the simulation of a single exponential of the form $\me^{-ia_v^{(J)}H_J t/r}$ acting on $\OO(1)$ sites is needed to accuracy $\eps/(\NH r)$. Performing this approximation using QSP (see Theorem \ref{thm:qspwq1}) requires a gate count
\begin{equation}
\begin{aligned}
\chi_J & = \OO\!\left(\NP \log \NP\left(\NP\cP \frac tr + \log(\NH r/\eps)\right)\right)\\
& = \OO\left(\NP\log \NP\log(\NH r/\eps)\right)
\,,
\end{aligned}
\end{equation}
where the first term can be discarded because the asymptotic assumption on $r$ given by \eqref{eq:r_ref} indicates that $\NP \cP t/r = \OO(\eps/(\NP \cP t)^{1/p})$, which is asymptotically bounded from above by a constant.

The total gate count from the $(\NH r)$-many applications of this QSP subroutine is then given by
\begin{equation}
\begin{aligned}
\chi &= \NH \, r \, \chi_J \\
&= \mathcal{O}\Bigl[\Nind \left(\NH \NP \cP t \right)^{1 + \frac{1}{p}}\varepsilon^{-\frac{1}{p}}\\
&\qquad \times
\NP\log \NP\log(\Nind \NH\NP \cP t/\eps)\Bigr]
\,.
\end{aligned}
\end{equation}
Since the block-encoding produced by QSP has a scale factor $1 \pm \OO(\eps/(r\NH))$, the final block-encoding produced by multiplying $(\stages(p)\NH r)$-many such block-encodings will again be a block-encoding, this time with scale-factor $1\pm \OO(\eps)$. This completes the proof.
\end{proof}

\begin{remark}
    The results above were for a general site-local Hamiltonian. To obtain the results for a geometrically site-local Hamiltonian one makes the replacement $\Nind = \OO(1)$ and $\NH = \OO(\NLat)$ from ~\cref{eq:Nindgeomloc,eq:NHgeomloc}.
\end{remark}

\subsection{HHKL\label{ssec:hhkl}}
As discussed in~\cref{sec:hhkl} the HHKL technique is a method of ``gluing'' together implementations of a time evolution operator on a ``local'' lattice with size $\ell^\dims$. In this section we consider HHKL with two approaches for the implementation of the local evolution operator, first using a PF and then QSP.

\begin{theorem}[HHKL with PF as helper]\label{thm:hhklpf}
Let $H$ be a geometrically site-local Hamiltonian as defined by ~\ref{def:geomloclattice}, and let $t$ and $\eps$ be positive values. Using the HHKL algorithm with a PF for the helper function gives rise to an algorithm that approximates the time evolution operator $\me^{-i H t}$ with exponentiation error requiring a gate complexity of 
\begin{align}\label{eq:hhklpf}
\chi = \NP \left( \NLat \NP \cP t \right)^{1+\frac{1}{p}} \varepsilon^{-\frac{1}{p}}
\,.
\end{align}
\end{theorem}

\begin{proof}
From~\cref{thm:hhkl_helper} we need the helper function to implement a time evolution operator for time $\tilde{t} = \|H_J\| t = \OO(\NP\cP t)$ on a system of size $N_\ell = \OO(\ell^d)$ with exponentiation error $\delta$, where $\ell$ and $\delta$ are given in~\cref{eq:HHLK_ell_def,eq:HHLK_delta_def}.

From~\cref{thm:PF_NH} we find that a PF to implement the time evolution on a system with $\Nell$ lattice sites and exponentiation error $\delta$ requires resources (taking $t = \OO(1)$ in~\cref{eq:qspwq2})
\begin{equation}
\begin{aligned}
    \chi_{\rm PF} &= \OO\!\left(\NP \left(\Nell \NP \cP\right)^{1+\frac{1}{p}} \delta^{-\frac{1}{p}}\right)\\
    &= \OO\!\Bigg(\NP \left(\NP \cP\right)^{1+\frac{1}{p}} \left[\log \left(\NLat t/\varepsilon\right)\right]^{\dims\left(1+\frac{1}{p}\right)}\\
    & \qquad \times\left(\frac{\NLat t}{\varepsilon \log^d\left(\NLat t/ \varepsilon\right)}\right)^{\frac{1}{p}}\Bigg)\\
    &= \OO\!\left(\NP \left(\NP \cP\right)^{1+\frac{1}{p}}\log^\dims\left(\frac{\NLat t}{\varepsilon}\right) \left( \frac{\NLat t}{\varepsilon}\right)^{\frac{1}{p}}\right).
\end{aligned}
\end{equation}
Furthermore, we know from~\cref{thm:hhkl_helper} that the number of calls to this helper function is  given by~\cref{eq:HHKL_calls}, such that the final gate complexity\ is
\begin{equation}
\begin{aligned}
    \chi & = \OO\!\left( \frac{\NLat t}{ \log^{\dims}\left(\NLat t/\eps \right)}\chi_{\rm PF}\right)\\
    &= \OO\!\left(\NP \left(\NLat \NP \cP t\right)^{1+\frac{1}{p}} \varepsilon^{-\frac{1}{p}}\right)
    .
\end{aligned}
\end{equation}
as desired. 
\end{proof}

\begin{theorem}[HHKL with QSP as helper]\label{thm:hhklqsp}
Let $H$ be a geometrically site-local Hamiltonian as given by ~\Cref{def:geomloclattice}, and let $t$ and $\eps$ be positive values. Using the HHKL algorithm with QSP for the helper function gives rise to an algorithm that approximates the time evolution operator $\me^{-i H t}$ with exponentiation error requiring a gate complexity of 
\begin{align}\label{eq:hhklqsp}
\chi = \OO\!\left(\NP^2 \NLat \cP t \log^\dims\left(\NLat t / \varepsilon\right) \log\left( \NP \log\left(\NLat t / \varepsilon\right)\right)\right)
\,.
\end{align}
\end{theorem}
\begin{proof}
    \Cref{thm:hhkl_helper} states that we need the helper function to implement a time evolution operator for time $\tilde{t} = \|H_J\| t = \OO(\NP\cP t)$ on a system of size $N_\ell = \OO(\ell^d)$ with exponentiation error $\delta$, where $\ell$ and $\delta$ are given in~\cref{eq:HHLK_ell_def,eq:HHLK_delta_def}.
    From~\cref{thm:qspwq2} we find that implementing the time evolution on a system with $\Nell$ lattice sites and exponentiation error $\delta$ with QSP requires resources (taking $t = \OO(1)$ in~\cref{eq:qspwq2}) 
    \begin{equation}
    \begin{aligned}
        \chi_{\rm QSP} &= \OO\! \left( \Nell \NP \log\left(\Nell\NP\right) \left( \Nell \NP \cP + \log(1/\delta)\right)\right]\\
        &= \OO \Bigg( \log^\dims\left(\NLat t / \varepsilon\right) \NP \log\left[ \NP \log\left(\NLat t / \varepsilon\right)\right] \\
        &\qquad \times \Bigg[ \NP \cP \log^\dims\left(\NLat t / \varepsilon\right) \\
        & \qquad \qquad +\log \left( \frac{\NLat t}{\varepsilon \log^\dims\left(\NLat t / \varepsilon\right)}\right) \Bigg]\Bigg)
        .
    \end{aligned}
    \end{equation}
    In the limit of large system size, the second term in the last square bracket (the one originating from the $\log(1/\delta)$ term) is subdominant to the first, which allows us to write
    \begin{align}
        \chi_{\rm QSP} &= \OO\! \left( \NP \cP \log^{2\dims}\left(\NLat t / \varepsilon\right) \log\left[ \NP \log\left(\NLat t / \varepsilon\right)\right]\right).
    \end{align}
    Using the number of calls to this helper function from~\cref{eq:HHKL_calls}, the final gate complexity is
    \begin{equation}
    \begin{aligned}
        \chi & = \OO\!\left( \frac{\NLat t}{ \log^{\dims}\left(\NLat t/\eps \right)}\chi_{\rm QSP}\right) \\
        &= \OO\!\left( \NP^2 \NLat \cP t \log^{\dims}\left(\NLat t / \varepsilon\right)\log\left[ \NP \log\left(\NLat t / \varepsilon\right)\right]\right)
        ,
    \end{aligned}
    \end{equation}
    as desired.
\end{proof}

\section{Lattice Quartic Theory\label{sec:quartic}}
The results we have discussed so far are theoretical worst-case asymptotic complexities, applicable to any lattice field theory with prescribed locality properties. 
In this section, we specialize to a particular model of interacting quantum field theories: the lattice scalar field theory.
From the computational perspective, this theory is similar to lattice gauge theories and so intuition can be drawn from its consideration. 
Furthermore, it is often used as a proverbial example in the studies of quantum simulation algorithms~\cite{Jordan:2011ci,Jordan:2017lea,Liu:2021otn}, and so detailed comparisons with past work can be performed.

In this section, we review the Hamiltonian of the theory and our procedure for digitizating the bosonic field.
Throughout the discussion, we will keep track of the parameters of the theory required for quoting the asymptotic gate complexities derived in Sec.~\ref{sec:lattice_algorithms}, e.g., $\NP$, $\cP$, $\NH$, etc.
We then determine the asymptotic gate complexities for implementing time evolution using PFs, QSP with two different BE methods, and the HHKL algorithm.
To determine the coefficients in these asymptotic gate complexities, we explicitly construct circuits that implement time evolution for a single site. We first compare the gate cost of constructing controlled calls to the walk operator $W_H$ obtained using the na\"ive and improved LCU procedures.
Then we compare the gate costs for time evolution simulated via 4th order PF and QSP.

\subsection{Introduction}

The model we consider is defined on a hypercubic lattice $\Lambda$ of spatial dimension $\dims$ with lattice spacing $a$. 
Assuming the simplest finite-difference approximation to the spatial derivative $(\nabla_\mu \hat\varphi)^2$, the Hamiltonian is given by\footnote{\label{foot:gen}The following discussion can be generalized to the case of higher-order finite-difference approximations to the derivative operator. 
One common approach is to first perform integration by parts to get the term $\hat{\varphi} \nabla^2 \hat{\varphi}$ and then use a symmetric approximation for the second derivative.}
\begin{subequations}
\label{eq:phi4}
\begin{alignat}{9}
    &\hat{H}_\text{quartic} = \hat{H}_\varphi + \hat{H}_\pi\,,
    \\
    \label{eq:Hphib}
    &\hat{H}_\varphi = a^\dims \left[ \sum_{i=1}^{\NLat} 
    \dfrac{m}{2} \hat\varphi_i^2 + \dfrac{\lambda}{4!} \hat\varphi_i^4 
    +
    \sum_{\langle i j \rangle} \dfrac{(\hat\varphi_j - \hat\varphi_i)^2}{2a^2}\right],
    \\
    &\hat{H}_\pi = a^\dims \sum_{i=1}^{\NLat} \dfrac{1}{2}\hat\pi^2_i\,,
\end{alignat}
\end{subequations}
where the indices $i$ and $j$ label lattice sites, and the notation $\langle ij\rangle$ runs over all distinct links joining nearest neighbors in the lattice.
In the remainder of this work, we will denote operators with a hat to differentiate between their digitized eigenvalues.

The Hamiltonian in~\cref{eq:phi4} is $2$-site-geometrically-local, which implies that $\Nind = \OO(1)$. 
We see that the Hamiltonian is a sum of $\NH = \OO(\NLat)$ terms.

The operators $\hat{\pi}_i$ and $\hat{\varphi}_i$ are conjugate operators satisfying $[\hat{\varphi}_j, \hat{\pi}_i] = i a^{-\dims} \delta_{ij}$.
This relation implies that the so-called field basis, in which the $\varphi_i$ operators are diagonal, is related to the so-called momentum basis, where the $\pi_i$ operators are diagonal, by the usual Fourier transform.
The parameters $m$ and $\lambda$ are bare quantities, and are in general complicated non-perturbative functions of the lattice spacing, \textit{i.e.} $m = m(a)$ and $\lambda = \lambda(a)$. 
In principle, this functional dependence could be included to determine the asymptotic complexity of a realistic physical calculation in terms of the lattice spacing. 
In order to keep our expressions as general as possible, however, we choose to leave the bare parameters $m$ and $\lambda$ as free variables in our quoted asymptotic complexities; this allows our results to be adapted in a straightforward way to any specific physical use case.
For the remainder of this work, we use units where the lattice spacing $a=1$.

The above Hamiltonian, although discretized on the lattice, is still an unbounded operator acting on an infinite dimensional Hilbert space. In order to obtain concrete time-evolution operators whose action can actually be simulated on a system of qubits, we must first truncate this Hamiltonian to a finite-dimensional operator.
This is achieved by a procedure known as \textit{digitization}, in which the Hilbert space at each lattice site is represented using $\nq$ qubits; the details of our digitization procedure are given in the next section.

\subsubsection{Digitization} \label{Digitization}

The basic idea is that we wish to construct a sequence of matrix operators $H_{\nq}$ that can act on $\nq$ qubits at a time, and whose spectrum converges to any prescribed single-site Hamiltonian $H$ as $\nq\to\infty$. 
There are numerous ways to proceed with such a construction~\cite{Macridin:2018gdw,Macridin:2018oli,Macridin:2021uwn}, but the one we follow is outlined in~\cite{Klco:2018zqz,Farrelly:2020ckc}.

The strategy is to \textit{digitize} the field on a given lattice site into $N = 2^{\nq}$ allowed field values
\begin{equation}
    \varphi \in \{-\varphi_{\mathrm{max}},-\varphi_{\mathrm{max}} + \delta \varphi, \dots, \varphi_{\mathrm{max}}-\delta\varphi, \varphi_{\mathrm{max}}\}\,, 
\end{equation}
where the spacing $\delta \varphi$ is given by
\begin{equation}
    \delta \varphi = \frac{2\varphi_{\mathrm{max}}}{N-1}\,.
\end{equation}

This restriction of the field values is to be understood as a restriction of the Hilbert space at the given lattice site, identifying the eigenvalues of the restricted field operator precisely with this set of allowed field values. 
In other words, the digitized field operator in the field basis at site $i$ becomes the diagonal matrix
\begin{equation}
\label{eq:phidig}
\hat{\varphi}_i = \mathrm{diag}\left(\varphi_{\mathrm{min}},\varphi_{\mathrm{min}} + \delta \varphi, \dots, \varphi_{\mathrm{max}}\right)_i\,.
\end{equation}
This single-site field operator acts nontrivially only on the $\nq$-qubit Hilbert space attributed to the site labeled by $i$.
Note that all lattice sites use the same digitized $\varphi$ values.

Because $\hat{\pi}_i$ and $\hat{\varphi}_i$ are conjugate operators, the momentum operator in the field basis can be written as
\begin{equation}
\label{eq:pidig}
\begin{aligned}
    \hat{\pi}_i &= \mathrm{FT} \cdot \hat{\pi}^{(p)}_i \cdot \mathrm{FT}^{-1}\,,  
    \\
    \hat{\pi}^{(p)}_i &= \mathrm{diag}\left(-\pi_{\mathrm{max}}, -\pi_{\mathrm{max}} + \delta \pi, \dots, \pi_{\mathrm{max}}\right)_i\,,
\end{aligned}
\end{equation}
where the superscript $(p)$ indicates the operator is represented in the momentum basis, $\mathrm{FT}$ denotes the \textit{symmetric Fourier transform}~\cite{Klco:2018zqz} and
\begin{equation}
    \pi_\text{max} = \frac{\pi}{\delta \varphi}, \quad \delta \pi = \frac{2\pi_\text{max}}{N-1}\,.
\end{equation}
Note that this choice implies anti-periodic boundary conditions in field space. 
Using this definition, the momentum part of the Hamiltonian can be written as
\begin{equation}
    \label{eq:Hpi}
    \hat{H}_\pi = \sum_{i=1}^{\NLat} \dfrac{1}{2}
    \mathrm{FT}_i
    \cdot
    \bigl(\hat{\pi}^{(p)}_i\bigr)^2
    \cdot
    \mathrm{FT}^{-1}_i\,.
\end{equation}

The final component of the digitization strategy is the choice of $\varphi_\text{max}$.
It was shown in~\cite{Klco:2018zqz} that, for a given value of $\nq$ and $\lambda$, there is an optimal choice of $\varphi_\text{max}$ that minimizes the digitization error of the low lying spectrum of the digitized Hamiltonian.
The procedure for determining $\varphi_\text{max}$ involved scanning over many values to minimize the digitization error.
While this procedure is in general possible, an explicit formula for the optimal value of $\varphi_\text{max}$ for the $\lambda=0$ case was given in~\cite{Bauer:2021gup,Macridin:2018gdw,Macridin:2021uwn}.
As argued in~\cite{Klco:2018zqz}, because $\lambda \neq 0$ results in a smooth deformation of the wavefunction, this choice of $\varphi_\text{max}$ is expected to work well even in the $\lambda \neq 0$ case. 
Therefore, we use the value 
\begin{equation}
    \varphi_{\mathrm{max}} = 2^{\nq}\sqrt{\frac{2\pi}{2^{\nq + 1} + 1}}\,,
\end{equation}
as in \cite{Bauer:2021gek}.

Altogether, these definitions allow for the construction of $2^{\nq}\times 2^{\nq}$ Hermitian matrices that digitize each of the single-site terms in Hamiltonian \eqref{eq:phi4} into operators that can act on $\nq$ qubits at a time.
An appropriate full digitization of Hamiltonian \eqref{eq:phi4} is then obtained by directly substituting the definitions of the digitized single-site operators $\hat{\varphi}_i$ and $\hat{\pi}_i$ into \eqref{eq:phi4}, together with the link terms $\hat{\varphi}_i \hat{\varphi}_j = \hat{\varphi}_i \otimes \hat{\varphi}_j.$

We are now in a position to discuss the parameters of this theory necessary for determining the cost of the various methods described in Sec.\ref{sec:methods_summary}.
Recall that we found $\NH = \OO(\NLat)$ and $\Nind = \OO(1)$. 
The final parameters needed are $\NP$ and $\cP$, which we now discuss.

To determine the maximum number of Pauli strings $\NP$ appearing in any term in $\hat{H}$, we start by writing the various operators in terms of Pauli operators. 
Note that the $\nq$ qubit operators $\hat{\varphi}_i$ and $\hat{\pi}_i^{(p)}$ are diagonal operators with evenly spaced eigenvalues.
It was shown in~\cite{Klco:2018zqz} that operators of this form can be decomposed into $\nq$ single Pauli-Z gates.
For example, one has $\hat{\varphi} = \frac{\varphi_\text{max}}{2^{\nq}-1} \sum_{j=0}^{\nq-1} 2^j Z_j$, where $Z_j$ is a single Pauli-Z gate acting on qubit $j$. 
Using this relation, we see that $\NP$ for the various terms in $H_\varphi$ are
\begin{subequations}
\begin{align}
    \NP^{\left(\varphi_i^2\right)} &= \OO( \nq^2 )\,,
    \\
    \NP^{\left(\varphi_i^4\right)} &= \OO(\nq^4)\,,
    \\
    \NP^{\left(\varphi_i \varphi_j\right)} &= \OO( \nq^2 )\,.
\end{align}
\end{subequations}

Turning now to $\hat{H}_\pi$, while the operator $(\hat{\pi}_i^{(p)})^2$ is a sum of $\OO(\nq^2)$ gates, the same is not true for $\hat{\pi}_i^2 = \text{FT}\cdot (\hat{\pi}_i^{(p)})\cdot \text{FT}^\dagger$; 
after using the Fourier transform to change to the field basis, the operator $\hat{\pi}_i^2$ is dense and requires $\OO(4^{\nq})$ Pauli strings. 
Since the quantum Fourier transform requires $\OO(\nq^2)$ gates, it is more cost effective to first decompose $(\hat{\pi}_i^{(p)})^2$ into $\OO(\nq^2)$ Pauli strings, and then change to the field basis using the Fourier transform. 
Using this method results in
\begin{equation}
    \NP^{(\pi^2)} = \OO(\nq^2)\,.
\end{equation}
From this we see that $\NP$ for the entire Hamiltonian is dominated by the $\hat{\varphi}^4$ term, and is given by
\begin{equation}
    \NP = \OO(\nq^4) \,.
    \label{eq:np_sft}
\end{equation}
This result for $\NP$ is significantly more favorable than the worst case scaling for an arbitrary 2-site-local Hamiltonian of $\OO(4^{2\nq})$.

The final piece is to determine the maximum coefficient $\cP$ of any Pauli string in $\hat{H}$. 
Using the fact that our choice of digitization leads to $\pi_\text{max} = \OO(2^{\nq/2})$ and $\varphi_\text{max} = \OO(2^{\nq/2})$, we see
\begin{align}
    \cP^{(\pi^2)} &= \OO(\pi_\text{max}^2) = \OO(2^{\nq}),
    \\
    \cP^{(\varphi^2)} &= \OO(m^2 \varphi_\text{max}^2) = \OO(m^2 2^{\nq}),
    \\
    \cP^{(\varphi^4)} &= \OO(\lambda \varphi_\text{max}^4) = \OO(\lambda 4^{\nq})\,,
    \\
    \cP^{(\varphi_i \varphi_j)} &= \OO(\varphi_\text{max}^2) = \OO(2^{\nq})\,.
\end{align}
As was the case with $\NP$, the overall value of $\cP$ is dominated by the $\hat{\varphi}^4$ term, and is given by
\begin{equation}
    \cP = \OO(\lambda 4^{\nq})\,.
    \label{eq:cp_sft}
\end{equation}

These values for $\NP$, $\cP$ and $\NH$ will be used in the following sections to determine the asymptotic gate complexity of the various time-evolution algorithms discussed in Sec.~\ref{sec:methods_summary} for geometrically site-local Hamiltonians.

\subsection{Product Formulas}

In this section, we consider using a $p^\text{th}$ order PF (introduced in Sec. \ref{PF}) to time evolve the scalar field theory Hamiltonian (given by~\cref{eq:phi4}) that has been digitized using the formalism presented in Section \ref{Digitization}. 

Specifically, we construct concrete error bounds and asymptotic gate complexity estimates as a special case of the results obtained in Sec.~\ref{PF_NH} for general geometrically site-local Hamiltonians.
We use the result of Theorem~\ref{thm:PF_NH_GEOM} for the asymptotic bound on the elementary gate count $\chi$ (see \cref{eq:PF_NH_gates_2}) and make the following substitutions based on the discussion provided in Section \ref{Digitization}:
\begin{align}
     \NH &= \OO(\NLat), \\
     \NP &= \OO(\nq^4), \\
     \cP &= \OO(\lambda 4^{\nq}).
\end{align}
These follow from \cref{eq:NHgeomloc}, \cref{eq:np_sft}, and \eqref{eq:cp_sft} respectively. Putting this all together, the asymptotic gate complexity is thus given by

\begin{align}
    \chi = \OO\!\left( \nq^{4\left(2 + \frac{1}{p}\right)} \frac{\left( 4^{\nq} \NLat \lambda t \right)^{1 + \frac{1}{p}}}{\varepsilon^{\frac{1}{p}}} \right).
    \label{eq:chi_sft_trotter}
\end{align}

\subsection{Quantum Signal Processing}
\label{ssec:phi4_qsp}

Obtaining gate counts for the case of pure QSP requires two logical steps.
The first is to determine the number of calls to the local block encoding, and the second is to estimate the complexity of implementing this block encoding.
By Theorem \ref{thm:qspwq}, the query complexity for simulating with evolution time $t$ to exponentiation error $\eps$ is given by
\begin{equation}
\label{eq:asymptotic_qsp}
\OO(\alpha t + \log(1/\eps))\,,
\end{equation}
where $\alpha$ is the scale factor of the BE used.

We now consider two approaches to implement the local block encoding.
The first approach is based on the fully general LCU procedure, and does not take into account any features of the model.
The second approach is a modification of the former which takes advantage from splitting the Hamiltonian into the field and momentum parts, in analogy with how it was done in the product-formula-based approach.
In this section, we consider only asymptotic gate complexities; detailed gate count comparisons of these techniques will be given in~\cref{sec:numerics}.

\subsubsection{Block encoding via LCU}
\label{sssec:lcu_vanilla}

In this section, we determine the cost of simulating time-evolution via QSP when the BE is implemented using a fully general LCU procedure.
This involves first digitizing the Hamiltonian $\hat{H}$ and then decomposing it into Pauli operators.
In such a scenario, the dominant cost arises from the $\hat \pi_i^2$ operator, which is represented by $\OO(3^{\nq})$ Pauli strings (excluding $Z$ gates in all combinations).
The value of $\cP$ in this case is the same as in \cref{eq:cp_sft}.
Lastly, we have $\NH = \OO(\NLat)$.
Plugging these results into \cref{thm:qspwq2} leads to gate complexity
\begin{equation}
    \chi = \OO\!\left[\NLat 3^{\nq} \log\left(\NLat 3^{\nq}\right) \left(\NLat 3^{\nq} (\lambda 4^{\nq})t + \log(1/\eps)\right)\right].
\end{equation}
We see that the naive decomposition of the $\hat{\pi}_i^2$ operators into Pauli gates results in a poor scaling with~$\nq$.
\begin{figure*}[t]
\begin{quantikz}
    \lstick{$\sket{}_b$} & \ctrl{1} & 
    \\
    \lstick{$\sket{}_a$} & \gate[2]{W_H} & 
    \\
    \lstick{$\sket{}_s$} & &
\end{quantikz}
=
\begin{quantikz}
    \lstick{$\sket{}_b$} & \ctrl{1} & \ctrl{1}&
    \\
    \lstick{$\sket{}_a$} & \gate[2]{\SELECT} & \gate{R_\PREPARE} &
    \\
    \lstick{$\sket{}_s$} & & &
\end{quantikz}
=
\begin{quantikz}[row sep={0.8cm,between origins}, column sep=0.25cm]
    \lstick[1]{$\sket{}_b$} & \ctrl{1} & & & \ctrl{1} & & & \ctrl{4} &
    \\
    \lstick[4]{$\sket{}_a$} & \gate[5]{\SELECT} & \gate[4]{\PREPARE^\dagger} & & \octrl{1} & & \gate[4]{\PREPARE} & \qw & 
    \\
    \vdots \ \setwiretype{n} & & & & \ \vdots \ & & & & & 
    \\
    & & & & \octrl{1} \wire[u][1]{q} & & & \qw & 
    \\
    & & & \gate{\Xgate} & \gate{\Zgate} & \gate{\Xgate} & & \gate{-1} &
    \\
    \lstick[1]{$\sket{}_s$} & & & & & & & &
\end{quantikz}
\caption{Circuit for controlled call to $W_H$ when the BE is prepared using LCU. Note that $\sket{g}_a = \PREPARE \sket{0}_a$ and so the reflector $R_\PREPARE$ (see Fig.~\ref{fig:Rg}) uses $G = \PREPARE$.}
\label{fig:ctrl_WH_lcu}
\end{figure*}
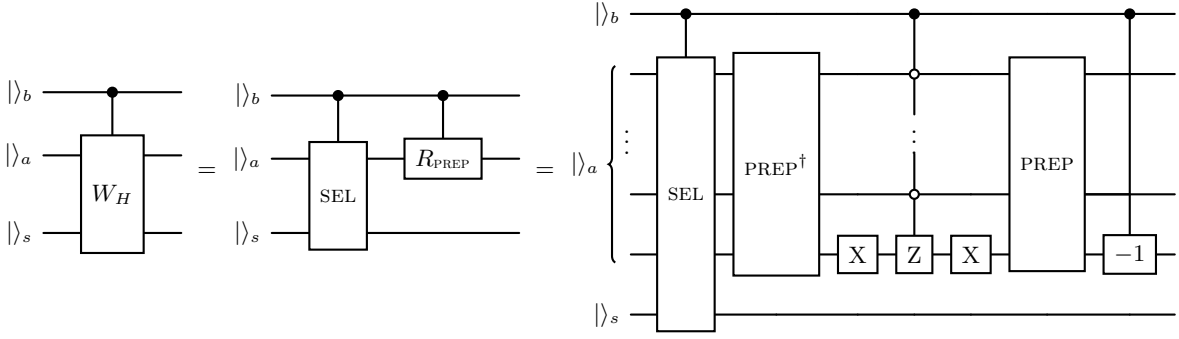
By exploiting the fact that $\SELECT^2 = \unit$, the walk operator $W_H$ can then be constructed without ancillary qubits using Eq.~\eqref{eq:WH_for_lcu}, as was discussed in~\Cref{sssec:qubitization}.

The final step is to construct a controlled call to $W_H$.
Rather than naively controlling every gate in $W_H$, one can achieve this by only controlling the $\SELECT$ oracle and the multi-controlled $Z$ gate in the reflector $R$; the $\PREPARE$ oracle does not need to be controlled~\cite{Babbush:2018ywg}. 
This can be understood by noting that if one wanted to implement a controlled call to some operator $U = A^\dagger B A$ where $A^\dagger A = \unit$, one only has to control $B$. 
The circuit for the controlled implementation of $W_H$ is shown in Fig.~\ref{fig:ctrl_WH_lcu}.

\subsubsection{Block encoding via LCU and FT}
\label{sssec:lcuft}
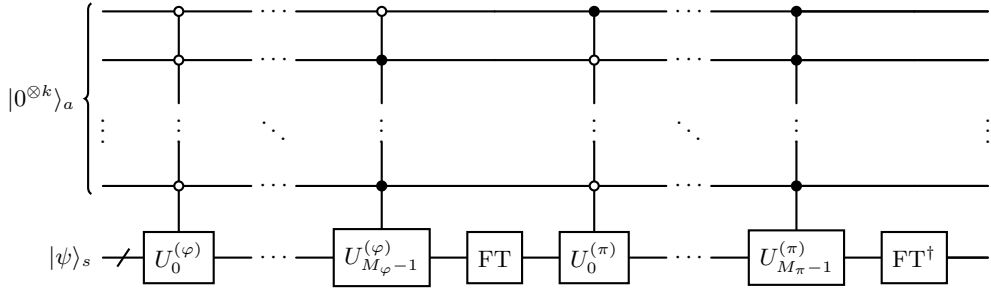
\begin{figure*}[t]
\begin{quantikz}
\lstick[4]{$\sket{0^{\otimes k}}_a$}
& \octrl{1} 
& \ \cdots \ & \octrl{1}  
& & \ctrl{1} & \ \cdots \ &\ctrl{1}
& \qw & \qw
\\
& \octrl{1} & \ \cdots \ & \ctrl{1}
& & \octrl{1} & \ \cdots \ &\ctrl{1}
& \qw & \qw
\\
\vdots\setwiretype{n}& \ \vdots \ & \ \ddots \  &  \ \vdots \ 
& \phantom{ \ \ddots \ }  \setwiretype{n} & \ \vdots \ & \ddots & \ \vdots \ &
& \vdots &
\\
& \octrl{1} \wire[u][1]{q} & \ \cdots \ & \ctrl{1} \wire[u][1]{q} 
& & \octrl{1} \wire[u][1]{q} & \ \cdots \ & \ctrl{1} \wire[u][1]{q}
& \qw & \qw
\\
\lstick{$\sket{\psi}_s$} 
& \gate{U_0^{(\varphi)}} \qwbundle{} & \ \cdots \ & \gate{U^{(\varphi)}_{M_\varphi-1}}
& \gate{\vphantom{U_1^(}\text{FT}} & \gate{U_{0}^{(\pi)}} & \ \cdots \ & 
\gate{U_{M_\pi-1}^{(\pi)}}
& \gate{\vphantom{U_1^(}\text{FT}^\dagger} & \qw
\end{quantikz}
\caption{Circuit for \SELECT oracle in~\cref{eq:selectlcuft}. As compared to $\OO(3^{\nq})$ and $\OO(\nq)$ in the usual LCU implementation from~\cref{sssec:lcu_vanilla}, the usage of Fourier Transform described in~\cref{sssec:lcuft} allows one to reduce the circuit depth and the number of ancillary qubits to $\OO(\nq^4)$ and $\OO(\log \nq)$, correspondingly.
Note that while in the usual LCU procedure the gate complexity $\OO(3^{\nq})$ stems from the number of terms in $\hat{H}_\pi$, upon adding the Fourier Transform to the circuit, the $\OO(\nq^4)$ complexity stems from the $\hat{\varphi}^4$ term in $\hat{H}_\varphi$.
}
\label{fig:BEFT}
\end{figure*}

In order to take advantage of the decomposition~\eqref{eq:Hpi} within the LCU construction, one needs to slightly modify the \PREPARE and \SELECT oracles.
To write $\hat{H}$ in the form required for LCU, we first write the momentum component of $\hat{H}$ in the momentum basis as a sum of $M_\pi$ unitary operators
\begin{equation}
    \hat{H}_\pi^{(p)} = \sum_{j=0}^{M_\pi} \beta_j^{(\pi^{(p)})} \hat{U}_j^{(\pi^{(p)})}\,.
\end{equation}
Next, we write the full Hamiltonian as
\begin{equation}
    \hat{H} = \sum_{i=0}^{M_\varphi-1} \beta_i^{(\varphi)} \hat{U}^{(\varphi)}_j + \text{FT} \left(\sum_{j=0}^{M_\pi-1} \beta_j^{(\pi^{(p)})} \hat{U}_j^{(\pi^{(p)})}\right) \text{FT}^\dagger\,,
\end{equation}
where it is understood that the Fourier transform is performed locally at each site. 
The Hamiltonian is now in a form suitable for the modified LCU procedure. 

First, the \PREPARE oracle is used to encode the real positive coefficients $\beta_i^{(\varphi)}$ and $\beta_j^{(\pi^{(p)})}$.
The new \SELECT  oracle is shown in~\cref{fig:BEFT} and is given by
\begin{align}
\label{eq:selectlcuft}
\begin{aligned}
&\SELECT = (\unit_a \otimes \text{FT})
\\
& \otimes\left( \sum_{j=0}^{M_\pi-1} \sket{j}_a\sbra{j}_a\otimes \hat{U}_j^{(\pi^{(p)})}\right)(\unit_a \otimes \text{FT}^\dagger)
    \\
    &+\sum_{i=0}^{M_\varphi-1} \sket{i+M_\pi}_a\sbra{i+M_\pi}_a \otimes \hat{U}_i^{(\varphi)}\,,
\end{aligned}
\end{align}
where, again, it is understood that the Fourier transform is performed locally at each site.
It is important to note that the modified $\SELECT$ operator still satisfies $\SELECT^2=\unit$, which enables the use of the efficient construction of $W_H$.

Note that the \SELECT oracle prepared using this method encodes the same operators as the \SELECT oracle prepared using the naive LCU procedure in Sec.\ref{sssec:lcu_vanilla}; the only difference is how it is implemented.
This different implementation, however, results in a more favorable scaling for $\NP$.
This can be understood by noting that the $\NP$ associated with the momentum Hamiltonian is reduced to $\NP^{(\pi^{(p)})} = \OO(\nq^2)$, implying that the value of $\NP$ is dominated by the $\hat{\varphi}^4$ term. 
Substituting $\NP = \OO(\nq^4)$, $\cP = \OO(\lambda 4^{\nq})$ and $\NH = \OO(\NLat)$ into the result in ~\cref{thm:qspwq2} gives the gate complexity
\begin{equation}
\label{eq:lcu_FT_be}
    \chi = \OO\!\left[\NLat \nq^4 \log(\NLat \nq^4) \left(\NLat \nq^4 (\lambda 4^{\nq}) t + \log(1/\eps) \right) \right],
\end{equation}
which has a more favorable scaling with $\nq$ than the naive LCU implementation in Sec.~\ref{sssec:lcu_vanilla}.
Since $\SELECT^2=\unit$, the controlled call to $W_H$ can be constructed in the same way as the naive LCU case in Sec.~\ref{sssec:lcu_vanilla}.

\subsection{HHKL}

In order to apply the HHKL results formulated in our conventions in~\cref{thm:hhklpf} and~\cref{thm:hhklqsp}, note that we must scale down the Hamiltonian \eqref{eq:phi4} by the appropriate factor that brings the local terms to at most unit spectral norm, and simultaneously rescale evolution time up by the same factor.
This factor is precisely the maximum spectral norm of any local term $\tilde{H}_J$,
\begin{align}
\begin{aligned}
    &\max_{J\in S} \| \tilde{H}_J \|
    \\
    &= \left\|\frac 12 \hat{\pi}_i^2 + \frac 12 m^2\hat{\varphi}_i^2 + \frac{\lambda}{4!} \hat{\varphi}_i^4 + \sum_{j\in\langle ij\rangle} \frac{(\hat{\varphi}_j-\hat{\varphi}_i)^2}{2}\right\|
    \\
    & = \mathcal O(\NP \cP) = \OO(\nq^4 (\lambda 4^{\nq}))\,,
\end{aligned}
\end{align}
where we have used $\NP = \OO(\nq^4)$ and $\cP = \OO(\lambda 4^{\nq})$, and where any lattice site with index $i$ in the interior of the lattice (or on the boundary if there are no sites in the interior) can be chosen to compute the spectral norm, and the sum runs over all lattice sites $j$ neighboring $i$.
We therefore make the rescalings
\begin{align}
\tilde{H} &= \frac{H}{\max_{J\in S} \| H_J \|}\\
\tilde{t} &= t \max_{J\in S} \| H_J \| = \OO(4^{\nq} \nq^4 \lambda t)
\,.
\end{align}

For HHKL built out of local Trotter operations, Theorem~\ref{thm:hhklpf} implies a gate complexity
\begin{equation}
    \chi = \OO\!\left[ \nq^4 (\NLat \nq^4 (\lambda 4^{\nq}) t)^{1+\frac{1}{p}} \eps^{-\frac{1}{p}} \right]\,.
\end{equation}
where we substituted $\cP = \OO(\lambda 4^{\nq})$ and $\NP = \OO(\nq^4)$.

Similarly, for HHKL built out of local QSP operations, Theorem~\ref{thm:hhklqsp} then implies a gate complexity
\begin{equation}
\begin{split}
    \chi &= \OO \Big[(\nq^4)^2 \NLat (\lambda 4^{\nq}) t \log^\dims\left(\NLat t / \varepsilon\right)
    \\
    &\hspace{0.8in} \times \log\left( \nq^4 \log\left(\NLat t / \varepsilon\right)\right) \Big].
\end{split}
\end{equation}

To conclude our discussion on HHKL scalings, we again stress an important point.
As we have previously remarked, while the purely asymptotic scaling of the Trotter-based HHKL gate count is worse than that for QSP-based HHKL, the difference is not striking when the Trotter order $p$ is taken to be large. 
This allows for the possibility of lower constant pre-factors due to the smaller overhead involved in implementing Trotter operations compared to QSP operations, and may even be a prudent choice for a more practical implementation on near-term devices, if the exact gate counts follow suit before asymptotics take over.

\subsection{Numerical results.}
\label{sec:numerics}

In the previous section we worked out the asymptotic gate count scaling for several time-evolution algorithms.
While these expressions are helpful for identifying which methods have the best scaling in a particular parameter, the practical question of the size of the prefactors requires a dedicated numerical study, which we perform in this section.
In particular, we first compare the cost of the two BE methods described in the previous section.
For the comparison, we choose to compare the cost of constructing a single controlled call to the $W_H$ operator. This choice allows a more direct comparison of the cost of each method as in some cases significant circuit simplifications occur compared to naively controlling every gate in the $W_H$ operator.
From there, we turn to time evolution and compare the cost of using QSP to that of a 4th order PF, focusing on the scaling with time $t$ and error $\eps$.

We choose to focus on a single site Hamiltonian given by
\begin{equation}
    \label{eq:single_site_phi4}
    \hat{H} = \frac{1}{2} \hat{\pi}^2 + \frac{1}{2} \hat{\varphi}^2 + \frac{\lambda}{4!} \hat{\varphi}^4\,,
\end{equation}
where we set $\lambda = 32$ (similarly to \cite{Klco:2018zqz}) as an instance of the strongly coupled regime where perturbation theory breaks down.
While this will not provide an explicit scaling with the number of lattice sites, analyzing the single site case in detail will provide intuition for what values of $t$ and $\eps$ QSP outperforms PF methods for $\NLat > 1$.

Before comparing gate counts of the various methods, we describe the details of constructing the circuits.
All gate counts were found using QISKIT's transpile function with an optimization level of 1; higher optimization levels result in the dropping of gates with small arguments and were therefore not used. 
We use the universal gate-set containing $R_x$, $R_z$, and CNOT gates.
When implementing multi-controlled gates, we compared the cost using two different methods. 
The first uses QISKIT's built-in functionality to decompose the multi-controlled gates into $\OO(M^2)$ CNOT gates, where $M$ is the number of qubits being controlled on.
The second method uses $\OO(M)$ additional work qubits to decompose the multi-controlled gates into $\OO(M)$ gates~\cite{Nielsen:2012yss}.
For each value of $\nq$, we compare the two and choose the method that results in the smallest number of gates.
Lastly, it is important to note that, when implementing multi-controlled gates, global phases must also be controlled. 
Because using LCU requires the coefficients in the $\PREPARE$ oracle to be real and positive, any overall phase must necessarily be absorbed into the individual operators in the $\SELECT$ oracle.
Controlling these phases requires an additional
controlled single-qubit rotation.
Alternatively, one could use the generalized LCU method to allow the coefficients in the $\PREPARE$ oracle to be complex and set the phase of the operators in the $\SELECT$ oracle to have zero overall phase (see e.g.~\cite{Lin:2022vrd}).
While this avoids having to control the phases, it leads to a more complicated $\PREPARE$ oracle.
Additionally, because the generalized LCU method uses circuits of the form $\widetilde{\PREPARE} \times \SELECT \times \PREPARE$ with $\widetilde{\PREPARE} \times \PREPARE \neq \unit$, one will have to also control both $\PREPARE$ and $\widetilde{\PREPARE}$ in the controlled $W_H$ circuits in Fig.~\ref{fig:ctrl_WH_lcu}. 
For these reasons, we use the former method to control the overall phases.

We now compare the cost of implementing the controlled $W_H$ operator using the two different BE methods. 
The top and middle plots in Fig.~\ref{fig:BEs} show the CNOT gate count and rotation gate count, respectively, as a function of $\nq$. 
The blue points show the gate counts using the LCU method where one naively decomposes the full digitized Hamiltonian into Pauli strings as described in \cref{sssec:lcu_vanilla}.
As expected, this method scales least favorably with $\nq$. 
As previously discussed in \cref{sssec:lcuft}, this poor gate scaling with $\nq$ can be improved by utilizing the Fourier Transform to implement the momentum component of the Hamiltonian.
This improved scaling can be seen by looking at the orange squares in the top and middle plots in Fig.~\ref{fig:BEs}.

The bottom plot in Fig.~\ref{fig:BEs} shows the number of ancillary qubits required to implement a controlled call to $W_H$ for each method.
As expected from the $\OO(\nq)$ scaling, using the naive LCU method requires the most ancillary qubits.
Using LCU combined with the FT reduces the scaling to $\OO(\log \nq)$.
Note that, for both methods, the overall coefficient in the number of ancillary qubits is large due to needing additional work qubits to efficiently implement multi-controlled gates.

Taken together, the gate count and ancillary qubit count comparisons show that using LCU combined with the FT is superior to the naive LCU method for all values of $n_q$.

\begin{figure}
    \centering
    \includegraphics[width=0.42\textwidth]{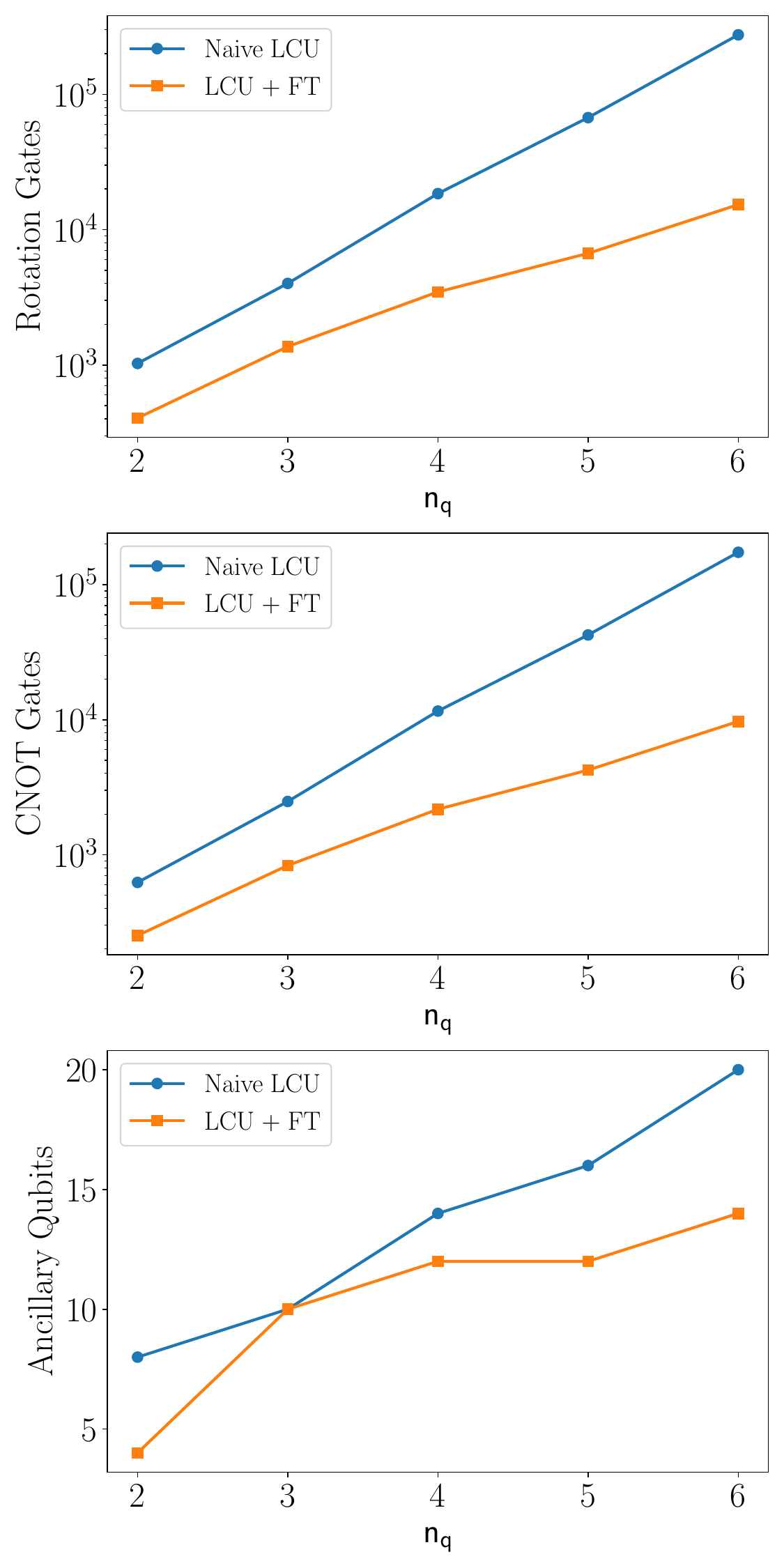}
    \caption{
    The top, middle, and bottom plots show the CNOT gate count, rotation gate count, and number of ancillary qubits needed to prepare a controlled call to the $W_H$ operator corresponding to the single site Hamiltonian in Eq.~\ref{eq:single_site_phi4} as a function of $\nq$.
    Blue circles correspond to using the naive LCU method described in Sec.~\ref{sssec:lcu_vanilla}.
    Orange squares correspond to using LCU combined with the Fourier Transform as described in Sec.~\ref{sssec:lcuft}.
    Using LCU combined with the FT results in an exponential improvement with $\nq$ in both the gate count and number of ancillary qubits required.
    }
    \label{fig:BEs}
\end{figure}

Now that we have identified the optimal method for preparing BEs, we turn to the cost of performing time evolution.
The first method we implement is a 4th order PF.
The second uses QSP, where the BE is prepared using the BE method that minimizes the number of rotation gates, \emph{i.e.}, using LCU combined with the FT.
The top (bottom) plot in Fig.~\ref{fig:cnot_3d} show the rotation (CNOT) gate count for both methods as a function of evolution time $t$ and error $\eps$ for $\nq = 3$.
The value of the scale factor $\alpha$ for the QSP data is $\alpha = 134.5$.
Due to QSP's large overall prefactor from implementing the BE, using the 4th order PF requires less rotation and CNOT gates for errors $\eps~\gtrsim~10^{-9}$.
For errors $\eps \lesssim 10^{-9}$, QSP wins due to its exponentially better scaling with respect to $\eps$.

It is interesting to note that the QSP gate count scaling is almost independent of the desired error $\eps$. 
This is due to the fact that the $\alpha t$ term in the asymptotic cost is significantly larger than the $\log(1/\eps)$ term (see e.g. Eq.~\eqref{eq:asymptotic_qsp}).
These findings therefore show that the additional cost of implementing time evolution to floating point precision using QSP is negligible compared to implementing time evolution approximately (not accounting for the increased cost from requiring higher precision rotation gates).
In \cref{sec:discussion} we discuss the implications of these results for multi-site systems.

\begin{figure*}
    \centering
    \includegraphics[width=0.9\textwidth]{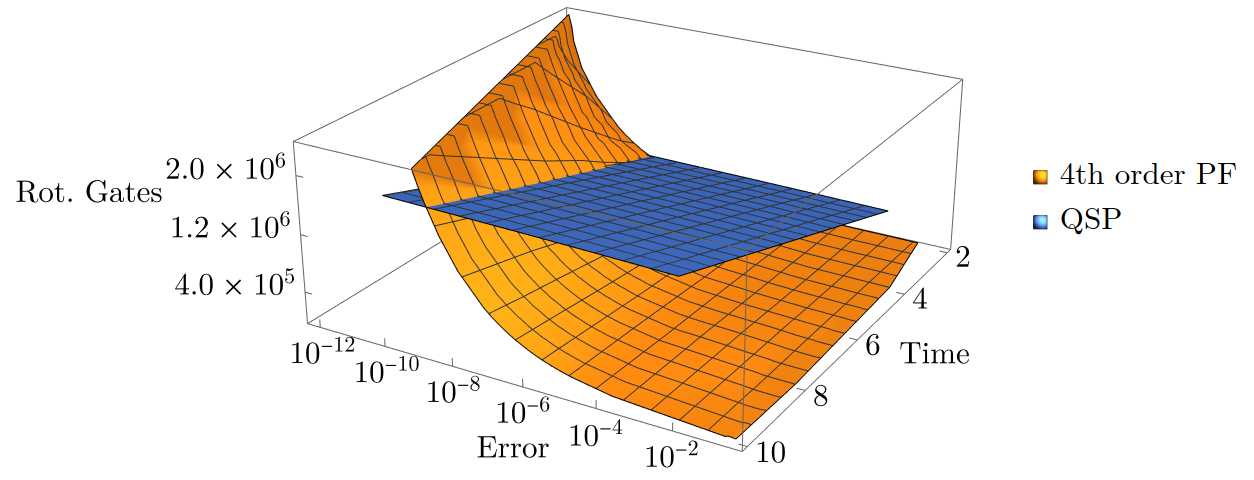}
    \includegraphics[width=0.9\textwidth]{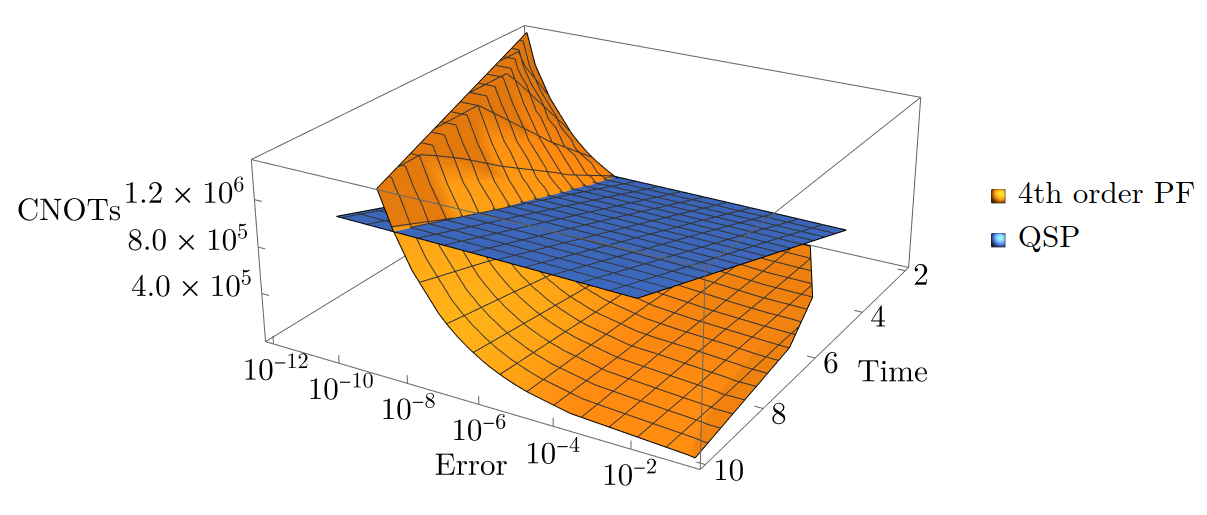}
    \caption{The top (bottom) plot shows the rotation (CNOT) gate count for implementing the time evolution operator for the single site Hamiltonian in Eq.~\eqref{eq:phi4} with $\nq=3$. Results are shown as a function of time $t$ and error $\eps$.
    The orange and blue colored surfaces show results using a 4th order PF and QSP, respectively. 
    The BE used in the QSP calculation was implemented using the method in Sec.~\ref{sssec:lcuft}.
    The PF requires less rotation and CNOT gates for errors $\eps \gtrsim 10^{-9}$.
    For error tolerances $\eps \lesssim 10^{-9}$, QSP wins due to its exponentially better scaling with $\eps$.
    Furthermore, the QSP gate count is almost independent of the error due to the optimal scaling with $\eps$.}
    \label{fig:cnot_3d}
\end{figure*}
\section{Discussion\label{sec:discussion}}

In this section we discuss how results of~\cref{sec:lattice_algorithms} provide guidance for choosing a simulation algorithm for the particular physical model and then examine the results of~\cref{sec:quartic}.

\begin{table}
    \centering
    \begin{tabular}{cl!{\vrule width 2\arrayrulewidth}c|c|c}
        \multicolumn{5}{c}{Non-geometrically-local} \\ \hlinewd{1pt}
        && $\eps$ & $t$ & $\NH$ \\
        \hline
        a)& PF & $\vphantom{\biggl(}\eps^{-\frac{1}{p}}$ & $t^{1+\frac{1}{p}}$ & $\bm{\NH^{1+\frac{1}{p}} \Nind}$ \\\hline
        b)& QSP+PF & $\eps^{-\frac{1}{p}}\log(1/\eps)$ & $\vphantom{\biggl(}t^{1+\frac{1}{p}}$ & $\begin{aligned}
        &\vphantom{\biggl(}\NH^{1+\frac{1}{p}}\Nind
        \\&\times\log(\NH\Nind)
        \end{aligned}$ \\\hline
        c)& QSP& $\vphantom{\biggl(}\bm{\log (1/\eps)}$ & $\bm{t}$ & $\bm{\NH^2\log(\NH)}$ \\
        \hlinewd{2pt}
        \multicolumn{5}{c}{Geometrically-local} \\ \hlinewd{1pt}
        && $\eps$ & $t$ & $\NLat$ \\
        \hline
        d)& PF& $\vphantom{\biggl(}\eps^{-\frac{1}{p}}$ & $t^{1+\frac{1}{p}}$ & $\NLat^{1+\frac{1}{p}}$ \\\hline
        e)& PF+HHKL & $\vphantom{\biggl(}\eps^{-\frac{1}{p}}$ & $t^{1+\frac{1}{p}}$ &  $\NLat^{1+\frac{1}{p}}$\\\hline
        f)& QSP+PF& $\vphantom{\biggl(}\eps^{-\frac{1}{p}}\log(1/\eps)$ & $t^{1+\frac{1}{p}}$ &  $\NLat^{1+\frac{1}{p}}\log(\NLat)$ \\\hline
        g)& QSP& $\vphantom{\biggl(}\bm{\log(1/\eps)}$ & $\bm{t}$ & $\NLat^2$\\\hline
        h)& QSP+HHKL& $\bm{\log(1/\eps)}$ & $\bm{t}$ & $\begin{aligned}
            &\vphantom{\biggl(}\bm{\NLat [\log(\NLat)]^\dims}
            \\
        &\bm{\times\log\log(\NLat)}
        \end{aligned}$
    \end{tabular}
    \caption{Asymptotic dependence ($\OO$ omitted) on major problem parameters for algorithms in~\cref{tab:site-local,tab:geometically-local}.
    Boldface indicates best asymptotic scaling among applicable methods.
    While $\Nind$ is constant for geometrically-local theories, it typically scales polynomially with $\NLat$ otherwise.
    For most non-geometrically-local theories, the order $p$ can be chosen large enough so that PFs demonstrate better scaling with $\NLat$ than QSP.
    }
    \label{tab:asympt_separate}
\end{table}

The major parameters of interest defining the asymptotic complexities of time evolution algorithms are the evolution time $t$, the error of the time-evolved state $\eps$, as well as some measure of Hamiltonian size/complexity.
For lattice systems, this can be the number of terms in the Hamiltonian $\NH$.
For a particular theory, $\NH$ can also be expressed in terms of lattice size $\NLat$.
The dependence of $\NH$ on $\NLat$ is linear for geometrically-site-local Hamiltonians and is, more generally, polynomial.
For algorithms based on PFs, one additionally needs to consider a parameter $\Nind$ related to the induced norm of the Hamiltonian.
The asymptotic dependence on individual parameters for algorithms discussed in~\cref{sec:review_alg,sec:lattice_algorithms} is summarized in~\cref{tab:asympt_separate}.

The parameter for which the difference in asymptotic dependence is most striking for PF- and QSP-based algorithms is the error $\eps$.
Nevertheless, using algorithms with $\log(1/\eps)$ dependence may not necessarily be an optimal choice for several reasons.
First, their asymptotic cost comes with large constant prefactors stemming from the implementation of the BE subroutine.
While progress is being made in this direction, the value of $\eps$ at which using such techniques clearly becomes advantageous is $\eps\sim10^{-9}$, which is much smaller than uncertainties coming from other sources (e.g., from digitizing a theory with bosonic degrees of freedom, finite lattice spacing and volume effects, to name just a few).
Second, as discussed below, in some scenarios PF-based algorithms have better dependence on problem size than those based on QSP.
Third, and least importantly (in the fault-tolerant regime), PF based algorithms do not require the usage of ancillary qubits, while QSP-based algorithms typically require a logarithmic (in problem parameters) number of those.

Recent research~\cite{rhodes2024exponential} indicates that the error tolerance at which near-optimal algorithms start outperforming the PF-based ones is highly problem-dependent.
In particular, when simulating lattice gauge theories, the complexity of PF-based simulation is likely to scale exponentially with the number of colors~\cite{Kan:2021xfc,Davoudi:2022xmb}, while the near-optimal algorithms can have polynomial scaling with the number of colors, which is achieved by leveraging the sparsity of gauge field operators.
As a result, the usage of near-optimal techniques may lead to tremendous savings starting at error rates as low as $\eps=10^{-1}$~\cite{rhodes2024exponential}.

In terms of the asymptotic dependence on problem size, a clear winner is the HHKL algorithm, assuming that for implementing the time evolution of local blocks, an algorithm with $\log(1/\eps)$ scaling is chosen, such as QSP.
Note, however, that for higher-order PFs~\cite{childs2019nearly} the improvement reached by the HHKL algorithm in terms of problem size may not be particularly pronounced, especially, unless a very high precision is required, as discussed above.

For non-geometrically-local Hamiltonians, the situation is largely determined by the Hamiltonian structure defining the dependence of $\Nind$ on $\NLat$.
In particular, in the second-quantized formulation, this dependence will be defined by the choice of single-particle basis states (plane waves~\cite{Kreshchuk:2020dla,Liu:2021otn}, wavelets~\cite{Brennen:2014iqu}, harmonic oscillator eigenstates~\cite{Vary:2009gt}, etc.).
Importantly, whether the Hamiltonian is geometrically-local or not, one can imagine scenarios in which, at a fixed value of $\eps$, switching from PF to QSP would make sense as one increases the problem size.
However, given the sharp difference in the dependence on $\eps$ for these approaches, such scenarios are not highly probable.

Lastly, we note that the dependence on $t$ is likely to not be a crucial parameter in choosing the simulation algorithm.
Nevertheless, a situation in which considering longer times can motivate one to switch from PF to QSP is possible as well.

In addition to factors considered in this work, other considerations should also be taken into account.
Most notably, the compatibility of methods based on BE with general sparse Hamiltonians and modern simulation techniques~\cite{Low:2018pte,PhysRevA.99.042314,Tong:2021rfv} and the potential for applying the PF-based methods on early-fault-tolerant quantum computers~\cite{Dong:2022mmq,Lin:2021rwb}.

Construction of efficient block encodings (BEs) is critical for improving the efficiency of QSP-based algorithms.
In~\cite{Jordan:2011ci,Klco:2018zqz} it was shown how digitizing a bosonic degree of freedom in the eigenbases of position/momentum operator is advantageous from the perspective of using PF algorithms.
There, the key factors are a) an economical qubit representation of arbitrary powers of position/momentum operators in their respective eigenbases and b) the efficiency of switching between the position and momentum bases with the aid of Fourier Transform implemented on the quantum computer.
In~\cref{sssec:lcuft} we showed how this approach can be equally well applied when constructing a BE with LCU.

\section{Conclusion\label{sec:conclusion}}

In this work we studied various approaches to simulating time evolution in application to Hamiltonian Lattice Field Theories (HLFTs).
We started off by providing a review of existing techniques based on Product Formula (PF) splitting~\cite{Childs:2019hts}, Quantum Signal Processing (QSP)~\cite{Low:2016sck} and Qubitization~\cite{Low:2016znh}, Linear Combination of Unitaries (LCU)~\cite{Childs:2012gwh},
as well as the HHKL algorithm~\cite{Haah:2018ekc}.
Following that, we derived asymptotic costs of these methods in terms of parameters describing site-local and geometrically-site-local Hamiltonians arising in HLFT, see~\cref{tab:site-local,tab:geometically-local}.
While in this work we chose to contrast the most accessible and widely used algorithms based on PF with the methods having the best asymptotic scaling (QSP, HHKL), we note that there exist intermediate approaches.
Most notably, this includes the direct usage of Taylor series and LCU without QSP~\cite{Berry:2014ivo}.

Next, we focused on a particular type of HLFT, the quartic scalar field theory defined on a hypercubic lattice.
Our specific choice was guided by the fact that all the fundamental interactions are mediated by bosons, and from the computational perspective the scalar field theory is largely reminiscent of lattice gauge theories~\cite{Bauer:2021gek,Haase:2020kaj,Ji:2020kjk,Bauer:2023jvw,Grabowska:2022uos,Kane:2022ejm,Ciavarella:2024fzw}.
Besides this, an extensive volume of literature exists focusing on simulations of fermionic systems~\cite{Bravyi:2000vfj,Verstraete:2005pn,Toloui:2013xab,tranter2015b,Babbush:2015utf,Aspuru-Guzik:2016mom,Setia:2018qmu,Babbush:2018ywg,Babbush:2019yrx,Kirby:2021vkt,sakamoto2023end}.
Developing simulation algorithms for systems combining particles of both statistics is an exciting direction which was pursued in~\cite{Macridin:2018gdw,Macridin:2018oli,Watson:2023oov} where authors assumed the usage of PF-based methods, and is currently investigated in the context of nearly-optimal techniques~\cite{rhodes2024exponential}.

For the scalar field theory, we worked out the query complexities of various simulation methods as well provided explicit gate counts (not the upper bounds) for the case of simulating the system on a single lattice site using PF and QSP. 
Next, we focused on block encoding a single site in scalar field theory and left to further work additional improvements, such as reducing the cost of \SELECT operator from $\OO(\NLat\log\NLat)$ to $\OO(\NLat)$~\cite{Babbush:2018ywg} or taking into account the structure of coefficients in the Hamiltonian operator while/in constructing \PREPARE operator.
We developed an improved version of the LCU block encoding for digitized bosons~\cite{Jordan:2011ci,Klco:2018zqz} which utilizes the Quantum Fourier Transform circuit for switching between the field and momentum eigenbases, see Fig.~\ref{fig:BEFT}.
We then provided explicit gate counts for simulating time evolution which showed that PF-based simulation works exceptionally well for simple systems, such as anharmonic oscillator, and outperforms the QSP-based algorithm for $\eps\gtrsim10^{-9}$.
We note, however, that considering more complex models (such as non-Abelian lattice gauge theories) and larger lattice system sizes may result in near-term algorithms outperforming PF for error tolerances as little as $\eps\lesssim 10^{-1}$~\cite{rhodes2024exponential}.

Another important distinction between PF and QSP methods is the ease of quoting accurate error bounds for the final result.
Starting with QSP methods, all that is needed to place an upper bound on the error is the scale factor $\alpha$, the time $t$, and the degree of the polynomial approximation.
While the error determined in this way is only an upper bound, because QSP methods have optimal scaling $\log(1/\eps)$ in the error, the cost of overestimating the error is low compared to the contribution stemming  from the $\alpha t$ term in~\cref{eq:nphi}.

For PF methods on the other hand, while rigorous error bounds also exist, they are known to generally significantly overestimate the actual error~\cite{Childs:2019hts}.
Choosing the value of the Trotter number $r$ based on these loose error bounds would therefore result in significant extra computational costs due to the sub-optimal $(1/\eps)^{1/p}$ scaling of PF methods.
Tightening error bounds in this case implies either performing laborious problem-specific commutator calculations~\cite{Childs:2019hts,Watson:2023oov} or performing numerical simulations at several values of the Trotter number in order to extrapolate the results.
As shown in Ref.~\cite{Carena:2021ltu}, since using PFs to approximate the time evolution operator is equivalent to introducing a temporal lattice with spacing $a_t$, such an extrapolation would necessitate a renormalization procedure with a costly scale setting procedure at each value of $a_t$.\footnote{As was also shown in Ref.~\cite{Carena:2021ltu}, to avoid a double extrapolation $a \to 0$ and $a_t \to 0$, one should work at fixed anisotropy $\xi \equiv a/a_t$ and take the limit $a \to 0$.}

\section*{Acknowledgements}

The authors are grateful to Jeongwan Haah, Lin Lin, and Ivan Mauricio Burbano Aldana for discussions and comments.
SH acknowledges support from the Berkeley Center for Theoretical Physics and the National Energy Research Scientific Computing Center (NERSC), a U.S. Department of Energy Office of Science User Facility located at Lawrence Berkeley National Laboratory, operated under Contract No. DE-AC02- 05CH11231. CWB and MK were supported by the DOE, Office of Science under contract DE-AC02-05CH11231, partially through Quantum Information Science Enabled Discovery (QuantISED) for High Energy Physics (KA2401032). 
This material is based upon work supported by the U.S. Department of
Energy, Office of Science, Office of Advanced Scientific Computing Research, Department of Energy Computational Science Graduate Fellowship under Award Number DE-SC0020347.

\bibliographystyle{apsrev4-1}
\bibliography{main}

\begin{thebibliography}{97}%
\makeatletter
\providecommand \@ifxundefined [1]{%
 \@ifx{#1\undefined}
}%
\providecommand \@ifnum [1]{%
 \ifnum #1\expandafter \@firstoftwo
 \else \expandafter \@secondoftwo
 \fi
}%
\providecommand \@ifx [1]{%
 \ifx #1\expandafter \@firstoftwo
 \else \expandafter \@secondoftwo
 \fi
}%
\providecommand \natexlab [1]{#1}%
\providecommand \enquote  [1]{``#1''}%
\providecommand \bibnamefont  [1]{#1}%
\providecommand \bibfnamefont [1]{#1}%
\providecommand \citenamefont [1]{#1}%
\providecommand \href@noop [0]{\@secondoftwo}%
\providecommand \href [0]{\begingroup \@sanitize@url \@href}%
\providecommand \@href[1]{\@@startlink{#1}\@@href}%
\providecommand \@@href[1]{\endgroup#1\@@endlink}%
\providecommand \@sanitize@url [0]{\catcode `\\12\catcode `\$12\catcode `\&12\catcode `\#12\catcode `\^12\catcode `\_12\catcode `\%12\relax}%
\providecommand \@@startlink[1]{}%
\providecommand \@@endlink[0]{}%
\providecommand \url  [0]{\begingroup\@sanitize@url \@url }%
\providecommand \@url [1]{\endgroup\@href {#1}{\urlprefix }}%
\providecommand \urlprefix  [0]{URL }%
\providecommand \Eprint [0]{\href }%
\providecommand \doibase [0]{http://dx.doi.org/}%
\providecommand \selectlanguage [0]{\@gobble}%
\providecommand \bibinfo  [0]{\@secondoftwo}%
\providecommand \bibfield  [0]{\@secondoftwo}%
\providecommand \translation [1]{[#1]}%
\providecommand \BibitemOpen [0]{}%
\providecommand \bibitemStop [0]{}%
\providecommand \bibitemNoStop [0]{.\EOS\space}%
\providecommand \EOS [0]{\spacefactor3000\relax}%
\providecommand \BibitemShut  [1]{\csname bibitem#1\endcsname}%
\let\auto@bib@innerbib\@empty
\bibitem [{\citenamefont {Davoudi}\ \emph {et~al.}(2022)\citenamefont {Davoudi} \emph {et~al.}}]{Davoudi:2022bnl}%
  \BibitemOpen
  \bibfield  {author} {\bibinfo {author} {\bibfnamefont {Z.}~\bibnamefont {Davoudi}} \emph {et~al.},\ }in\ \href@noop {} {\emph {\bibinfo {booktitle} {{Snowmass 2021}}}}\ (\bibinfo {year} {2022})\ \Eprint {http://arxiv.org/abs/2209.10758} {arXiv:2209.10758 [hep-lat]} \BibitemShut {NoStop}%
\bibitem [{\citenamefont {Morningstar}(2007)}]{Morningstar:2007zm}%
  \BibitemOpen
  \bibfield  {author} {\bibinfo {author} {\bibfnamefont {C.}~\bibnamefont {Morningstar}},\ }in\ \href@noop {} {\emph {\bibinfo {booktitle} {{21st Annual Hampton University Graduate Studies Program (HUGS 2006)}}}}\ (\bibinfo {year} {2007})\ \Eprint {http://arxiv.org/abs/hep-lat/0702020} {arXiv:hep-lat/0702020} \BibitemShut {NoStop}%
\bibitem [{\citenamefont {Loh}\ \emph {et~al.}(1990)\citenamefont {Loh}, \citenamefont {Gubernatis}, \citenamefont {Scalettar}, \citenamefont {White}, \citenamefont {Scalapino},\ and\ \citenamefont {Sugar}}]{Loh:1990zz}%
  \BibitemOpen
  \bibfield  {author} {\bibinfo {author} {\bibfnamefont {E.~Y.}\ \bibnamefont {Loh}}, \bibinfo {author} {\bibfnamefont {J.~E.}\ \bibnamefont {Gubernatis}}, \bibinfo {author} {\bibfnamefont {R.~T.}\ \bibnamefont {Scalettar}}, \bibinfo {author} {\bibfnamefont {S.~R.}\ \bibnamefont {White}}, \bibinfo {author} {\bibfnamefont {D.~J.}\ \bibnamefont {Scalapino}}, \ and\ \bibinfo {author} {\bibfnamefont {R.~L.}\ \bibnamefont {Sugar}},\ }\href {\doibase 10.1103/PhysRevB.41.9301} {\bibfield  {journal} {\bibinfo  {journal} {Phys. Rev. B}\ }\textbf {\bibinfo {volume} {41}},\ \bibinfo {pages} {9301} (\bibinfo {year} {1990})}\BibitemShut {NoStop}%
\bibitem [{\citenamefont {Pan}\ and\ \citenamefont {Meng}(2022)}]{Pan:2022fgf}%
  \BibitemOpen
  \bibfield  {author} {\bibinfo {author} {\bibfnamefont {G.}~\bibnamefont {Pan}}\ and\ \bibinfo {author} {\bibfnamefont {Z.~Y.}\ \bibnamefont {Meng}},\ }\href {\doibase 10.1016/B978-0-323-90800-9.00095-0} {\  (\bibinfo {year} {2022}),\ 10.1016/B978-0-323-90800-9.00095-0},\ \Eprint {http://arxiv.org/abs/2204.08777} {arXiv:2204.08777 [cond-mat.str-el]} \BibitemShut {NoStop}%
\bibitem [{\citenamefont {Kogut}\ and\ \citenamefont {Susskind}(1975)}]{Kogut:1974ag}%
  \BibitemOpen
  \bibfield  {author} {\bibinfo {author} {\bibfnamefont {J.~B.}\ \bibnamefont {Kogut}}\ and\ \bibinfo {author} {\bibfnamefont {L.}~\bibnamefont {Susskind}},\ }\href {\doibase 10.1103/PhysRevD.11.395} {\bibfield  {journal} {\bibinfo  {journal} {Phys. Rev. D}\ }\textbf {\bibinfo {volume} {11}},\ \bibinfo {pages} {395} (\bibinfo {year} {1975})}\BibitemShut {NoStop}%
\bibitem [{\citenamefont {Drell}\ \emph {et~al.}(1979)\citenamefont {Drell}, \citenamefont {Quinn}, \citenamefont {Svetitsky},\ and\ \citenamefont {Weinstein}}]{drell1979quantum}%
  \BibitemOpen
  \bibfield  {author} {\bibinfo {author} {\bibfnamefont {S.~D.}\ \bibnamefont {Drell}}, \bibinfo {author} {\bibfnamefont {H.~R.}\ \bibnamefont {Quinn}}, \bibinfo {author} {\bibfnamefont {B.}~\bibnamefont {Svetitsky}}, \ and\ \bibinfo {author} {\bibfnamefont {M.}~\bibnamefont {Weinstein}},\ }\href@noop {} {\bibfield  {journal} {\bibinfo  {journal} {Physical Review D}\ }\textbf {\bibinfo {volume} {19}},\ \bibinfo {pages} {619} (\bibinfo {year} {1979})}\BibitemShut {NoStop}%
\bibitem [{\citenamefont {Bardeen}\ and\ \citenamefont {Pearson}(1976)}]{bardeen1976local}%
  \BibitemOpen
  \bibfield  {author} {\bibinfo {author} {\bibfnamefont {W.~A.}\ \bibnamefont {Bardeen}}\ and\ \bibinfo {author} {\bibfnamefont {R.~B.}\ \bibnamefont {Pearson}},\ }\href@noop {} {\bibfield  {journal} {\bibinfo  {journal} {Physical Review D}\ }\textbf {\bibinfo {volume} {14}},\ \bibinfo {pages} {547} (\bibinfo {year} {1976})}\BibitemShut {NoStop}%
\bibitem [{\citenamefont {Bardeen}\ \emph {et~al.}(1980)\citenamefont {Bardeen}, \citenamefont {Pearson},\ and\ \citenamefont {Rabinovici}}]{Bardeen:1979xx}%
  \BibitemOpen
  \bibfield  {author} {\bibinfo {author} {\bibfnamefont {W.~A.}\ \bibnamefont {Bardeen}}, \bibinfo {author} {\bibfnamefont {R.~B.}\ \bibnamefont {Pearson}}, \ and\ \bibinfo {author} {\bibfnamefont {E.}~\bibnamefont {Rabinovici}},\ }\href {\doibase 10.1103/PhysRevD.21.1037} {\bibfield  {journal} {\bibinfo  {journal} {Phys. Rev. D}\ }\textbf {\bibinfo {volume} {21}},\ \bibinfo {pages} {1037} (\bibinfo {year} {1980})}\BibitemShut {NoStop}%
\bibitem [{\citenamefont {Liu}\ \emph {et~al.}(2022)\citenamefont {Liu}, \citenamefont {Li}, \citenamefont {Zheng}, \citenamefont {Yuan},\ and\ \citenamefont {Sun}}]{Liu:2021otn}%
  \BibitemOpen
  \bibfield  {author} {\bibinfo {author} {\bibfnamefont {J.}~\bibnamefont {Liu}}, \bibinfo {author} {\bibfnamefont {Z.}~\bibnamefont {Li}}, \bibinfo {author} {\bibfnamefont {H.}~\bibnamefont {Zheng}}, \bibinfo {author} {\bibfnamefont {X.}~\bibnamefont {Yuan}}, \ and\ \bibinfo {author} {\bibfnamefont {J.}~\bibnamefont {Sun}},\ }\href {\doibase 10.1088/2632-2153/aca06b} {\bibfield  {journal} {\bibinfo  {journal} {Mach. Learn. Sci. Tech.}\ }\textbf {\bibinfo {volume} {3}},\ \bibinfo {pages} {045030} (\bibinfo {year} {2022})},\ \Eprint {http://arxiv.org/abs/2109.05547} {arXiv:2109.05547 [quant-ph]} \BibitemShut {NoStop}%
\bibitem [{\citenamefont {Vary}\ \emph {et~al.}(2010)\citenamefont {Vary}, \citenamefont {Honkanen}, \citenamefont {Li}, \citenamefont {Maris}, \citenamefont {Brodsky}, \citenamefont {Harindranath}, \citenamefont {de~Teramond}, \citenamefont {Sternberg}, \citenamefont {Ng},\ and\ \citenamefont {Yang}}]{Vary:2009gt}%
  \BibitemOpen
  \bibfield  {author} {\bibinfo {author} {\bibfnamefont {J.~P.}\ \bibnamefont {Vary}}, \bibinfo {author} {\bibfnamefont {H.}~\bibnamefont {Honkanen}}, \bibinfo {author} {\bibfnamefont {J.}~\bibnamefont {Li}}, \bibinfo {author} {\bibfnamefont {P.}~\bibnamefont {Maris}}, \bibinfo {author} {\bibfnamefont {S.~J.}\ \bibnamefont {Brodsky}}, \bibinfo {author} {\bibfnamefont {A.}~\bibnamefont {Harindranath}}, \bibinfo {author} {\bibfnamefont {G.~F.}\ \bibnamefont {de~Teramond}}, \bibinfo {author} {\bibfnamefont {P.}~\bibnamefont {Sternberg}}, \bibinfo {author} {\bibfnamefont {E.~G.}\ \bibnamefont {Ng}}, \ and\ \bibinfo {author} {\bibfnamefont {C.}~\bibnamefont {Yang}},\ }\href {\doibase 10.1103/PhysRevC.81.035205} {\bibfield  {journal} {\bibinfo  {journal} {Phys. Rev. C}\ }\textbf {\bibinfo {volume} {81}},\ \bibinfo {pages} {035205} (\bibinfo {year} {2010})},\ \Eprint {http://arxiv.org/abs/0905.1411} {arXiv:0905.1411 [nucl-th]} \BibitemShut {NoStop}%
\bibitem [{\citenamefont {Klco}\ and\ \citenamefont {Savage}(2019)}]{Klco:2018zqz}%
  \BibitemOpen
  \bibfield  {author} {\bibinfo {author} {\bibfnamefont {N.}~\bibnamefont {Klco}}\ and\ \bibinfo {author} {\bibfnamefont {M.~J.}\ \bibnamefont {Savage}},\ }\href {\doibase 10.1103/PhysRevA.99.052335} {\bibfield  {journal} {\bibinfo  {journal} {Phys. Rev. A}\ }\textbf {\bibinfo {volume} {99}},\ \bibinfo {pages} {052335} (\bibinfo {year} {2019})},\ \Eprint {http://arxiv.org/abs/1808.10378} {arXiv:1808.10378 [quant-ph]} \BibitemShut {NoStop}%
\bibitem [{\citenamefont {Pauli}\ and\ \citenamefont {Brodsky}(1985{\natexlab{a}})}]{Pauli:1985pv}%
  \BibitemOpen
  \bibfield  {author} {\bibinfo {author} {\bibfnamefont {H.~C.}\ \bibnamefont {Pauli}}\ and\ \bibinfo {author} {\bibfnamefont {S.~J.}\ \bibnamefont {Brodsky}},\ }\href {\doibase 10.1103/PhysRevD.32.1993} {\bibfield  {journal} {\bibinfo  {journal} {Phys. Rev. D}\ }\textbf {\bibinfo {volume} {32}},\ \bibinfo {pages} {1993} (\bibinfo {year} {1985}{\natexlab{a}})}\BibitemShut {NoStop}%
\bibitem [{\citenamefont {Pauli}\ and\ \citenamefont {Brodsky}(1985{\natexlab{b}})}]{pauli1985discretized}%
  \BibitemOpen
  \bibfield  {author} {\bibinfo {author} {\bibfnamefont {H.-C.}\ \bibnamefont {Pauli}}\ and\ \bibinfo {author} {\bibfnamefont {S.~J.}\ \bibnamefont {Brodsky}},\ }\href@noop {} {\bibfield  {journal} {\bibinfo  {journal} {Physical Review D}\ }\textbf {\bibinfo {volume} {32}},\ \bibinfo {pages} {2001} (\bibinfo {year} {1985}{\natexlab{b}})}\BibitemShut {NoStop}%
\bibitem [{\citenamefont {Brodsky}\ \emph {et~al.}(1998)\citenamefont {Brodsky}, \citenamefont {Pauli},\ and\ \citenamefont {Pinsky}}]{Brodsky:1997de}%
  \BibitemOpen
  \bibfield  {author} {\bibinfo {author} {\bibfnamefont {S.~J.}\ \bibnamefont {Brodsky}}, \bibinfo {author} {\bibfnamefont {H.-C.}\ \bibnamefont {Pauli}}, \ and\ \bibinfo {author} {\bibfnamefont {S.~S.}\ \bibnamefont {Pinsky}},\ }\href {\doibase 10.1016/S0370-1573(97)00089-6} {\bibfield  {journal} {\bibinfo  {journal} {Phys. Rept.}\ }\textbf {\bibinfo {volume} {301}},\ \bibinfo {pages} {299} (\bibinfo {year} {1998})},\ \Eprint {http://arxiv.org/abs/hep-ph/9705477} {arXiv:hep-ph/9705477} \BibitemShut {NoStop}%
\bibitem [{\citenamefont {Jordan}\ \emph {et~al.}(2018)\citenamefont {Jordan}, \citenamefont {Krovi}, \citenamefont {Lee},\ and\ \citenamefont {Preskill}}]{Jordan:2017lea}%
  \BibitemOpen
  \bibfield  {author} {\bibinfo {author} {\bibfnamefont {S.~P.}\ \bibnamefont {Jordan}}, \bibinfo {author} {\bibfnamefont {H.}~\bibnamefont {Krovi}}, \bibinfo {author} {\bibfnamefont {K.~S.~M.}\ \bibnamefont {Lee}}, \ and\ \bibinfo {author} {\bibfnamefont {J.}~\bibnamefont {Preskill}},\ }\href {\doibase 10.22331/q-2018-01-08-44} {\bibfield  {journal} {\bibinfo  {journal} {Quantum}\ }\textbf {\bibinfo {volume} {2}},\ \bibinfo {pages} {44} (\bibinfo {year} {2018})},\ \Eprint {http://arxiv.org/abs/1703.00454} {arXiv:1703.00454 [quant-ph]} \BibitemShut {NoStop}%
\bibitem [{\citenamefont {Jordan}\ \emph {et~al.}(2012)\citenamefont {Jordan}, \citenamefont {Lee},\ and\ \citenamefont {Preskill}}]{Jordan:2012xnu}%
  \BibitemOpen
  \bibfield  {author} {\bibinfo {author} {\bibfnamefont {S.~P.}\ \bibnamefont {Jordan}}, \bibinfo {author} {\bibfnamefont {K.~S.~M.}\ \bibnamefont {Lee}}, \ and\ \bibinfo {author} {\bibfnamefont {J.}~\bibnamefont {Preskill}},\ }\href {\doibase 10.1126/science.1217069} {\bibfield  {journal} {\bibinfo  {journal} {Science}\ }\textbf {\bibinfo {volume} {336}},\ \bibinfo {pages} {1130} (\bibinfo {year} {2012})},\ \Eprint {http://arxiv.org/abs/1111.3633} {arXiv:1111.3633 [quant-ph]} \BibitemShut {NoStop}%
\bibitem [{\citenamefont {Feynman}(1982)}]{Feynman:1981tf}%
  \BibitemOpen
  \bibfield  {author} {\bibinfo {author} {\bibfnamefont {R.~P.}\ \bibnamefont {Feynman}},\ }\href {\doibase 10.1007/BF02650179} {\bibfield  {journal} {\bibinfo  {journal} {Int. J. Theor. Phys.}\ }\textbf {\bibinfo {volume} {21}},\ \bibinfo {pages} {467} (\bibinfo {year} {1982})}\BibitemShut {NoStop}%
\bibitem [{\citenamefont {Bauer}\ \emph {et~al.}(2023{\natexlab{a}})\citenamefont {Bauer} \emph {et~al.}}]{Bauer:2022hpo}%
  \BibitemOpen
  \bibfield  {author} {\bibinfo {author} {\bibfnamefont {C.~W.}\ \bibnamefont {Bauer}} \emph {et~al.},\ }\href {\doibase 10.1103/PRXQuantum.4.027001} {\bibfield  {journal} {\bibinfo  {journal} {PRX Quantum}\ }\textbf {\bibinfo {volume} {4}},\ \bibinfo {pages} {027001} (\bibinfo {year} {2023}{\natexlab{a}})},\ \Eprint {http://arxiv.org/abs/2204.03381} {arXiv:2204.03381 [quant-ph]} \BibitemShut {NoStop}%
\bibitem [{\citenamefont {Jordan}\ \emph {et~al.}(2014)\citenamefont {Jordan}, \citenamefont {Lee},\ and\ \citenamefont {Preskill}}]{Jordan:2011ci}%
  \BibitemOpen
  \bibfield  {author} {\bibinfo {author} {\bibfnamefont {S.~P.}\ \bibnamefont {Jordan}}, \bibinfo {author} {\bibfnamefont {K.~S.~M.}\ \bibnamefont {Lee}}, \ and\ \bibinfo {author} {\bibfnamefont {J.}~\bibnamefont {Preskill}},\ }\href@noop {} {\bibfield  {journal} {\bibinfo  {journal} {Quant. Inf. Comput.}\ }\textbf {\bibinfo {volume} {14}},\ \bibinfo {pages} {1014} (\bibinfo {year} {2014})},\ \Eprint {http://arxiv.org/abs/1112.4833} {arXiv:1112.4833 [hep-th]} \BibitemShut {NoStop}%
\bibitem [{\citenamefont {Lloyd}(1996)}]{Lloyd:1996aai}%
  \BibitemOpen
  \bibfield  {author} {\bibinfo {author} {\bibfnamefont {S.}~\bibnamefont {Lloyd}},\ }\href {\doibase 10.1126/science.273.5278.1073} {\bibfield  {journal} {\bibinfo  {journal} {Science}\ }\textbf {\bibinfo {volume} {273}},\ \bibinfo {pages} {1073} (\bibinfo {year} {1996})}\BibitemShut {NoStop}%
\bibitem [{\citenamefont {Zalka}(1998)}]{zalka1998simulating}%
  \BibitemOpen
  \bibfield  {author} {\bibinfo {author} {\bibfnamefont {C.}~\bibnamefont {Zalka}},\ }\href@noop {} {\bibfield  {journal} {\bibinfo  {journal} {Proceedings of the Royal Society of London. Series A: Mathematical, Physical and Engineering Sciences}\ }\textbf {\bibinfo {volume} {454}},\ \bibinfo {pages} {313} (\bibinfo {year} {1998})}\BibitemShut {NoStop}%
\bibitem [{\citenamefont {Wiesner}(1996)}]{Wiesner:1996xg}%
  \BibitemOpen
  \bibfield  {author} {\bibinfo {author} {\bibfnamefont {S.}~\bibnamefont {Wiesner}},\ }\href@noop {} {\  (\bibinfo {year} {1996})},\ \Eprint {http://arxiv.org/abs/quant-ph/9603028} {arXiv:quant-ph/9603028} \BibitemShut {NoStop}%
\bibitem [{\citenamefont {Cohen}\ \emph {et~al.}(1982)\citenamefont {Cohen}, \citenamefont {Friedland}, \citenamefont {Kato},\ and\ \citenamefont {Kelly}}]{cohen1982eigenvalue}%
  \BibitemOpen
  \bibfield  {author} {\bibinfo {author} {\bibfnamefont {J.~E.}\ \bibnamefont {Cohen}}, \bibinfo {author} {\bibfnamefont {S.}~\bibnamefont {Friedland}}, \bibinfo {author} {\bibfnamefont {T.}~\bibnamefont {Kato}}, \ and\ \bibinfo {author} {\bibfnamefont {F.~P.}\ \bibnamefont {Kelly}},\ }\href@noop {} {\bibfield  {journal} {\bibinfo  {journal} {Linear Algebra and its Applications}\ }\textbf {\bibinfo {volume} {45}},\ \bibinfo {pages} {55} (\bibinfo {year} {1982})}\BibitemShut {NoStop}%
\bibitem [{\citenamefont {Suzuki}(1976)}]{suzuki1976generalized}%
  \BibitemOpen
  \bibfield  {author} {\bibinfo {author} {\bibfnamefont {M.}~\bibnamefont {Suzuki}},\ }\href@noop {} {\bibfield  {journal} {\bibinfo  {journal} {Communications in Mathematical Physics}\ }\textbf {\bibinfo {volume} {51}},\ \bibinfo {pages} {183} (\bibinfo {year} {1976})}\BibitemShut {NoStop}%
\bibitem [{\citenamefont {Suzuki}(1991)}]{Suzuki:1991jtk}%
  \BibitemOpen
  \bibfield  {author} {\bibinfo {author} {\bibfnamefont {M.}~\bibnamefont {Suzuki}},\ }\href {\doibase 10.1063/1.529425} {\bibfield  {journal} {\bibinfo  {journal} {J. Math. Phys.}\ }\textbf {\bibinfo {volume} {32}},\ \bibinfo {pages} {400} (\bibinfo {year} {1991})}\BibitemShut {NoStop}%
\bibitem [{\citenamefont {Childs}\ and\ \citenamefont {Su}(2019)}]{childs2019nearly}%
  \BibitemOpen
  \bibfield  {author} {\bibinfo {author} {\bibfnamefont {A.~M.}\ \bibnamefont {Childs}}\ and\ \bibinfo {author} {\bibfnamefont {Y.}~\bibnamefont {Su}},\ }\href@noop {} {\bibfield  {journal} {\bibinfo  {journal} {Physical review letters}\ }\textbf {\bibinfo {volume} {123}},\ \bibinfo {pages} {050503} (\bibinfo {year} {2019})}\BibitemShut {NoStop}%
\bibitem [{\citenamefont {Childs}\ \emph {et~al.}(2021)\citenamefont {Childs}, \citenamefont {Su}, \citenamefont {Tran}, \citenamefont {Wiebe},\ and\ \citenamefont {Zhu}}]{Childs:2019hts}%
  \BibitemOpen
  \bibfield  {author} {\bibinfo {author} {\bibfnamefont {A.~M.}\ \bibnamefont {Childs}}, \bibinfo {author} {\bibfnamefont {Y.}~\bibnamefont {Su}}, \bibinfo {author} {\bibfnamefont {M.~C.}\ \bibnamefont {Tran}}, \bibinfo {author} {\bibfnamefont {N.}~\bibnamefont {Wiebe}}, \ and\ \bibinfo {author} {\bibfnamefont {S.}~\bibnamefont {Zhu}},\ }\href {\doibase 10.1103/physrevx.11.011020} {\bibfield  {journal} {\bibinfo  {journal} {Phys. Rev. X}\ }\textbf {\bibinfo {volume} {11}},\ \bibinfo {pages} {011020} (\bibinfo {year} {2021})},\ \Eprint {http://arxiv.org/abs/1912.08854} {arXiv:1912.08854 [quant-ph]} \BibitemShut {NoStop}%
\bibitem [{\citenamefont {Hatano}\ and\ \citenamefont {Suzuki}(2005)}]{Hatano:2005gh}%
  \BibitemOpen
  \bibfield  {author} {\bibinfo {author} {\bibfnamefont {N.}~\bibnamefont {Hatano}}\ and\ \bibinfo {author} {\bibfnamefont {M.}~\bibnamefont {Suzuki}},\ }\href {\doibase 10.1007/11526216_2} {\bibfield  {journal} {\bibinfo  {journal} {Lect. Notes Phys.}\ }\textbf {\bibinfo {volume} {679}},\ \bibinfo {pages} {37} (\bibinfo {year} {2005})},\ \Eprint {http://arxiv.org/abs/math-ph/0506007} {arXiv:math-ph/0506007} \BibitemShut {NoStop}%
\bibitem [{\citenamefont {Low}\ and\ \citenamefont {Chuang}(2017)}]{Low:2016sck}%
  \BibitemOpen
  \bibfield  {author} {\bibinfo {author} {\bibfnamefont {G.~H.}\ \bibnamefont {Low}}\ and\ \bibinfo {author} {\bibfnamefont {I.~L.}\ \bibnamefont {Chuang}},\ }\href {\doibase 10.1103/PhysRevLett.118.010501} {\bibfield  {journal} {\bibinfo  {journal} {Phys. Rev. Lett.}\ }\textbf {\bibinfo {volume} {118}},\ \bibinfo {pages} {010501} (\bibinfo {year} {2017})},\ \Eprint {http://arxiv.org/abs/1606.02685} {arXiv:1606.02685 [quant-ph]} \BibitemShut {NoStop}%
\bibitem [{\citenamefont {Low}\ and\ \citenamefont {Chuang}(2019)}]{Low:2016znh}%
  \BibitemOpen
  \bibfield  {author} {\bibinfo {author} {\bibfnamefont {G.~H.}\ \bibnamefont {Low}}\ and\ \bibinfo {author} {\bibfnamefont {I.~L.}\ \bibnamefont {Chuang}},\ }\href {\doibase 10.22331/q-2019-07-12-163} {\bibfield  {journal} {\bibinfo  {journal} {Quantum}\ }\textbf {\bibinfo {volume} {3}},\ \bibinfo {pages} {163} (\bibinfo {year} {2019})},\ \Eprint {http://arxiv.org/abs/1610.06546} {arXiv:1610.06546 [quant-ph]} \BibitemShut {NoStop}%
\bibitem [{\citenamefont {Zeytino\u{g}lu}\ and\ \citenamefont {Sugiura}(2024)}]{Zeytinoglu:2022zda}%
  \BibitemOpen
  \bibfield  {author} {\bibinfo {author} {\bibfnamefont {S.}~\bibnamefont {Zeytino\u{g}lu}}\ and\ \bibinfo {author} {\bibfnamefont {S.}~\bibnamefont {Sugiura}},\ }\href {\doibase 10.1103/PhysRevResearch.6.013003} {\bibfield  {journal} {\bibinfo  {journal} {Phys. Rev. Res.}\ }\textbf {\bibinfo {volume} {6}},\ \bibinfo {pages} {013003} (\bibinfo {year} {2024})},\ \Eprint {http://arxiv.org/abs/2201.04665} {arXiv:2201.04665 [quant-ph]} \BibitemShut {NoStop}%
\bibitem [{\citenamefont {Motlagh}\ and\ \citenamefont {Wiebe}(2023)}]{Motlagh:2023oqc}%
  \BibitemOpen
  \bibfield  {author} {\bibinfo {author} {\bibfnamefont {D.}~\bibnamefont {Motlagh}}\ and\ \bibinfo {author} {\bibfnamefont {N.}~\bibnamefont {Wiebe}},\ }\href@noop {} {\  (\bibinfo {year} {2023})},\ \Eprint {http://arxiv.org/abs/2308.01501} {arXiv:2308.01501 [quant-ph]} \BibitemShut {NoStop}%
\bibitem [{\citenamefont {Kikuchi}\ \emph {et~al.}(2023)\citenamefont {Kikuchi}, \citenamefont {Mc~Keever}, \citenamefont {Coopmans}, \citenamefont {Lubasch},\ and\ \citenamefont {Benedetti}}]{Kikuchi:2023qbb}%
  \BibitemOpen
  \bibfield  {author} {\bibinfo {author} {\bibfnamefont {Y.}~\bibnamefont {Kikuchi}}, \bibinfo {author} {\bibfnamefont {C.}~\bibnamefont {Mc~Keever}}, \bibinfo {author} {\bibfnamefont {L.}~\bibnamefont {Coopmans}}, \bibinfo {author} {\bibfnamefont {M.}~\bibnamefont {Lubasch}}, \ and\ \bibinfo {author} {\bibfnamefont {M.}~\bibnamefont {Benedetti}},\ }\href {\doibase 10.1038/s41534-023-00762-0} {\bibfield  {journal} {\bibinfo  {journal} {npj Quantum Inf.}\ }\textbf {\bibinfo {volume} {9}},\ \bibinfo {pages} {93} (\bibinfo {year} {2023})},\ \Eprint {http://arxiv.org/abs/2303.05533} {arXiv:2303.05533 [quant-ph]} \BibitemShut {NoStop}%
\bibitem [{\citenamefont {Haah}\ \emph {et~al.}(2018)\citenamefont {Haah}, \citenamefont {Hastings}, \citenamefont {Kothari},\ and\ \citenamefont {Low}}]{Haah:2018ekc}%
  \BibitemOpen
  \bibfield  {author} {\bibinfo {author} {\bibfnamefont {J.}~\bibnamefont {Haah}}, \bibinfo {author} {\bibfnamefont {M.~B.}\ \bibnamefont {Hastings}}, \bibinfo {author} {\bibfnamefont {R.}~\bibnamefont {Kothari}}, \ and\ \bibinfo {author} {\bibfnamefont {G.~H.}\ \bibnamefont {Low}},\ }\href {\doibase 10.1137/18M1231511} {\  (\bibinfo {year} {2018}),\ 10.1137/18M1231511},\ \Eprint {http://arxiv.org/abs/1801.03922} {arXiv:1801.03922 [quant-ph]} \BibitemShut {NoStop}%
\bibitem [{\citenamefont {Childs}\ and\ \citenamefont {Wiebe}(2012)}]{Childs:2012gwh}%
  \BibitemOpen
  \bibfield  {author} {\bibinfo {author} {\bibfnamefont {A.~M.}\ \bibnamefont {Childs}}\ and\ \bibinfo {author} {\bibfnamefont {N.}~\bibnamefont {Wiebe}},\ }\href {\doibase 10.26421/QIC12.11-12-1} {\bibfield  {journal} {\bibinfo  {journal} {Quant. Inf. Comput.}\ }\textbf {\bibinfo {volume} {12}},\ \bibinfo {pages} {0901} (\bibinfo {year} {2012})},\ \Eprint {http://arxiv.org/abs/1202.5822} {arXiv:1202.5822 [quant-ph]} \BibitemShut {NoStop}%
\bibitem [{\citenamefont {Toloui}\ and\ \citenamefont {Love}(2013)}]{Toloui:2013xab}%
  \BibitemOpen
  \bibfield  {author} {\bibinfo {author} {\bibfnamefont {B.}~\bibnamefont {Toloui}}\ and\ \bibinfo {author} {\bibfnamefont {P.~J.}\ \bibnamefont {Love}},\ }\href@noop {} {\  (\bibinfo {year} {2013})},\ \Eprint {http://arxiv.org/abs/1312.2579} {arXiv:1312.2579 [quant-ph]} \BibitemShut {NoStop}%
\bibitem [{\citenamefont {Babbush}\ \emph {et~al.}(2018)\citenamefont {Babbush}, \citenamefont {Gidney}, \citenamefont {Berry}, \citenamefont {Wiebe}, \citenamefont {McClean}, \citenamefont {Paler}, \citenamefont {Fowler},\ and\ \citenamefont {Neven}}]{Babbush:2018ywg}%
  \BibitemOpen
  \bibfield  {author} {\bibinfo {author} {\bibfnamefont {R.}~\bibnamefont {Babbush}}, \bibinfo {author} {\bibfnamefont {C.}~\bibnamefont {Gidney}}, \bibinfo {author} {\bibfnamefont {D.~W.}\ \bibnamefont {Berry}}, \bibinfo {author} {\bibfnamefont {N.}~\bibnamefont {Wiebe}}, \bibinfo {author} {\bibfnamefont {J.}~\bibnamefont {McClean}}, \bibinfo {author} {\bibfnamefont {A.}~\bibnamefont {Paler}}, \bibinfo {author} {\bibfnamefont {A.}~\bibnamefont {Fowler}}, \ and\ \bibinfo {author} {\bibfnamefont {H.}~\bibnamefont {Neven}},\ }\href {\doibase 10.1103/PhysRevX.8.041015} {\bibfield  {journal} {\bibinfo  {journal} {Phys. Rev. X}\ }\textbf {\bibinfo {volume} {8}},\ \bibinfo {pages} {041015} (\bibinfo {year} {2018})}\BibitemShut {NoStop}%
\bibitem [{\citenamefont {Kreshchuk}\ \emph {et~al.}(2022)\citenamefont {Kreshchuk}, \citenamefont {Kirby}, \citenamefont {Goldstein}, \citenamefont {Beauchemin},\ and\ \citenamefont {Love}}]{Kreshchuk:2020dla}%
  \BibitemOpen
  \bibfield  {author} {\bibinfo {author} {\bibfnamefont {M.}~\bibnamefont {Kreshchuk}}, \bibinfo {author} {\bibfnamefont {W.~M.}\ \bibnamefont {Kirby}}, \bibinfo {author} {\bibfnamefont {G.}~\bibnamefont {Goldstein}}, \bibinfo {author} {\bibfnamefont {H.}~\bibnamefont {Beauchemin}}, \ and\ \bibinfo {author} {\bibfnamefont {P.~J.}\ \bibnamefont {Love}},\ }\href {\doibase 10.1103/PhysRevA.105.032418} {\bibfield  {journal} {\bibinfo  {journal} {Phys. Rev. A}\ }\textbf {\bibinfo {volume} {105}},\ \bibinfo {pages} {032418} (\bibinfo {year} {2022})},\ \Eprint {http://arxiv.org/abs/2002.04016} {arXiv:2002.04016 [quant-ph]} \BibitemShut {NoStop}%
\bibitem [{\citenamefont {Kirby}\ \emph {et~al.}(2021)\citenamefont {Kirby}, \citenamefont {Hadi}, \citenamefont {Kreshchuk},\ and\ \citenamefont {Love}}]{Kirby:2021ajp}%
  \BibitemOpen
  \bibfield  {author} {\bibinfo {author} {\bibfnamefont {W.~M.}\ \bibnamefont {Kirby}}, \bibinfo {author} {\bibfnamefont {S.}~\bibnamefont {Hadi}}, \bibinfo {author} {\bibfnamefont {M.}~\bibnamefont {Kreshchuk}}, \ and\ \bibinfo {author} {\bibfnamefont {P.~J.}\ \bibnamefont {Love}},\ }\href {\doibase 10.1103/PhysRevA.104.042607} {\bibfield  {journal} {\bibinfo  {journal} {Phys. Rev. A}\ }\textbf {\bibinfo {volume} {104}},\ \bibinfo {pages} {042607} (\bibinfo {year} {2021})},\ \Eprint {http://arxiv.org/abs/2105.10941} {arXiv:2105.10941 [quant-ph]} \BibitemShut {NoStop}%
\bibitem [{\citenamefont {Rhodes}\ \emph {et~al.}(2024)\citenamefont {Rhodes}, \citenamefont {Kreshchuk},\ and\ \citenamefont {Pathak}}]{rhodes2024exponential}%
  \BibitemOpen
  \bibfield  {author} {\bibinfo {author} {\bibfnamefont {M.}~\bibnamefont {Rhodes}}, \bibinfo {author} {\bibfnamefont {M.}~\bibnamefont {Kreshchuk}}, \ and\ \bibinfo {author} {\bibfnamefont {S.}~\bibnamefont {Pathak}},\ }\href@noop {} {\bibfield  {journal} {\bibinfo  {journal} {arXiv preprint arXiv:2405.10416}\ } (\bibinfo {year} {2024})}\BibitemShut {NoStop}%
\bibitem [{\citenamefont {Haase}\ \emph {et~al.}(2021)\citenamefont {Haase}, \citenamefont {Dellantonio}, \citenamefont {Celi}, \citenamefont {Paulson}, \citenamefont {Kan}, \citenamefont {Jansen},\ and\ \citenamefont {Muschik}}]{Haase:2020kaj}%
  \BibitemOpen
  \bibfield  {author} {\bibinfo {author} {\bibfnamefont {J.~F.}\ \bibnamefont {Haase}}, \bibinfo {author} {\bibfnamefont {L.}~\bibnamefont {Dellantonio}}, \bibinfo {author} {\bibfnamefont {A.}~\bibnamefont {Celi}}, \bibinfo {author} {\bibfnamefont {D.}~\bibnamefont {Paulson}}, \bibinfo {author} {\bibfnamefont {A.}~\bibnamefont {Kan}}, \bibinfo {author} {\bibfnamefont {K.}~\bibnamefont {Jansen}}, \ and\ \bibinfo {author} {\bibfnamefont {C.~A.}\ \bibnamefont {Muschik}},\ }\href {\doibase 10.22331/q-2021-02-04-393} {\bibfield  {journal} {\bibinfo  {journal} {Quantum}\ }\textbf {\bibinfo {volume} {5}},\ \bibinfo {pages} {393} (\bibinfo {year} {2021})},\ \Eprint {http://arxiv.org/abs/2006.14160} {arXiv:2006.14160 [quant-ph]} \BibitemShut {NoStop}%
\bibitem [{\citenamefont {Ji}\ \emph {et~al.}(2020)\citenamefont {Ji}, \citenamefont {Lamm},\ and\ \citenamefont {Zhu}}]{Ji:2020kjk}%
  \BibitemOpen
  \bibfield  {author} {\bibinfo {author} {\bibfnamefont {Y.}~\bibnamefont {Ji}}, \bibinfo {author} {\bibfnamefont {H.}~\bibnamefont {Lamm}}, \ and\ \bibinfo {author} {\bibfnamefont {S.}~\bibnamefont {Zhu}} (\bibinfo {collaboration} {NuQS}),\ }\href {\doibase 10.1103/PhysRevD.102.114513} {\bibfield  {journal} {\bibinfo  {journal} {Phys. Rev. D}\ }\textbf {\bibinfo {volume} {102}},\ \bibinfo {pages} {114513} (\bibinfo {year} {2020})},\ \Eprint {http://arxiv.org/abs/2005.14221} {arXiv:2005.14221 [hep-lat]} \BibitemShut {NoStop}%
\bibitem [{\citenamefont {Bauer}\ and\ \citenamefont {Grabowska}(2023)}]{Bauer:2021gek}%
  \BibitemOpen
  \bibfield  {author} {\bibinfo {author} {\bibfnamefont {C.~W.}\ \bibnamefont {Bauer}}\ and\ \bibinfo {author} {\bibfnamefont {D.~M.}\ \bibnamefont {Grabowska}},\ }\href {\doibase 10.1103/PhysRevD.107.L031503} {\bibfield  {journal} {\bibinfo  {journal} {Phys. Rev. D}\ }\textbf {\bibinfo {volume} {107}},\ \bibinfo {pages} {L031503} (\bibinfo {year} {2023})},\ \Eprint {http://arxiv.org/abs/2111.08015} {arXiv:2111.08015 [hep-ph]} \BibitemShut {NoStop}%
\bibitem [{\citenamefont {Bauer}\ \emph {et~al.}(2023{\natexlab{b}})\citenamefont {Bauer}, \citenamefont {D'Andrea}, \citenamefont {Freytsis},\ and\ \citenamefont {Grabowska}}]{Bauer:2023jvw}%
  \BibitemOpen
  \bibfield  {author} {\bibinfo {author} {\bibfnamefont {C.~W.}\ \bibnamefont {Bauer}}, \bibinfo {author} {\bibfnamefont {I.}~\bibnamefont {D'Andrea}}, \bibinfo {author} {\bibfnamefont {M.}~\bibnamefont {Freytsis}}, \ and\ \bibinfo {author} {\bibfnamefont {D.~M.}\ \bibnamefont {Grabowska}},\ }\href@noop {} {\  (\bibinfo {year} {2023}{\natexlab{b}})},\ \Eprint {http://arxiv.org/abs/2307.11829} {arXiv:2307.11829 [hep-ph]} \BibitemShut {NoStop}%
\bibitem [{\citenamefont {Grabowska}\ \emph {et~al.}(2022)\citenamefont {Grabowska}, \citenamefont {Kane}, \citenamefont {Nachman},\ and\ \citenamefont {Bauer}}]{Grabowska:2022uos}%
  \BibitemOpen
  \bibfield  {author} {\bibinfo {author} {\bibfnamefont {D.~M.}\ \bibnamefont {Grabowska}}, \bibinfo {author} {\bibfnamefont {C.}~\bibnamefont {Kane}}, \bibinfo {author} {\bibfnamefont {B.}~\bibnamefont {Nachman}}, \ and\ \bibinfo {author} {\bibfnamefont {C.~W.}\ \bibnamefont {Bauer}},\ }\href@noop {} {\  (\bibinfo {year} {2022})},\ \Eprint {http://arxiv.org/abs/2208.03333} {arXiv:2208.03333 [quant-ph]} \BibitemShut {NoStop}%
\bibitem [{\citenamefont {Kane}\ \emph {et~al.}(2022)\citenamefont {Kane}, \citenamefont {Grabowska}, \citenamefont {Nachman},\ and\ \citenamefont {Bauer}}]{Kane:2022ejm}%
  \BibitemOpen
  \bibfield  {author} {\bibinfo {author} {\bibfnamefont {C.}~\bibnamefont {Kane}}, \bibinfo {author} {\bibfnamefont {D.~M.}\ \bibnamefont {Grabowska}}, \bibinfo {author} {\bibfnamefont {B.}~\bibnamefont {Nachman}}, \ and\ \bibinfo {author} {\bibfnamefont {C.~W.}\ \bibnamefont {Bauer}},\ }\href@noop {} {\  (\bibinfo {year} {2022})},\ \Eprint {http://arxiv.org/abs/2211.10497} {arXiv:2211.10497 [quant-ph]} \BibitemShut {NoStop}%
\bibitem [{\citenamefont {Hastings}\ \emph {et~al.}(2014)\citenamefont {Hastings}, \citenamefont {Wecker}, \citenamefont {Bauer},\ and\ \citenamefont {Troyer}}]{Hastings:2014wyq}%
  \BibitemOpen
  \bibfield  {author} {\bibinfo {author} {\bibfnamefont {M.~B.}\ \bibnamefont {Hastings}}, \bibinfo {author} {\bibfnamefont {D.}~\bibnamefont {Wecker}}, \bibinfo {author} {\bibfnamefont {B.}~\bibnamefont {Bauer}}, \ and\ \bibinfo {author} {\bibfnamefont {M.}~\bibnamefont {Troyer}},\ }\href@noop {} {\  (\bibinfo {year} {2014})},\ \Eprint {http://arxiv.org/abs/1403.1539} {arXiv:1403.1539 [quant-ph]} \BibitemShut {NoStop}%
\bibitem [{\citenamefont {Poulin}\ \emph {et~al.}(2015)\citenamefont {Poulin}, \citenamefont {Hastings}, \citenamefont {Wecker}, \citenamefont {Wiebe}, \citenamefont {Doberty},\ and\ \citenamefont {Troyer}}]{Poulin:2014yec}%
  \BibitemOpen
  \bibfield  {author} {\bibinfo {author} {\bibfnamefont {D.}~\bibnamefont {Poulin}}, \bibinfo {author} {\bibfnamefont {M.~B.}\ \bibnamefont {Hastings}}, \bibinfo {author} {\bibfnamefont {D.}~\bibnamefont {Wecker}}, \bibinfo {author} {\bibfnamefont {N.}~\bibnamefont {Wiebe}}, \bibinfo {author} {\bibfnamefont {A.~C.}\ \bibnamefont {Doberty}}, \ and\ \bibinfo {author} {\bibfnamefont {M.}~\bibnamefont {Troyer}},\ }\href {\doibase 10.26421/QIC15.5-6-1} {\bibfield  {journal} {\bibinfo  {journal} {Quant. Inf. Comput.}\ }\textbf {\bibinfo {volume} {15}},\ \bibinfo {pages} {0361} (\bibinfo {year} {2015})},\ \Eprint {http://arxiv.org/abs/1406.4920} {arXiv:1406.4920 [quant-ph]} \BibitemShut {NoStop}%
\bibitem [{\citenamefont {Tranter}\ \emph {et~al.}(2018)\citenamefont {Tranter}, \citenamefont {Love}, \citenamefont {Mintert},\ and\ \citenamefont {Coveney}}]{tranter2018comparison}%
  \BibitemOpen
  \bibfield  {author} {\bibinfo {author} {\bibfnamefont {A.}~\bibnamefont {Tranter}}, \bibinfo {author} {\bibfnamefont {P.~J.}\ \bibnamefont {Love}}, \bibinfo {author} {\bibfnamefont {F.}~\bibnamefont {Mintert}}, \ and\ \bibinfo {author} {\bibfnamefont {P.~V.}\ \bibnamefont {Coveney}},\ }\href@noop {} {\bibfield  {journal} {\bibinfo  {journal} {Journal of chemical theory and computation}\ }\textbf {\bibinfo {volume} {14}},\ \bibinfo {pages} {5617} (\bibinfo {year} {2018})}\BibitemShut {NoStop}%
\bibitem [{\citenamefont {Hu}\ \emph {et~al.}(2022)\citenamefont {Hu}, \citenamefont {Meng}, \citenamefont {Wang}, \citenamefont {Luan}, \citenamefont {Fu}, \citenamefont {Zhang}, \citenamefont {Zhang},\ and\ \citenamefont {Yu}}]{Hu:2022nyd}%
  \BibitemOpen
  \bibfield  {author} {\bibinfo {author} {\bibfnamefont {Y.}~\bibnamefont {Hu}}, \bibinfo {author} {\bibfnamefont {F.}~\bibnamefont {Meng}}, \bibinfo {author} {\bibfnamefont {X.}~\bibnamefont {Wang}}, \bibinfo {author} {\bibfnamefont {T.}~\bibnamefont {Luan}}, \bibinfo {author} {\bibfnamefont {Y.}~\bibnamefont {Fu}}, \bibinfo {author} {\bibfnamefont {Z.}~\bibnamefont {Zhang}}, \bibinfo {author} {\bibfnamefont {X.}~\bibnamefont {Zhang}}, \ and\ \bibinfo {author} {\bibfnamefont {X.}~\bibnamefont {Yu}},\ }\href {\doibase 10.1088/2058-9565/ac796b} {\bibfield  {journal} {\bibinfo  {journal} {Quantum Sci. Technol.}\ }\textbf {\bibinfo {volume} {7}},\ \bibinfo {pages} {045001} (\bibinfo {year} {2022})}\BibitemShut {NoStop}%
\bibitem [{\citenamefont {Tranter}\ \emph {et~al.}(2019)\citenamefont {Tranter}, \citenamefont {Love}, \citenamefont {Mintert}, \citenamefont {Wiebe},\ and\ \citenamefont {Coveney}}]{tranter2019ordering}%
  \BibitemOpen
  \bibfield  {author} {\bibinfo {author} {\bibfnamefont {A.}~\bibnamefont {Tranter}}, \bibinfo {author} {\bibfnamefont {P.~J.}\ \bibnamefont {Love}}, \bibinfo {author} {\bibfnamefont {F.}~\bibnamefont {Mintert}}, \bibinfo {author} {\bibfnamefont {N.}~\bibnamefont {Wiebe}}, \ and\ \bibinfo {author} {\bibfnamefont {P.~V.}\ \bibnamefont {Coveney}},\ }\href@noop {} {\bibfield  {journal} {\bibinfo  {journal} {Entropy}\ }\textbf {\bibinfo {volume} {21}},\ \bibinfo {pages} {1218} (\bibinfo {year} {2019})}\BibitemShut {NoStop}%
\bibitem [{\citenamefont {Schmitz}\ \emph {et~al.}(2021)\citenamefont {Schmitz}, \citenamefont {Sawaya}, \citenamefont {Johri},\ and\ \citenamefont {Matsuura}}]{Schmitz:2021ruv}%
  \BibitemOpen
  \bibfield  {author} {\bibinfo {author} {\bibfnamefont {A.~T.}\ \bibnamefont {Schmitz}}, \bibinfo {author} {\bibfnamefont {N.~P.~D.}\ \bibnamefont {Sawaya}}, \bibinfo {author} {\bibfnamefont {S.}~\bibnamefont {Johri}}, \ and\ \bibinfo {author} {\bibfnamefont {A.~Y.}\ \bibnamefont {Matsuura}},\ }\href@noop {} {\  (\bibinfo {year} {2021})},\ \Eprint {http://arxiv.org/abs/2103.08602} {arXiv:2103.08602 [quant-ph]} \BibitemShut {NoStop}%
\bibitem [{\citenamefont {Ostmeyer}(2023)}]{Ostmeyer:2022lxs}%
  \BibitemOpen
  \bibfield  {author} {\bibinfo {author} {\bibfnamefont {J.}~\bibnamefont {Ostmeyer}},\ }\href {\doibase 10.1088/1751-8121/acde7a} {\bibfield  {journal} {\bibinfo  {journal} {J. Phys. A}\ }\textbf {\bibinfo {volume} {56}},\ \bibinfo {pages} {285303} (\bibinfo {year} {2023})},\ \Eprint {http://arxiv.org/abs/2211.02691} {arXiv:2211.02691 [quant-ph]} \BibitemShut {NoStop}%
\bibitem [{\citenamefont {Berry}\ \emph {et~al.}(2007)\citenamefont {Berry}, \citenamefont {Ahokas}, \citenamefont {Cleve},\ and\ \citenamefont {Sanders}}]{Berry:2005yrf}%
  \BibitemOpen
  \bibfield  {author} {\bibinfo {author} {\bibfnamefont {D.~W.}\ \bibnamefont {Berry}}, \bibinfo {author} {\bibfnamefont {G.}~\bibnamefont {Ahokas}}, \bibinfo {author} {\bibfnamefont {R.}~\bibnamefont {Cleve}}, \ and\ \bibinfo {author} {\bibfnamefont {B.~C.}\ \bibnamefont {Sanders}},\ }\href {\doibase 10.1007/s00220-006-0150-x} {\bibfield  {journal} {\bibinfo  {journal} {Commun. Math. Phys.}\ }\textbf {\bibinfo {volume} {270}},\ \bibinfo {pages} {359} (\bibinfo {year} {2007})},\ \Eprint {http://arxiv.org/abs/quant-ph/0508139} {arXiv:quant-ph/0508139} \BibitemShut {NoStop}%
\bibitem [{\citenamefont {Huang}\ \emph {et~al.}(2021)\citenamefont {Huang}, \citenamefont {Yu}, \citenamefont {Lu}, \citenamefont {Yang}, \citenamefont {Li}, \citenamefont {Wu}, \citenamefont {Wu},\ and\ \citenamefont {Chen}}]{Huang:2021pwq}%
  \BibitemOpen
  \bibfield  {author} {\bibinfo {author} {\bibfnamefont {X.-Y.}\ \bibnamefont {Huang}}, \bibinfo {author} {\bibfnamefont {L.}~\bibnamefont {Yu}}, \bibinfo {author} {\bibfnamefont {X.}~\bibnamefont {Lu}}, \bibinfo {author} {\bibfnamefont {Y.}~\bibnamefont {Yang}}, \bibinfo {author} {\bibfnamefont {D.-S.}\ \bibnamefont {Li}}, \bibinfo {author} {\bibfnamefont {C.-W.}\ \bibnamefont {Wu}}, \bibinfo {author} {\bibfnamefont {W.}~\bibnamefont {Wu}}, \ and\ \bibinfo {author} {\bibfnamefont {P.-X.}\ \bibnamefont {Chen}},\ }\href@noop {} {\enquote {\bibinfo {title} {{Qubitization of Bosons}},}\ } (\bibinfo {year} {2021}),\ \Eprint {http://arxiv.org/abs/2105.12563} {arXiv:2105.12563 [quant-ph]} \BibitemShut {NoStop}%
\bibitem [{\citenamefont {Alexandru}\ \emph {et~al.}(2023)\citenamefont {Alexandru}, \citenamefont {Bedaque}, \citenamefont {Carosso}, \citenamefont {Cervia},\ and\ \citenamefont {Sheng}}]{Alexandru:2022son}%
  \BibitemOpen
  \bibfield  {author} {\bibinfo {author} {\bibfnamefont {A.}~\bibnamefont {Alexandru}}, \bibinfo {author} {\bibfnamefont {P.~F.}\ \bibnamefont {Bedaque}}, \bibinfo {author} {\bibfnamefont {A.}~\bibnamefont {Carosso}}, \bibinfo {author} {\bibfnamefont {M.~J.}\ \bibnamefont {Cervia}}, \ and\ \bibinfo {author} {\bibfnamefont {A.}~\bibnamefont {Sheng}},\ }\href {\doibase 10.1103/PhysRevD.107.034503} {\bibfield  {journal} {\bibinfo  {journal} {Phys. Rev. D}\ }\textbf {\bibinfo {volume} {107}},\ \bibinfo {pages} {034503} (\bibinfo {year} {2023})},\ \Eprint {http://arxiv.org/abs/2209.00098} {arXiv:2209.00098 [hep-lat]} \BibitemShut {NoStop}%
\bibitem [{\citenamefont {Babbush}\ \emph {et~al.}(2019)\citenamefont {Babbush}, \citenamefont {Berry}, \citenamefont {McClean},\ and\ \citenamefont {Neven}}]{Babbush:2019yrx}%
  \BibitemOpen
  \bibfield  {author} {\bibinfo {author} {\bibfnamefont {R.}~\bibnamefont {Babbush}}, \bibinfo {author} {\bibfnamefont {D.~W.}\ \bibnamefont {Berry}}, \bibinfo {author} {\bibfnamefont {J.~R.}\ \bibnamefont {McClean}}, \ and\ \bibinfo {author} {\bibfnamefont {H.}~\bibnamefont {Neven}},\ }\href {\doibase 10.1038/s41534-019-0199-y} {\bibfield  {journal} {\bibinfo  {journal} {npj Quantum Inf.}\ }\textbf {\bibinfo {volume} {5}},\ \bibinfo {pages} {92} (\bibinfo {year} {2019})}\BibitemShut {NoStop}%
\bibitem [{\citenamefont {Camps}\ and\ \citenamefont {Van~Beeumen}(2022)}]{camps2022fable}%
  \BibitemOpen
  \bibfield  {author} {\bibinfo {author} {\bibfnamefont {D.}~\bibnamefont {Camps}}\ and\ \bibinfo {author} {\bibfnamefont {R.}~\bibnamefont {Van~Beeumen}},\ }in\ \href@noop {} {\emph {\bibinfo {booktitle} {2022 IEEE International Conference on Quantum Computing and Engineering (QCE)}}}\ (\bibinfo {organization} {IEEE},\ \bibinfo {year} {2022})\ pp.\ \bibinfo {pages} {104--113}\BibitemShut {NoStop}%
\bibitem [{\citenamefont {Lin}(2022)}]{Lin:2022vrd}%
  \BibitemOpen
  \bibfield  {author} {\bibinfo {author} {\bibfnamefont {L.}~\bibnamefont {Lin}},\ }\href@noop {} {\  (\bibinfo {year} {2022})},\ \Eprint {http://arxiv.org/abs/2201.08309} {arXiv:2201.08309 [quant-ph]} \BibitemShut {NoStop}%
\bibitem [{\citenamefont {Dong}\ \emph {et~al.}(2022)\citenamefont {Dong}, \citenamefont {Lin},\ and\ \citenamefont {Tong}}]{Dong:2022mmq}%
  \BibitemOpen
  \bibfield  {author} {\bibinfo {author} {\bibfnamefont {Y.}~\bibnamefont {Dong}}, \bibinfo {author} {\bibfnamefont {L.}~\bibnamefont {Lin}}, \ and\ \bibinfo {author} {\bibfnamefont {Y.}~\bibnamefont {Tong}},\ }\href {\doibase 10.1103/PRXQuantum.3.040305} {\bibfield  {journal} {\bibinfo  {journal} {PRX Quantum}\ }\textbf {\bibinfo {volume} {3}},\ \bibinfo {pages} {040305} (\bibinfo {year} {2022})},\ \Eprint {http://arxiv.org/abs/2204.05955} {arXiv:2204.05955 [quant-ph]} \BibitemShut {NoStop}%
\bibitem [{\citenamefont {Babbush}\ \emph {et~al.}(2017)\citenamefont {Babbush}, \citenamefont {Berry}, \citenamefont {Sanders}, \citenamefont {Kivlichan}, \citenamefont {Scherer}, \citenamefont {Wei}, \citenamefont {Love},\ and\ \citenamefont {Aspuru-Guzik}}]{Babbush:2017oum}%
  \BibitemOpen
  \bibfield  {author} {\bibinfo {author} {\bibfnamefont {R.}~\bibnamefont {Babbush}}, \bibinfo {author} {\bibfnamefont {D.~W.}\ \bibnamefont {Berry}}, \bibinfo {author} {\bibfnamefont {Y.~R.}\ \bibnamefont {Sanders}}, \bibinfo {author} {\bibfnamefont {I.~D.}\ \bibnamefont {Kivlichan}}, \bibinfo {author} {\bibfnamefont {A.}~\bibnamefont {Scherer}}, \bibinfo {author} {\bibfnamefont {A.~Y.}\ \bibnamefont {Wei}}, \bibinfo {author} {\bibfnamefont {P.~J.}\ \bibnamefont {Love}}, \ and\ \bibinfo {author} {\bibfnamefont {A.}~\bibnamefont {Aspuru-Guzik}},\ }\href {\doibase 10.1088/2058-9565/aa9463} {\bibfield  {journal} {\bibinfo  {journal} {Quantum Sci. Technol.}\ }\textbf {\bibinfo {volume} {3}},\ \bibinfo {pages} {015006} (\bibinfo {year} {2017})}\BibitemShut {NoStop}%
\bibitem [{\citenamefont {Kirby}\ \emph {et~al.}(2023)\citenamefont {Kirby}, \citenamefont {Motta},\ and\ \citenamefont {Mezzacapo}}]{Kirby:2022ncy}%
  \BibitemOpen
  \bibfield  {author} {\bibinfo {author} {\bibfnamefont {W.}~\bibnamefont {Kirby}}, \bibinfo {author} {\bibfnamefont {M.}~\bibnamefont {Motta}}, \ and\ \bibinfo {author} {\bibfnamefont {A.}~\bibnamefont {Mezzacapo}},\ }\href {\doibase 10.22331/q-2023-05-23-1018} {\bibfield  {journal} {\bibinfo  {journal} {Quantum}\ }\textbf {\bibinfo {volume} {7}},\ \bibinfo {pages} {1018} (\bibinfo {year} {2023})},\ \Eprint {http://arxiv.org/abs/2208.00567} {arXiv:2208.00567 [quant-ph]} \BibitemShut {NoStop}%
\bibitem [{\citenamefont {Gily{\'e}n}\ \emph {et~al.}(2019)\citenamefont {Gily{\'e}n}, \citenamefont {Su}, \citenamefont {Low},\ and\ \citenamefont {Wiebe}}]{gilyen2019quantum}%
  \BibitemOpen
  \bibfield  {author} {\bibinfo {author} {\bibfnamefont {A.}~\bibnamefont {Gily{\'e}n}}, \bibinfo {author} {\bibfnamefont {Y.}~\bibnamefont {Su}}, \bibinfo {author} {\bibfnamefont {G.~H.}\ \bibnamefont {Low}}, \ and\ \bibinfo {author} {\bibfnamefont {N.}~\bibnamefont {Wiebe}},\ }in\ \href@noop {} {\emph {\bibinfo {booktitle} {Proceedings of the 51st Annual ACM SIGACT Symposium on Theory of Computing}}}\ (\bibinfo {year} {2019})\ pp.\ \bibinfo {pages} {193--204}\BibitemShut {NoStop}%
\bibitem [{\citenamefont {Novikau}\ \emph {et~al.}(2022)\citenamefont {Novikau}, \citenamefont {Startsev},\ and\ \citenamefont {Dodin}}]{Novikau:2021yra}%
  \BibitemOpen
  \bibfield  {author} {\bibinfo {author} {\bibfnamefont {I.}~\bibnamefont {Novikau}}, \bibinfo {author} {\bibfnamefont {E.~A.}\ \bibnamefont {Startsev}}, \ and\ \bibinfo {author} {\bibfnamefont {I.~Y.}\ \bibnamefont {Dodin}},\ }\href {\doibase 10.1103/PhysRevA.105.062444} {\bibfield  {journal} {\bibinfo  {journal} {Phys. Rev. A}\ }\textbf {\bibinfo {volume} {105}},\ \bibinfo {pages} {062444} (\bibinfo {year} {2022})},\ \Eprint {http://arxiv.org/abs/2112.06086} {arXiv:2112.06086 [physics.plasm-ph]} \BibitemShut {NoStop}%
\bibitem [{\citenamefont {Low}\ \emph {et~al.}(2016)\citenamefont {Low}, \citenamefont {Yoder},\ and\ \citenamefont {Chuang}}]{Low:2016jxt}%
  \BibitemOpen
  \bibfield  {author} {\bibinfo {author} {\bibfnamefont {G.~H.}\ \bibnamefont {Low}}, \bibinfo {author} {\bibfnamefont {T.~J.}\ \bibnamefont {Yoder}}, \ and\ \bibinfo {author} {\bibfnamefont {I.~L.}\ \bibnamefont {Chuang}},\ }\href {\doibase 10.1103/PhysRevX.6.041067} {\bibfield  {journal} {\bibinfo  {journal} {Phys. Rev. X}\ }\textbf {\bibinfo {volume} {6}},\ \bibinfo {pages} {041067} (\bibinfo {year} {2016})}\BibitemShut {NoStop}%
\bibitem [{\citenamefont {Chan}\ \emph {et~al.}(2023)\citenamefont {Chan}, \citenamefont {Ramo},\ and\ \citenamefont {Fitzpatrick}}]{Chan:2023lbm}%
  \BibitemOpen
  \bibfield  {author} {\bibinfo {author} {\bibfnamefont {H.~H.~S.}\ \bibnamefont {Chan}}, \bibinfo {author} {\bibfnamefont {D.~M.~n.}\ \bibnamefont {Ramo}}, \ and\ \bibinfo {author} {\bibfnamefont {N.}~\bibnamefont {Fitzpatrick}},\ }\href@noop {} {\  (\bibinfo {year} {2023})},\ \Eprint {http://arxiv.org/abs/2303.06161} {arXiv:2303.06161 [quant-ph]} \BibitemShut {NoStop}%
\bibitem [{\citenamefont {Kane}\ \emph {et~al.}(2023)\citenamefont {Kane}, \citenamefont {Gomes},\ and\ \citenamefont {Kreshchuk}}]{Kane:2023jdo}%
  \BibitemOpen
  \bibfield  {author} {\bibinfo {author} {\bibfnamefont {C.~F.}\ \bibnamefont {Kane}}, \bibinfo {author} {\bibfnamefont {N.}~\bibnamefont {Gomes}}, \ and\ \bibinfo {author} {\bibfnamefont {M.}~\bibnamefont {Kreshchuk}},\ }\href@noop {} {\  (\bibinfo {year} {2023})},\ \Eprint {http://arxiv.org/abs/2310.13757} {arXiv:2310.13757 [quant-ph]} \BibitemShut {NoStop}%
\bibitem [{\citenamefont {Lin}\ and\ \citenamefont {Tong}(2020)}]{Lin:2020zni}%
  \BibitemOpen
  \bibfield  {author} {\bibinfo {author} {\bibfnamefont {L.}~\bibnamefont {Lin}}\ and\ \bibinfo {author} {\bibfnamefont {Y.}~\bibnamefont {Tong}},\ }\href {\doibase 10.22331/q-2020-12-14-372} {\bibfield  {journal} {\bibinfo  {journal} {Quantum}\ }\textbf {\bibinfo {volume} {4}},\ \bibinfo {pages} {372} (\bibinfo {year} {2020})},\ \Eprint {http://arxiv.org/abs/2002.12508} {arXiv:2002.12508 [quant-ph]} \BibitemShut {NoStop}%
\bibitem [{\citenamefont {Kane}\ \emph {et~al.}(2024)\citenamefont {Kane}, \citenamefont {Hariprakash}, \citenamefont {Modi}, \citenamefont {Kreshchuk},\ and\ \citenamefont {Bauer}}]{bebyqsp}%
  \BibitemOpen
  \bibfield  {author} {\bibinfo {author} {\bibfnamefont {C.~F.}\ \bibnamefont {Kane}}, \bibinfo {author} {\bibfnamefont {S.}~\bibnamefont {Hariprakash}}, \bibinfo {author} {\bibfnamefont {N.~S.}\ \bibnamefont {Modi}}, \bibinfo {author} {\bibfnamefont {M.}~\bibnamefont {Kreshchuk}}, \ and\ \bibinfo {author} {\bibfnamefont {C.~W.}\ \bibnamefont {Bauer}},\ }\href@noop {} {\enquote {\bibinfo {title} {``{B}lock encoding by signal processing (manuscript in preparation)''},}\ } (\bibinfo {year} {2024})\BibitemShut {NoStop}%
\bibitem [{\citenamefont {Lieb}\ and\ \citenamefont {Robinson}(1972)}]{Lieb:1972wy}%
  \BibitemOpen
  \bibfield  {author} {\bibinfo {author} {\bibfnamefont {E.~H.}\ \bibnamefont {Lieb}}\ and\ \bibinfo {author} {\bibfnamefont {D.~W.}\ \bibnamefont {Robinson}},\ }\href {\doibase 10.1007/BF01645779} {\bibfield  {journal} {\bibinfo  {journal} {Commun. Math. Phys.}\ }\textbf {\bibinfo {volume} {28}},\ \bibinfo {pages} {251} (\bibinfo {year} {1972})}\BibitemShut {NoStop}%
\bibitem [{\citenamefont {Shende}\ \emph {et~al.}(2006)\citenamefont {Shende}, \citenamefont {Bullock},\ and\ \citenamefont {Markov}}]{Shende:2006onn}%
  \BibitemOpen
  \bibfield  {author} {\bibinfo {author} {\bibfnamefont {V.~V.}\ \bibnamefont {Shende}}, \bibinfo {author} {\bibfnamefont {S.~S.}\ \bibnamefont {Bullock}}, \ and\ \bibinfo {author} {\bibfnamefont {I.~L.}\ \bibnamefont {Markov}},\ }\href {\doibase 10.1109/tcad.2005.855930} {\bibfield  {journal} {\bibinfo  {journal} {IEEE Trans. Comput. Aided Design Integr. Circuits Syst.}\ }\textbf {\bibinfo {volume} {25}},\ \bibinfo {pages} {1000} (\bibinfo {year} {2006})}\BibitemShut {NoStop}%
\bibitem [{\citenamefont {Plesch}\ and\ \citenamefont {Brukner}(2011)}]{Plesch:2011vwn}%
  \BibitemOpen
  \bibfield  {author} {\bibinfo {author} {\bibfnamefont {M.}~\bibnamefont {Plesch}}\ and\ \bibinfo {author} {\bibfnamefont {v.}~\bibnamefont {Brukner}},\ }\href {\doibase 10.1103/PhysRevA.83.032302} {\bibfield  {journal} {\bibinfo  {journal} {Phys. Rev. A}\ }\textbf {\bibinfo {volume} {83}},\ \bibinfo {pages} {032302} (\bibinfo {year} {2011})}\BibitemShut {NoStop}%
\bibitem [{\citenamefont {Nielsen}\ and\ \citenamefont {Chuang}(2012)}]{Nielsen:2012yss}%
  \BibitemOpen
  \bibfield  {author} {\bibinfo {author} {\bibfnamefont {M.~A.}\ \bibnamefont {Nielsen}}\ and\ \bibinfo {author} {\bibfnamefont {I.~L.}\ \bibnamefont {Chuang}},\ }\href {\doibase 10.1017/cbo9780511976667} {\emph {\bibinfo {title} {{Quantum Computation and Quantum Information}}}}\ (\bibinfo  {publisher} {Cambridge University Press},\ \bibinfo {year} {2012})\BibitemShut {NoStop}%
\bibitem [{\citenamefont {Macridin}\ \emph {et~al.}(2018{\natexlab{a}})\citenamefont {Macridin}, \citenamefont {Spentzouris}, \citenamefont {Amundson},\ and\ \citenamefont {Harnik}}]{Macridin:2018gdw}%
  \BibitemOpen
  \bibfield  {author} {\bibinfo {author} {\bibfnamefont {A.}~\bibnamefont {Macridin}}, \bibinfo {author} {\bibfnamefont {P.}~\bibnamefont {Spentzouris}}, \bibinfo {author} {\bibfnamefont {J.}~\bibnamefont {Amundson}}, \ and\ \bibinfo {author} {\bibfnamefont {R.}~\bibnamefont {Harnik}},\ }\href {\doibase 10.1103/PhysRevLett.121.110504} {\bibfield  {journal} {\bibinfo  {journal} {Phys. Rev. Lett.}\ }\textbf {\bibinfo {volume} {121}},\ \bibinfo {pages} {110504} (\bibinfo {year} {2018}{\natexlab{a}})},\ \Eprint {http://arxiv.org/abs/1802.07347} {arXiv:1802.07347 [quant-ph]} \BibitemShut {NoStop}%
\bibitem [{\citenamefont {Macridin}\ \emph {et~al.}(2018{\natexlab{b}})\citenamefont {Macridin}, \citenamefont {Spentzouris}, \citenamefont {Amundson},\ and\ \citenamefont {Harnik}}]{Macridin:2018oli}%
  \BibitemOpen
  \bibfield  {author} {\bibinfo {author} {\bibfnamefont {A.}~\bibnamefont {Macridin}}, \bibinfo {author} {\bibfnamefont {P.}~\bibnamefont {Spentzouris}}, \bibinfo {author} {\bibfnamefont {J.}~\bibnamefont {Amundson}}, \ and\ \bibinfo {author} {\bibfnamefont {R.}~\bibnamefont {Harnik}},\ }\href {\doibase 10.1103/PhysRevA.98.042312} {\bibfield  {journal} {\bibinfo  {journal} {Phys. Rev. A}\ }\textbf {\bibinfo {volume} {98}},\ \bibinfo {pages} {042312} (\bibinfo {year} {2018}{\natexlab{b}})},\ \Eprint {http://arxiv.org/abs/1805.09928} {arXiv:1805.09928 [quant-ph]} \BibitemShut {NoStop}%
\bibitem [{\citenamefont {Macridin}\ \emph {et~al.}(2022)\citenamefont {Macridin}, \citenamefont {Li}, \citenamefont {Mrenna},\ and\ \citenamefont {Spentzouris}}]{Macridin:2021uwn}%
  \BibitemOpen
  \bibfield  {author} {\bibinfo {author} {\bibfnamefont {A.}~\bibnamefont {Macridin}}, \bibinfo {author} {\bibfnamefont {A.~C.~Y.}\ \bibnamefont {Li}}, \bibinfo {author} {\bibfnamefont {S.}~\bibnamefont {Mrenna}}, \ and\ \bibinfo {author} {\bibfnamefont {P.}~\bibnamefont {Spentzouris}},\ }\href {\doibase 10.1103/PhysRevA.105.052405} {\bibfield  {journal} {\bibinfo  {journal} {Phys. Rev. A}\ }\textbf {\bibinfo {volume} {105}},\ \bibinfo {pages} {052405} (\bibinfo {year} {2022})},\ \Eprint {http://arxiv.org/abs/2108.10793} {arXiv:2108.10793 [quant-ph]} \BibitemShut {NoStop}%
\bibitem [{\citenamefont {Farrelly}\ and\ \citenamefont {Streich}(2020)}]{Farrelly:2020ckc}%
  \BibitemOpen
  \bibfield  {author} {\bibinfo {author} {\bibfnamefont {T.}~\bibnamefont {Farrelly}}\ and\ \bibinfo {author} {\bibfnamefont {J.}~\bibnamefont {Streich}},\ }\href@noop {} {\  (\bibinfo {year} {2020})},\ \Eprint {http://arxiv.org/abs/2002.02643} {arXiv:2002.02643 [quant-ph]} \BibitemShut {NoStop}%
\bibitem [{\citenamefont {Bauer}\ \emph {et~al.}(2021)\citenamefont {Bauer}, \citenamefont {Freytsis},\ and\ \citenamefont {Nachman}}]{Bauer:2021gup}%
  \BibitemOpen
  \bibfield  {author} {\bibinfo {author} {\bibfnamefont {C.~W.}\ \bibnamefont {Bauer}}, \bibinfo {author} {\bibfnamefont {M.}~\bibnamefont {Freytsis}}, \ and\ \bibinfo {author} {\bibfnamefont {B.}~\bibnamefont {Nachman}},\ }\href {\doibase 10.1103/PhysRevLett.127.212001} {\bibfield  {journal} {\bibinfo  {journal} {Phys. Rev. Lett.}\ }\textbf {\bibinfo {volume} {127}},\ \bibinfo {pages} {212001} (\bibinfo {year} {2021})},\ \Eprint {http://arxiv.org/abs/2102.05044} {arXiv:2102.05044 [hep-ph]} \BibitemShut {NoStop}%
\bibitem [{\citenamefont {Kan}\ and\ \citenamefont {Nam}(2021)}]{Kan:2021xfc}%
  \BibitemOpen
  \bibfield  {author} {\bibinfo {author} {\bibfnamefont {A.}~\bibnamefont {Kan}}\ and\ \bibinfo {author} {\bibfnamefont {Y.}~\bibnamefont {Nam}},\ }\href@noop {} {\  (\bibinfo {year} {2021})},\ \Eprint {http://arxiv.org/abs/2107.12769} {arXiv:2107.12769 [quant-ph]} \BibitemShut {NoStop}%
\bibitem [{\citenamefont {Davoudi}\ \emph {et~al.}(2023)\citenamefont {Davoudi}, \citenamefont {Shaw},\ and\ \citenamefont {Stryker}}]{Davoudi:2022xmb}%
  \BibitemOpen
  \bibfield  {author} {\bibinfo {author} {\bibfnamefont {Z.}~\bibnamefont {Davoudi}}, \bibinfo {author} {\bibfnamefont {A.~F.}\ \bibnamefont {Shaw}}, \ and\ \bibinfo {author} {\bibfnamefont {J.~R.}\ \bibnamefont {Stryker}},\ }\href {\doibase 10.22331/q-2023-12-20-1213} {\bibfield  {journal} {\bibinfo  {journal} {Quantum}\ }\textbf {\bibinfo {volume} {7}},\ \bibinfo {pages} {1213} (\bibinfo {year} {2023})},\ \Eprint {http://arxiv.org/abs/2212.14030} {arXiv:2212.14030 [hep-lat]} \BibitemShut {NoStop}%
\bibitem [{\citenamefont {Brennen}\ \emph {et~al.}(2015)\citenamefont {Brennen}, \citenamefont {Rohde}, \citenamefont {Sanders},\ and\ \citenamefont {Singh}}]{Brennen:2014iqu}%
  \BibitemOpen
  \bibfield  {author} {\bibinfo {author} {\bibfnamefont {G.~K.}\ \bibnamefont {Brennen}}, \bibinfo {author} {\bibfnamefont {P.}~\bibnamefont {Rohde}}, \bibinfo {author} {\bibfnamefont {B.~C.}\ \bibnamefont {Sanders}}, \ and\ \bibinfo {author} {\bibfnamefont {S.}~\bibnamefont {Singh}},\ }\href {\doibase 10.1103/PhysRevA.92.032315} {\bibfield  {journal} {\bibinfo  {journal} {Phys. Rev. A}\ }\textbf {\bibinfo {volume} {92}},\ \bibinfo {pages} {032315} (\bibinfo {year} {2015})},\ \Eprint {http://arxiv.org/abs/1412.0750} {arXiv:1412.0750 [quant-ph]} \BibitemShut {NoStop}%
\bibitem [{\citenamefont {Low}\ and\ \citenamefont {Wiebe}(2018)}]{Low:2018pte}%
  \BibitemOpen
  \bibfield  {author} {\bibinfo {author} {\bibfnamefont {G.~H.}\ \bibnamefont {Low}}\ and\ \bibinfo {author} {\bibfnamefont {N.}~\bibnamefont {Wiebe}},\ }\href@noop {} {\  (\bibinfo {year} {2018})},\ \Eprint {http://arxiv.org/abs/1805.00675} {arXiv:1805.00675 [quant-ph]} \BibitemShut {NoStop}%
\bibitem [{\citenamefont {Kieferov\'a}\ \emph {et~al.}(2019)\citenamefont {Kieferov\'a}, \citenamefont {Scherer},\ and\ \citenamefont {Berry}}]{PhysRevA.99.042314}%
  \BibitemOpen
  \bibfield  {author} {\bibinfo {author} {\bibfnamefont {M.}~\bibnamefont {Kieferov\'a}}, \bibinfo {author} {\bibfnamefont {A.}~\bibnamefont {Scherer}}, \ and\ \bibinfo {author} {\bibfnamefont {D.~W.}\ \bibnamefont {Berry}},\ }\href {\doibase 10.1103/PhysRevA.99.042314} {\bibfield  {journal} {\bibinfo  {journal} {Phys. Rev. A}\ }\textbf {\bibinfo {volume} {99}},\ \bibinfo {pages} {042314} (\bibinfo {year} {2019})}\BibitemShut {NoStop}%
\bibitem [{\citenamefont {Tong}\ \emph {et~al.}(2022)\citenamefont {Tong}, \citenamefont {Albert}, \citenamefont {McClean}, \citenamefont {Preskill},\ and\ \citenamefont {Su}}]{Tong:2021rfv}%
  \BibitemOpen
  \bibfield  {author} {\bibinfo {author} {\bibfnamefont {Y.}~\bibnamefont {Tong}}, \bibinfo {author} {\bibfnamefont {V.~V.}\ \bibnamefont {Albert}}, \bibinfo {author} {\bibfnamefont {J.~R.}\ \bibnamefont {McClean}}, \bibinfo {author} {\bibfnamefont {J.}~\bibnamefont {Preskill}}, \ and\ \bibinfo {author} {\bibfnamefont {Y.}~\bibnamefont {Su}},\ }\href {\doibase 10.22331/q-2022-09-22-816} {\bibfield  {journal} {\bibinfo  {journal} {Quantum}\ }\textbf {\bibinfo {volume} {6}},\ \bibinfo {pages} {816} (\bibinfo {year} {2022})},\ \Eprint {http://arxiv.org/abs/2110.06942} {arXiv:2110.06942 [quant-ph]} \BibitemShut {NoStop}%
\bibitem [{\citenamefont {Lin}\ and\ \citenamefont {Tong}(2022)}]{Lin:2021rwb}%
  \BibitemOpen
  \bibfield  {author} {\bibinfo {author} {\bibfnamefont {L.}~\bibnamefont {Lin}}\ and\ \bibinfo {author} {\bibfnamefont {Y.}~\bibnamefont {Tong}},\ }\href {\doibase 10.1103/PRXQuantum.3.010318} {\bibfield  {journal} {\bibinfo  {journal} {PRX Quantum}\ }\textbf {\bibinfo {volume} {3}},\ \bibinfo {pages} {010318} (\bibinfo {year} {2022})},\ \Eprint {http://arxiv.org/abs/2102.11340} {arXiv:2102.11340 [quant-ph]} \BibitemShut {NoStop}%
\bibitem [{\citenamefont {Berry}\ \emph {et~al.}(2015)\citenamefont {Berry}, \citenamefont {Childs}, \citenamefont {Cleve}, \citenamefont {Kothari},\ and\ \citenamefont {Somma}}]{Berry:2014ivo}%
  \BibitemOpen
  \bibfield  {author} {\bibinfo {author} {\bibfnamefont {D.~W.}\ \bibnamefont {Berry}}, \bibinfo {author} {\bibfnamefont {A.~M.}\ \bibnamefont {Childs}}, \bibinfo {author} {\bibfnamefont {R.}~\bibnamefont {Cleve}}, \bibinfo {author} {\bibfnamefont {R.}~\bibnamefont {Kothari}}, \ and\ \bibinfo {author} {\bibfnamefont {R.~D.}\ \bibnamefont {Somma}},\ }\href {\doibase 10.1103/PhysRevLett.114.090502} {\bibfield  {journal} {\bibinfo  {journal} {Phys. Rev. Lett.}\ }\textbf {\bibinfo {volume} {114}},\ \bibinfo {pages} {090502} (\bibinfo {year} {2015})},\ \Eprint {http://arxiv.org/abs/1412.4687} {arXiv:1412.4687 [quant-ph]} \BibitemShut {NoStop}%
\bibitem [{\citenamefont {Ciavarella}\ and\ \citenamefont {Bauer}(2024)}]{Ciavarella:2024fzw}%
  \BibitemOpen
  \bibfield  {author} {\bibinfo {author} {\bibfnamefont {A.~N.}\ \bibnamefont {Ciavarella}}\ and\ \bibinfo {author} {\bibfnamefont {C.~W.}\ \bibnamefont {Bauer}},\ }\href@noop {} {\  (\bibinfo {year} {2024})},\ \Eprint {http://arxiv.org/abs/2402.10265} {arXiv:2402.10265 [hep-ph]} \BibitemShut {NoStop}%
\bibitem [{\citenamefont {Bravyi}\ and\ \citenamefont {Kitaev}(2002)}]{Bravyi:2000vfj}%
  \BibitemOpen
  \bibfield  {author} {\bibinfo {author} {\bibfnamefont {S.~B.}\ \bibnamefont {Bravyi}}\ and\ \bibinfo {author} {\bibfnamefont {A.~Y.}\ \bibnamefont {Kitaev}},\ }\href {\doibase 10.1006/aphy.2002.6254} {\bibfield  {journal} {\bibinfo  {journal} {Annals Phys.}\ }\textbf {\bibinfo {volume} {298}},\ \bibinfo {pages} {210} (\bibinfo {year} {2002})},\ \Eprint {http://arxiv.org/abs/quant-ph/0003137} {arXiv:quant-ph/0003137} \BibitemShut {NoStop}%
\bibitem [{\citenamefont {Verstraete}\ and\ \citenamefont {Cirac}(2005)}]{Verstraete:2005pn}%
  \BibitemOpen
  \bibfield  {author} {\bibinfo {author} {\bibfnamefont {F.}~\bibnamefont {Verstraete}}\ and\ \bibinfo {author} {\bibfnamefont {J.~I.}\ \bibnamefont {Cirac}},\ }\href {\doibase 10.1088/1742-5468/2005/09/P09012} {\bibfield  {journal} {\bibinfo  {journal} {J. Stat. Mech.}\ }\textbf {\bibinfo {volume} {0509}},\ \bibinfo {pages} {P09012} (\bibinfo {year} {2005})},\ \Eprint {http://arxiv.org/abs/cond-mat/0508353} {arXiv:cond-mat/0508353} \BibitemShut {NoStop}%
\bibitem [{\citenamefont {Tranter}\ \emph {et~al.}(2015)\citenamefont {Tranter}, \citenamefont {Sofia}, \citenamefont {Seeley}, \citenamefont {Kaicher}, \citenamefont {McClean}, \citenamefont {Babbush}, \citenamefont {Coveney}, \citenamefont {Mintert}, \citenamefont {Wilhelm},\ and\ \citenamefont {Love}}]{tranter2015b}%
  \BibitemOpen
  \bibfield  {author} {\bibinfo {author} {\bibfnamefont {A.}~\bibnamefont {Tranter}}, \bibinfo {author} {\bibfnamefont {S.}~\bibnamefont {Sofia}}, \bibinfo {author} {\bibfnamefont {J.}~\bibnamefont {Seeley}}, \bibinfo {author} {\bibfnamefont {M.}~\bibnamefont {Kaicher}}, \bibinfo {author} {\bibfnamefont {J.}~\bibnamefont {McClean}}, \bibinfo {author} {\bibfnamefont {R.}~\bibnamefont {Babbush}}, \bibinfo {author} {\bibfnamefont {P.~V.}\ \bibnamefont {Coveney}}, \bibinfo {author} {\bibfnamefont {F.}~\bibnamefont {Mintert}}, \bibinfo {author} {\bibfnamefont {F.}~\bibnamefont {Wilhelm}}, \ and\ \bibinfo {author} {\bibfnamefont {P.~J.}\ \bibnamefont {Love}},\ }\href@noop {} {\bibfield  {journal} {\bibinfo  {journal} {International Journal of Quantum Chemistry}\ }\textbf {\bibinfo {volume} {115}},\ \bibinfo {pages} {1431} (\bibinfo {year} {2015})}\BibitemShut {NoStop}%
\bibitem [{\citenamefont {Babbush}\ \emph {et~al.}(2015)\citenamefont {Babbush}, \citenamefont {McClean}, \citenamefont {Wecker}, \citenamefont {Aspuru-Guzik},\ and\ \citenamefont {Wiebe}}]{Babbush:2015utf}%
  \BibitemOpen
  \bibfield  {author} {\bibinfo {author} {\bibfnamefont {R.}~\bibnamefont {Babbush}}, \bibinfo {author} {\bibfnamefont {J.}~\bibnamefont {McClean}}, \bibinfo {author} {\bibfnamefont {D.}~\bibnamefont {Wecker}}, \bibinfo {author} {\bibfnamefont {A.}~\bibnamefont {Aspuru-Guzik}}, \ and\ \bibinfo {author} {\bibfnamefont {N.}~\bibnamefont {Wiebe}},\ }\href {\doibase 10.1103/PhysRevA.91.022311} {\bibfield  {journal} {\bibinfo  {journal} {Phys. Rev. A}\ }\textbf {\bibinfo {volume} {91}},\ \bibinfo {pages} {022311} (\bibinfo {year} {2015})}\BibitemShut {NoStop}%
\bibitem [{\citenamefont {Aspuru-Guzik}\ \emph {et~al.}(2016)\citenamefont {Aspuru-Guzik}, \citenamefont {Love}, \citenamefont {Wei}, \citenamefont {Kivlichan}, \citenamefont {Berry},\ and\ \citenamefont {Babbush}}]{Aspuru-Guzik:2016mom}%
  \BibitemOpen
  \bibfield  {author} {\bibinfo {author} {\bibfnamefont {A.}~\bibnamefont {Aspuru-Guzik}}, \bibinfo {author} {\bibfnamefont {P.~J.}\ \bibnamefont {Love}}, \bibinfo {author} {\bibfnamefont {A.~Y.}\ \bibnamefont {Wei}}, \bibinfo {author} {\bibfnamefont {I.~D.}\ \bibnamefont {Kivlichan}}, \bibinfo {author} {\bibfnamefont {D.~W.}\ \bibnamefont {Berry}}, \ and\ \bibinfo {author} {\bibfnamefont {R.}~\bibnamefont {Babbush}},\ }\href {\doibase 10.1088/1367-2630/18/3/033032} {\bibfield  {journal} {\bibinfo  {journal} {New J. Phys.}\ }\textbf {\bibinfo {volume} {18}},\ \bibinfo {pages} {033032} (\bibinfo {year} {2016})}\BibitemShut {NoStop}%
\bibitem [{\citenamefont {Setia}\ and\ \citenamefont {Whitfield}(2018)}]{Setia:2018qmu}%
  \BibitemOpen
  \bibfield  {author} {\bibinfo {author} {\bibfnamefont {K.}~\bibnamefont {Setia}}\ and\ \bibinfo {author} {\bibfnamefont {J.~D.}\ \bibnamefont {Whitfield}},\ }\href {\doibase 10.1063/1.5019371} {\bibfield  {journal} {\bibinfo  {journal} {J. Chem. Phys.}\ }\textbf {\bibinfo {volume} {148}},\ \bibinfo {pages} {164104} (\bibinfo {year} {2018})}\BibitemShut {NoStop}%
\bibitem [{\citenamefont {Kirby}\ \emph {et~al.}(2022)\citenamefont {Kirby}, \citenamefont {Fuller}, \citenamefont {Hadfield},\ and\ \citenamefont {Mezzacapo}}]{Kirby:2021vkt}%
  \BibitemOpen
  \bibfield  {author} {\bibinfo {author} {\bibfnamefont {W.}~\bibnamefont {Kirby}}, \bibinfo {author} {\bibfnamefont {B.}~\bibnamefont {Fuller}}, \bibinfo {author} {\bibfnamefont {C.}~\bibnamefont {Hadfield}}, \ and\ \bibinfo {author} {\bibfnamefont {A.}~\bibnamefont {Mezzacapo}},\ }\href {\doibase 10.1103/PRXQuantum.3.020351} {\bibfield  {journal} {\bibinfo  {journal} {PRX Quantum}\ }\textbf {\bibinfo {volume} {3}},\ \bibinfo {pages} {020351} (\bibinfo {year} {2022})},\ \Eprint {http://arxiv.org/abs/2109.14465} {arXiv:2109.14465 [quant-ph]} \BibitemShut {NoStop}%
\bibitem [{\citenamefont {Sakamoto}\ \emph {et~al.}(2023)\citenamefont {Sakamoto}, \citenamefont {Morisaki}, \citenamefont {Haruna}, \citenamefont {Itou}, \citenamefont {Fujii},\ and\ \citenamefont {Mitarai}}]{sakamoto2023end}%
  \BibitemOpen
  \bibfield  {author} {\bibinfo {author} {\bibfnamefont {K.}~\bibnamefont {Sakamoto}}, \bibinfo {author} {\bibfnamefont {H.}~\bibnamefont {Morisaki}}, \bibinfo {author} {\bibfnamefont {J.}~\bibnamefont {Haruna}}, \bibinfo {author} {\bibfnamefont {E.}~\bibnamefont {Itou}}, \bibinfo {author} {\bibfnamefont {K.}~\bibnamefont {Fujii}}, \ and\ \bibinfo {author} {\bibfnamefont {K.}~\bibnamefont {Mitarai}},\ }\href@noop {} {\bibfield  {journal} {\bibinfo  {journal} {arXiv preprint arXiv:2311.17388}\ } (\bibinfo {year} {2023})}\BibitemShut {NoStop}%
\bibitem [{\citenamefont {Watson}\ \emph {et~al.}(2023)\citenamefont {Watson}, \citenamefont {Bringewatt}, \citenamefont {Shaw}, \citenamefont {Childs}, \citenamefont {Gorshkov},\ and\ \citenamefont {Davoudi}}]{Watson:2023oov}%
  \BibitemOpen
  \bibfield  {author} {\bibinfo {author} {\bibfnamefont {J.~D.}\ \bibnamefont {Watson}}, \bibinfo {author} {\bibfnamefont {J.}~\bibnamefont {Bringewatt}}, \bibinfo {author} {\bibfnamefont {A.~F.}\ \bibnamefont {Shaw}}, \bibinfo {author} {\bibfnamefont {A.~M.}\ \bibnamefont {Childs}}, \bibinfo {author} {\bibfnamefont {A.~V.}\ \bibnamefont {Gorshkov}}, \ and\ \bibinfo {author} {\bibfnamefont {Z.}~\bibnamefont {Davoudi}},\ }\href@noop {} {\  (\bibinfo {year} {2023})},\ \Eprint {http://arxiv.org/abs/2312.05344} {arXiv:2312.05344 [quant-ph]} \BibitemShut {NoStop}%
\bibitem [{\citenamefont {Carena}\ \emph {et~al.}(2021)\citenamefont {Carena}, \citenamefont {Lamm}, \citenamefont {Li},\ and\ \citenamefont {Liu}}]{Carena:2021ltu}%
  \BibitemOpen
  \bibfield  {author} {\bibinfo {author} {\bibfnamefont {M.}~\bibnamefont {Carena}}, \bibinfo {author} {\bibfnamefont {H.}~\bibnamefont {Lamm}}, \bibinfo {author} {\bibfnamefont {Y.-Y.}\ \bibnamefont {Li}}, \ and\ \bibinfo {author} {\bibfnamefont {W.}~\bibnamefont {Liu}},\ }\href {\doibase 10.1103/PhysRevD.104.094519} {\bibfield  {journal} {\bibinfo  {journal} {Phys. Rev. D}\ }\textbf {\bibinfo {volume} {104}},\ \bibinfo {pages} {094519} (\bibinfo {year} {2021})},\ \Eprint {http://arxiv.org/abs/2107.01166} {arXiv:2107.01166 [hep-lat]} \BibitemShut {NoStop}%
\end{thebibliography}%

\clearpage
\newpage
\onecolumngrid

\appendix

\section{General qubitization procedure}
\label{app:general_W}

In all the cases considered in this work the quantum walk operator $W_H$ can be constructed from the block encoding $U_H$ using the recipe given in~\cref{eq:Whdefas}.
For completeness, below we provide a general approach from~\cite{Low:2016znh} which allows one to assemble the circuit for the $W_H$ operator using an ancillary qubit and two controlled calls to $U_H$.
This is necessary in cases when one is unable to find an operator $S$ satisfying~\cref{eq:S_cond}. 

The general construction requires adding one qubit $\sket{q}$ to the ancillary Hilbert space, such that ${\cal H}_a \to {\cal  H}_{qa} = {\cal H}_q \otimes {\cal H}_a$.
In this new ancillary Hilbert space we define
\begin{align}
    \sket{g}_{qa} = \sket{+}_q \sket{g}_a
    \,,
\end{align}
and the operator $W_H$ is defined on the Hilbert space ${\cal H}_{qas} = {\cal H}_q \otimes {\cal H}_a \otimes {\cal H}_s$. 
\begin{figure}[h]
\begin{quantikz}
\lstick{$\sket{}_q$}    & \gate{\Hgate}          
& \octrl{1} \gategroup[3,steps=2,style={dashed, rounded corners,inner sep=6pt},label style={label position=below,anchor=north}]{$\vphantom{\biggl(}U'_{qas}$}        & \ctrl{1}  & \gate{\Xgate} & \gate[2]{R_{qa}} & \gate{\Hgate} & \qw \\
\lstick{$\sket{}_a$}  & \qw \qwbundle{}
& \gate[2]{U_H^{\phantom{\dagger}}} & \gate[2]{U_H^\dagger} & \qw & \qw & \qw & \qw\\
\lstick{$\sket{}_s$} &\qw\qwbundle{}  
&  &  & \qw     &\qw & \qw & \qw 
\end{quantikz}
\caption{
Circuit for the qubitized block encoding $W_H$, given an arbitrary block encoding $U_H$. The circuit for $R_{qa}$ is shown in~\cref{fig:R}.
}
\label{fig:BEW}
\end{figure}
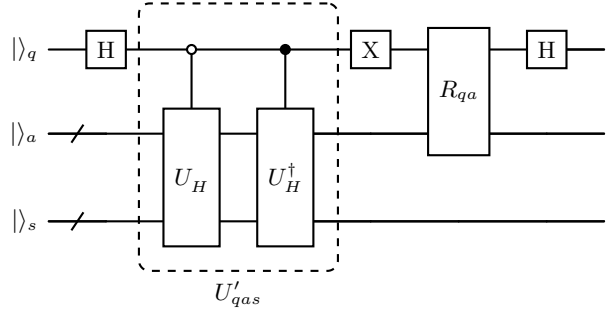
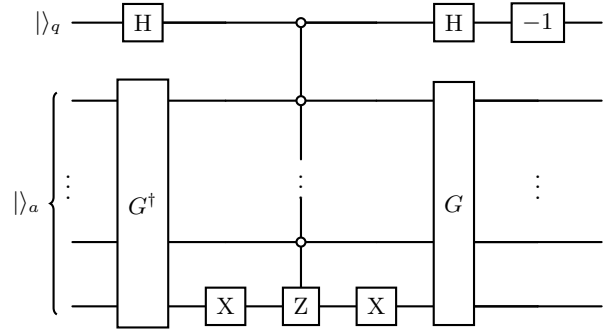
\begin{figure}[h]
\begin{quantikz}
\lstick{$\sket{}_q$} & \gate{\Hgate} & \qw & \octrl{1} & \qw & \gate{\Hgate} & \gate{-1} &\\
\lstick[4]{$\sket{}_a$} & \gate[4]{G^\dagger} & & \octrl{1} & & \gate[4]{G} & \qw & \\
\vdots \ \setwiretype{n} &  & & \ \vdots \ & & & \vdots & \\
& & & \octrl{1} \wire[u][1]{q} & & & \qw & \\
& & \gate{\Xgate} & \gate{\Zgate} & \gate{\Xgate} & & \qw &
\end{quantikz}
\caption{
Circuit for the reflection operator $R_{qa}$ defined in~\cref{eq:refl} and used in~\cref{eq:Whdefqas}, see~\cite{Lin:2022vrd} for alternative implementations. $G$ is the oracle for preparing the state $\sket{g}_a$, i.e. ${G\sket{0}_a=\sket{g}_a}$.
The $\boxed{-1}$ gate adds an overall phase to the circuit (negative sign), and can be implemented, e.g., as $\Xgate\Zgate\Xgate\Zgate$.
}
\label{fig:R}
\end{figure}
As was shown in Section 4 of~\cite{Low:2016znh}, it is given by 
\begin{align}
\label{eq:Whdefqas}
\begin{aligned}
W_H =&  (\Hgate_q\otimes\unit_{as})
(R_{qa}\otimes\unit_s)
(\Xgate_q\otimes\unit_{as})
\\
&
\times
U'_{qas}
(\Hgate_q\otimes\unit_{as}) \,,
\end{aligned}
\end{align}
where $\Hgate_q$ and $\Xgate_q$ are the Hadamard and X gate on the qubit $\sket{q}$, $R_{qa}$ is a standard reflection operator
\begin{align}
    \label{eq:refl}
    R_{qa} = 2 \sket{g}_a \sket{+}_q \sbra{+}_q \sbra{g}_a - \unit_{qa}
    \,,
\end{align}
and
\begin{align}
    U'_{qas} = \sket{0}_q \sbra{0}_q U_H + \sket{1}_q \sbra{1}_q U_H^\dagger
    \,,
\end{align}
is an operator that applies the $U_H$ or $U_H^\dagger$ operator, depending on the state of $\sket{q}$. 

One can show through explicit computation that
\begin{align}
    W_H \sket{g}_{qa} \sket{\lambda}_s &= \left( \lambda\, \sket{g}_{qa} - \sqrt{1-\lambda^2} \sket{\bot^\lambda}_{qa} \right)\sket{\lambda}_s
    \,,
\end{align}
with
\begin{align}
    \sket{\bot^\lambda}_{qa} = \sket{-}_q \sket{\bot^\lambda}_a
    \,,
\end{align}
verifying that \cref{eq:Whdef} gives indeed a qubitized BE of $H$. 
The circuits for operators $W_H$, $U'_{qas}$, and $R_{qa}$ are shown in~\cref{fig:BEW,fig:R}.

\end{document}